\documentclass[11pt,a4paper]{article}
\usepackage[margin=26mm,headheight=15pt,headsep=8mm]{geometry}
\usepackage{amsmath,amsthm}
\usepackage{newpxtext,newpxmath}
\usepackage[round,authoryear]{natbib}
\usepackage{amsmath,amsthm,booktabs,array,graphicx,adjustbox,etoolbox,longtable,float,import}
\usepackage{microtype,fancyhdr,caption,pgfplots}
\pgfplotsset{compat=1.18}
\usepackage[colorlinks=true,linkcolor=blue!45!black,citecolor=blue!45!black,urlcolor=blue!55!black]{hyperref}
\usepackage{xcolor}
\usepackage{bookmark}
\newtheorem{theorem}{Theorem}[section]
\newtheorem{proposition}[theorem]{Proposition}
\newtheorem{corollary}[theorem]{Corollary}

\numberwithin{equation}{section}
\graphicspath{{figures/}{../figures/}}
\AtBeginEnvironment{table}{\linespread{1}\selectfont}
\AtBeginEnvironment{longtable}{\linespread{1}\selectfont}
\title{\LARGE\bfseries Covariance-Calibrated Prediction-Powered Inference}
\author{Siyuan Chen$^1$\quad Yang Bai$^1$\quad Huichen Zhu$^2$\\[5pt]
\small $^1$Shanghai University of Finance and Economics\\
\small $^2$The Chinese University of Hong Kong}
\date{}
\begin{document}
\maketitle\thispagestyle{empty}
\begin{abstract}
In prediction-powered inference, predictions supplement labeled outcomes to improve estimation. We combine a prediction-corrected score and an auxiliary moment in a generalized method of moments (GMM) criterion. The resulting estimator is first-order equivalent to the corresponding corrected-score root under regularity conditions, but the auxiliary penalty can change the finite-sample fit. In a multiparameter Gaussian linear model with low noise, retaining this penalty reduces risk relative to bounded scalar corrected-score roots. We then introduce covariance-calibrated prediction-powered inference (CC-PPI), which uses prediction scores from both samples to reweight the penalty. With a spectrally protected half-strength update, CC-PPI retains regular first-order efficiency and strictly reduces risk relative to uncalibrated full GMM in a Gaussian linear low-noise regime that tolerates some prediction error and in a structured multiparameter logistic model conditional on the existence of a finite pilot estimate. Targeted simulations isolate the gains from calibration and show that they can reverse as prediction error grows. Broader simulations and four applications evaluate point estimation and the coverage and width of full-refit bootstrap intervals for linear projections.

\end{abstract}
\section{Introduction}
\label{sec:intro}\label{sec:related}

Predictions are often available for far more observations than measured outcomes. Prediction-powered inference (PPI) uses these predictions to improve estimation, with a correction based on labeled outcomes that protects the target against prediction error \citep{angelopoulos2023ppi}. How much the predictions help depends on how they are combined with the outcomes. PPI++ estimates a scalar correction coefficient \citep{angelopoulos2024ppiplus}, PSPA allows coordinate-specific coefficients \citep{miao2025pspa}, and prediction de-correlated inference uses a full matrix \citep{gan2024prediction}. These choices are usually understood through their first-order variances. With few labels, the estimated covariance weight and the empirical auxiliary derivative also affect the fitted parameter.

The complete generalized method of moments criterion, GMM-PPI, makes this finite-sample effect explicit. Completing its square gives a corrected-score term and an auxiliary-moment penalty. The auxiliary moment has zero population derivative, but varies with the parameter in a realized sample. In linear regression, its derivative is the difference between the labeled and unlabeled empirical design matrices. Retaining this term can change the fit without changing its regular first-order limit.

We propose covariance-calibrated PPI (CC-PPI). Its construction uses the prediction scores observed in both samples: it replaces the labeled contribution to auxiliary covariance by its unlabeled counterpart and preserves the covariance and cross covariance of the corrected score. In the resulting criterion, calibration changes the weight on the auxiliary residual. A spectral bound protects positive definiteness, and a fixed half-strength update gives a fit without response-based strength selection. The method uses the same predictions and moments as uncalibrated GMM-PPI; the covariance metric is what changes.

Optimal moment combination provides the first-order benchmark \citep{hansen1982large,chen2008semiparametric,lynch2013using,byun2025synthetic}. Estimated weights and moment--Jacobian dependence are established sources of finite-sample GMM effects \citep{altonji1996small,neweysmith2003,windmeijer2005}. The closest connection is auxiliary-moment shrinkage: \citet{okui2009shrinkage} downweights an orthogonalized supplemental-moment criterion, whereas CC-PPI calibrates the covariance blocks of the prediction-assisted criterion.

Another line of work improves the auxiliary score, including RePPI \citep{ji2025predictions} and surrogate-powered inference \citep{chen2025surrogate}. Here the supplied predictor is fixed; the adjustment concerns how its information is weighted. This also differs from shrinkage across mean-estimation tasks \citep{li2025pas}.

We prove that uncalibrated full GMM strictly improves on every data-selected bounded scalar corrected-score root in a structured multiparameter Gaussian linear low-noise regime, with a common prediction-error neighborhood. For covariance calibration, an exact risk decomposition retains the dependence among the pilot, estimated weight and final fit. We establish further risk improvement in a multiparameter Gaussian linear model and exact conditional improvement in a structured multiparameter logistic model when the finite pilot exists. Both calibrated and uncalibrated fits retain regular first-order efficiency.

Simulations and four applications examine the gains and their limits. The experiments under the theory designs isolate calibration from the complete-criterion effect and show that mismatch can reverse its benefit (Figure~\ref{fig:calibration-mechanism}). Independent fixed-strength comparisons assess the half-strength default. Full-refit bootstrap intervals support linear-projection inference, with coverage evaluated alongside length. Sections~\ref{sec:setup}--\ref{sec:model-risk} develop the method and theory, followed by simulations and applications in Sections~\ref{sec:simulation} and~\ref{sec:real-data}.

\section{Covariance-Calibrated Prediction-Powered Inference}
\label{sec:setup}

\subsection{Data and moments}
\label{subsec:data-moments}
We observe independent iid samples $\{(X_i,Y_i)\}_{i=1}^n$ and $\{U_k\}_{k=1}^N$ with the same covariate distribution. The predictor $f$ is fixed or trained independently of these samples. The parameter $\theta_*\in\mathbb R^d$ satisfies $E\{\psi(X,Y;\theta_*)\}=0$. Let $h(x;\theta)$ be a supplied prediction score, taking values in $\mathbb R^d$; when the outcome score accepts a predicted response, we use $h(x;\theta)=\psi(x,f(x);\theta)$. Define
\begin{equation}\begin{aligned}
 m_1(\theta)&=\frac1n\sum_{i=1}^n\psi(X_i,Y_i;\theta),\\
 m_2(\theta)&=\frac1N\sum_{k=1}^N h(U_k;\theta)-\frac1n\sum_{i=1}^n h(X_i;\theta).
\end{aligned}\label{eq:m2}\end{equation}
The stacked moment is $\widehat g(\theta)=(m_1(\theta)^\top,m_2(\theta)^\top)^\top$. Since $E\{m_2(\theta)\}=0$ for every fixed $\theta$, misspecification of $f$ does not change the population identifying equation. The auxiliary moment contributes information through its correlation with the outcome moment.

\subsection{The complete GMM criterion}
For a positive definite weight $W$, GMM-PPI minimizes $Q_W(\theta)=\widehat g(\theta)^\top W\widehat g(\theta)$. The weight is estimated at a supervised pilot and held fixed during minimization. Partition $W$ into four $d\times d$ blocks, with
$\Lambda_W=W_{11}^{-1}W_{12}$ and $\Pi_W=W_{22}-W_{21}W_{11}^{-1}W_{12}\succ0$.
Completing the square gives
\begin{equation}
 Q_W(\theta)=(m_1+\Lambda_Wm_2)^\top W_{11}(m_1+\Lambda_Wm_2)
                  +m_2^\top\Pi_Wm_2.
 \label{eq:fs-anchor-decomposition}
\end{equation}
The corrected-score equation uses the first term. The complete fit also minimizes the auxiliary-moment penalty. For linear regression, the derivative of $m_2$ is $U^\top U/N-X^\top X/n$, so this penalty generally changes the estimate. It disappears from the population derivative under a common covariate law, but remains in the sample criterion. This observation motivates working with the full criterion when comparing finite-sample fits.

For a scalar regression without an intercept, let $f(x)=\alpha x$, $a_L=n^{-1}\sum_iX_i^2>0$, $a_U=N^{-1}\sum_kU_k^2$, and $b_X=a_U-a_L$. With supervised OLS slope $t$, the moments are $m_1(\beta)=a_L(\beta-t)$ and $m_2(\beta)=b_X(\beta-\alpha)$. For $c_W=a_L+\Lambda_Wb_X$, minimization gives
\begin{equation*}
 \widehat\beta(W)=\alpha+
 \frac{W_{11}c_Wa_L}{W_{11}c_W^2+\Pi_Wb_X^2}(t-\alpha).
\end{equation*}
The term $\Pi_Wb_X^2$ is the auxiliary penalty. If the two empirical design second moments agree, then $b_X=0$ and every positive definite weight gives $t$. Otherwise the fit depends on both $\Lambda_W$ and $\Pi_W$. The denominator is positive because the stacked design has full column rank. Supplementary Section~\ref{supp:worked-linear} gives the derivation.

\subsection{Covariance calibration and computation}
\label{subsec:ccppi}\label{subsec:implementation}
At the supervised pilot $t$, evaluate $s_i=\psi(X_i,Y_i;t)$, $h_i=h(X_i;t)$ and $h_k^U=h(U_k;t)$. All empirical covariances are centered and use their sample-size divisor. Write $C_L=\widehat{\operatorname{Cov}}_L(h,h)$, $C_U=\widehat{\operatorname{Cov}}_U(h^U,h^U)$ and $\rho=n/N$. The uncalibrated fit uses the inverse of
\begin{equation*}
 \mathsf S=\begin{pmatrix}
 \widehat{\operatorname{Cov}}_L(s,s)&-\widehat{\operatorname{Cov}}_L(s,h)\\
 -\widehat{\operatorname{Cov}}_L(s,h)^\top&C_L+\rho C_U
 \end{pmatrix}+\eta I_{2d},\qquad \eta>0.
\end{equation*}
The ridge makes this sample covariance invertible; experiments use $\eta=n^{-2}$.

The prediction score is observed in both samples. Its covariance therefore has two estimates, $C_L$ and $C_U$, whereas the outcome-score covariance can only be estimated from labeled data. CC-PPI uses $C_U$ in place of the labeled contribution $C_L$ to the auxiliary block. We make this replacement while preserving the covariance and cross covariance of the corrected moment $m_1+m_2$, so that calibration acts on the auxiliary residual penalty. These choices determine the raw update
\begin{equation}
 \Delta=\begin{pmatrix}
 C_U-C_L&C_L-C_U\\C_L-C_U&C_U-C_L
 \end{pmatrix}.
 \label{eq:ccppi-raw-update}
\end{equation}

In corrected-moment coordinates, this update preserves the covariance and cross covariance of $m_1+m_2$ and changes only the auxiliary residual metric. Supplementary Section~\ref{subsec:calibrated-target} gives the block characterization and its proof.

The raw target need not be positive definite. We bound its relative change by clipping the eigenvalues of $E=\mathsf S^{-1/2}\Delta\mathsf S^{-1/2}$ to $[-\kappa,\kappa]$, with $\kappa=1/2$, and set
\begin{equation}
 \Delta_p=\mathsf S^{1/2}\operatorname{clip}_{[-\kappa,\kappa]}(E)\mathsf S^{1/2}.
 \label{eq:ccppi-spectral-protection}
\end{equation}
For $0\le\lambda\le1$, let $\mathsf S_\lambda=\mathsf S+\lambda\Delta_p$. The half-strength fit is
\begin{equation}
 \widehat\theta_\lambda\in\arg\min_{\theta\in\Theta}
       \widehat g(\theta)^\top\mathsf S_\lambda^{-1}\widehat g(\theta),
 \qquad \widehat\theta_{\rm CC}=\widehat\theta_{1/2}.
 \label{eq:fs-calibrated-member}
\end{equation}
Half strength gives $\tfrac34\mathsf S\preceq\mathsf S_{1/2}\preceq\tfrac54\mathsf S$. It provides a bounded adjustment to the original weight; Section~\ref{subsec:new-default-experiment} evaluates the remaining gain from choosing another fixed strength.

For unconstrained affine moments $\widehat g(\theta)=H\theta-b$ with full-column-rank $H$, the fitted parameter is
\begin{equation}\widehat\theta_\lambda=(H^\top\mathsf S_\lambda^{-1}H)^{-1}H^\top\mathsf S_\lambda^{-1}b.\label{eq:affine-solution}\end{equation}

For linear regression, use $\psi(x,y;\beta)=x(x^\top\beta-y)$ and $h(x;\beta)=x(x^\top\beta-f(x))$. The quadratic criterion has a unique unconstrained solution when $X$ has full column rank, given by \eqref{eq:affine-solution}. Logistic regression uses $\psi(x,y;\theta)=x\{\operatorname{expit}(x^\top\theta)-y\}$, with supplied probabilities in the auxiliary score, and requires nonlinear minimization.

Computation consists of fitting the pilot, forming the centered covariance, protecting the update and minimizing the frozen criterion. The predictor and weight stay fixed during the last step. We initialize nonlinear fits at the supervised estimate. Supplementary Section~\ref{supp:method} gives the numerical conventions, constrained-fit requirements and equivalent low-dimensional protection calculation.

\section{Asymptotic Equivalence and Finite-Sample Risk}
\label{sec:theory}
A common first-order limit does not determine finite-sample risk. We examine the effects of retaining the complete criterion and calibrating its covariance weight.

\subsection{A common first-order limit}
\label{sec:asymptotic}
Let $A=\partial_\theta E\psi(X,Y;\theta_*)$ be nonsingular, put $G=(A^\top,0^\top)^\top$, and suppose $\sqrt n\,\widehat g(\theta_*)$ has positive definite limiting covariance $\Omega$. We work at fixed dimension with $n/N\to\rho\in[0,\infty)$. Supplementary Section~\ref{subsec:regularity} states the identification, uniform convergence, pilot, covariance and minimization conditions; the ridge vanishes in this limit. The pilot and weight may use the same inference samples. Numerical returns must meet the stated objective and stationarity accuracy, with sufficient logistic conditions given in Supplementary Section~\ref{supp:logistic-solver-equivalence}.

\begin{samepage}
\begin{theorem}[Efficiency of covariance calibration]\label{thm:new-efficiency}
Under these conditions and the independent-sample moment CLT, $\widehat\theta_{\rm CC}\to_p\theta_*$ and
\begin{equation}\sqrt n(\widehat\theta_{\rm CC}-\theta_*)
 \Rightarrow N\{0,(G^\top\Omega^{-1}G)^{-1}\}.\label{eq:new-efficient-limit}\end{equation}
The same limit holds for uncalibrated inverse-covariance GMM, and $\sqrt n(\widehat\theta_{\rm CC}-\widehat\theta_0)=o_p(1)$.
\end{theorem}
\end{samepage}
Because $C_U-C_L\to_p0$, calibration vanishes in this regular limit. The covariance in \eqref{eq:new-efficient-limit} is the GMM efficiency bound for the supplied moments \citep{hansen1982large}. Full-matrix correction attains it; scalar and diagonal corrections do so when they can represent the unrestricted optimal correction. Supplementary Section~\ref{subsec:loewner} gives the covariance gap and equality conditions. These comparisons keep the auxiliary score fixed.

The first-order comparison can be expressed directly in terms of a correction matrix. Partition the limiting moment covariance as $\Omega=[V,-R;-R^\top,D]$ and let $\mathcal V(L)$ denote the parameter covariance induced by the corrected moment $m_1+Lm_2$. The unrestricted optimum is $L_*=RD^{-1}$. Completing the covariance square gives
\begin{equation*}
 \mathcal V(L)-\mathcal V(L_*)
 =A^{-1}(L-L_*)D(L-L_*)^\top A^{-\top}\succeq0.
\end{equation*}
A scalar correction restricts $L$ to $\ell I_d$; a coordinatewise correction restricts it to a diagonal matrix. Such restrictions attain the full-matrix bound when they can represent $L_*$. Calibration addresses a further choice: how to weight the realized sample criterion among fits sharing that efficient limit. The following example shows how much that choice can matter conditionally on a finite design. The covariance identity and its equality conditions are proved in Supplementary Section~\ref{subsec:loewner}.

The common limit can conceal substantial conditional risk differences. For example, consider scalar linear regression with predictor $\alpha x$, mismatch $\mu=\beta-\alpha\ne0$, Gaussian error variance $\sigma^2$, and weights $W_\tau=\operatorname{diag}(1,\tau)$. Under the iid law and finite designs specified in Proposition~\ref{prop:finite-separation}, these fits share the efficient first-order limit, yet their conditional risk ratios to supervised OLS are
\begin{equation*}
 \frac{R_X(G_1)}{R_X(t)}=.64+.04\frac{n\mu^2}{\sigma^2},\qquad
 \frac{R_X(G_4)}{R_X(t)}=.25+.25\frac{n\mu^2}{\sigma^2}.
\end{equation*}
At $\mu=.1$ and $\sigma=1$, the ratios are $.656,.350$ for $n=40$, but $2.24,10.25$ for $n=4000$. The ranking reverses as conditional bias becomes more important relative to variance. The conditioning events become rare as $n$ grows; this example establishes conditional separation within the weighting family, not unconditional risk separation for CC-PPI. The full construction and proof are in Supplementary Section~\ref{supp:finite-separation}.

\subsection{Finite-sample advantage of the complete criterion}
\label{subsec:full-gmm-advantage}
Finite-sample separation does not by itself establish a risk advantage. We next compare the uncalibrated complete fit with the entire class of roots $m_1+\widehat\ell m_2=0$, where $\widehat\ell(Y)\in[0,1]$ is any measurable rule. This includes linear PPI++ implemented as a bounded scalar corrected-score root. The correction coefficient $\ell$ is distinct from the covariance-calibration strength.

\begin{theorem}[Full GMM versus bounded scalar correction]\label{thm:full-gmm-scalar}
Fix $d\ge2$, sample sizes $n,N$, and designs satisfying $X^\top X/n=aI_d$, $U^\top U/N=uI_d$, with $0<u<a$ and no additional intercept. Let $Y=X\beta+\sigma Z$, $Z\sim N(0,I_n)$, and $f(x)=x^\top\alpha$. Use the OLS pilot and the uncalibrated full GMM fit $\widehat\beta_G=\widehat\theta_0$, with the centered covariance of Section~\ref{subsec:implementation} and fixed ridge $\eta=n^{-2}$. Denote the corrected root by $\widehat\beta^P_{\widehat\ell}$. At $\beta=\alpha$, with $v_0=d/(na)$,
\begin{align}
 \inf_{\widehat\ell(Y)\in[0,1]}E\|\widehat\beta^P_{\widehat\ell}-\beta\|^2
 &=v_0\sigma^2,\label{eq:full-gmm-scalar-floor}\\
 E\|\widehat\beta_G-\beta\|^2&=k^2v_0\sigma^2+O(\sigma^4),\qquad
 k=\frac{a^2}{a^2+(u-a)^2}<1.\label{eq:full-gmm-scalar-risk}
\end{align}
There exists $\sigma_0>0$ such that, for every $0<\sigma<\sigma_0$, a common $\delta_\sigma>0$ satisfies
\begin{equation}\label{eq:full-gmm-scalar-dominance}
 E_\beta\|\widehat\beta_G-\beta\|^2
 <\inf_{\widehat\ell(Y)\in[0,1]}E_\beta\|\widehat\beta^P_{\widehat\ell}-\beta\|^2
 \quad\text{if }\|\beta-\alpha\|<\delta_\sigma.
\end{equation}
\end{theorem}

At matched prediction, every bounded correction gives $\widehat\beta^P_\ell-\beta=c_\ell(t-\beta)$ with $c_\ell=a/\{a-\ell(a-u)\}\ge1$, even for response-selected $\ell$. The complete fit instead has leading coefficient $k<1$. The theorem therefore compares with the best bounded coefficient, not only $\ell=1$. It conditions on structured designs and keeps sample sizes and ridge fixed as noise decreases; it does not establish unconditional random-design dominance or cover unrestricted coefficients. Supplementary Section~\ref{supp:full-gmm-scalar} gives the full proof, an explicit conservative noise threshold and the common mismatch neighborhood.

\subsection{Calibration risk decomposition and strength paths}
\label{sec:finite-sample}\label{subsec:finite-risk-selection}
We next compare the calibrated fit with uncalibrated full GMM.

For the finite-sample calculation, condition on a full-column-rank labeled design $X$, a supplied unlabeled design $U$ and fixed finite predictions, and use unconstrained linear fits. Let $Y=X\beta+\epsilon$, with $E\epsilon=0$, $E\|\epsilon\|^4<\infty$ and positive definite covariance of the OLS pilot $t$. Write $R(\lambda)=E\|\widehat\beta_\lambda-\beta\|^2$ and $h=1/2$. A positive ridge and bounded relative protection ensure square-integrable fits. Errors may be heteroskedastic or dependent.

Project each fitted parameter onto affine functions of $t$ in $L^2$, and denote the projection by $\widetilde\beta_\lambda$. The projected risk difference is $\mathcal D=E\|\widetilde\beta_h-\beta\|^2-E\|\widetilde\beta_0-\beta\|^2$. Let $g_\perp=\widehat\beta_0-\widetilde\beta_0$ be the baseline projection residual and $u_\perp=(\widehat\beta_h-\widetilde\beta_h)-g_\perp$ its change under calibration. Set $Q_g=E\|g_\perp\|^2$ and $Q_u=E\|u_\perp\|^2$.
\begin{theorem}[Signed finite-sample transfer]\label{thm:fs-signed-transfer}
In every parameter dimension, for the actual same-data fits,
\begin{align*}
 R(h)-R(0)&=\mathcal D+2E\langle g_\perp,u_\perp\rangle+Q_u,\\
 R(h)-R(0)&\le\mathcal D+2\sqrt{Q_gQ_u}+Q_u.
\end{align*}
Thus $-\mathcal D>2\sqrt{Q_gQ_u}+Q_u$ guarantees $R(h)<R(0)$.
\end{theorem}
The projected difference $\mathcal D$ measures the risk change in the affine response to the pilot. The remaining terms measure the interaction with the baseline residual and the squared cost of moving that residual. No independence between the pilot and the estimated covariance is required.

The decomposition follows from orthogonality of the affine projection and its residual. Its terms depend on the unknown truth. The model calculations below establish the sign of the risk difference; Supplementary Section~\ref{app:risk-proofs} gives the projection coefficients and proof.

Changing strength traces a smooth, generally nonlinear parameter path even for affine moments. Half strength therefore need not give the midpoint of the endpoint estimates. Section~\ref{subsec:new-default-experiment} compares it with fixed strengths trained on independent responses from the same law. The exact path and risk bounds are in Supplementary Section~\ref{app:path-proofs}.

\section{Finite-Sample Improvement in Regression Models}
\label{sec:model-risk}
We now compare CC-PPI with uncalibrated full GMM in linear and logistic regression. At the exact-prediction reference points used below, the auxiliary covariance is singular; the risk calculations therefore proceed directly rather than through the regular limit in Theorem~\ref{thm:new-efficiency}.

\subsection{Multiparameter linear regression}
Consider the fixed-design linear model of Section~\ref{sec:finite-sample}. The next result fixes the sample sizes and ridge, then studies the effect of response noise and prediction error.

\begin{theorem}[Multiparameter linear improvement]\label{thm:linear-multiparameter}
Fix $d\ge2$, $a>0$, and finite designs with $m\ge2$ labeled rows and $M\ge1$ unlabeled rows in each coordinate direction, each row equal to $ae_j$ or $-ae_j$, and all other rows zero. There is no additional intercept. Put $p=m/n$, $u=M/N$, and assume $dp\le1$, $du\le1$, $0<u<p$, and $(d+2)(p+u)<3$. Let $Y=X\beta+\epsilon$, $\epsilon\sim N(0,\sigma^2I_n)$, and $f(x)=x^\top\alpha$. Use the OLS pilot, centered divisors $n,N$, any fixed ridge $\eta>0$, cap $\kappa=1/2$, and the full protected estimator at $h=1/2$.

Write $\mu=\beta-\alpha$ and $R_\mu(\lambda)=E_\mu\|\widehat\beta_\lambda-\beta\|^2$. For each fixed design and ridge satisfying these conditions, there exists $\sigma_0>0$ such that $R_0(h)<R_0(0)$ whenever $0<\sigma<\sigma_0$. For each noise level in this range, strict improvement persists for $\mu$ in a nonempty neighborhood of zero. It also persists under sufficiently small design and prediction perturbations, including nonorthogonal designs.
\end{theorem}

At matched prediction, calibration leaves the leading error term unchanged. Its first effect is cubic in the noise scale; the cross product with the leading error produces a fourth-order risk difference. The conditions $u<p$ and $(d+2)(p+u)<3$ make the leading gain positive. Global centering couples the coordinate groups, so this sign is established in expectation rather than by a samplewise covariance ordering. Supplementary Section~\ref{supp:linear-multiparameter} derives the expansion and bounds its $O(\sigma^6)$ risk remainder, including responses on which clipping is active.

For each favorable noise level, the Gaussian-shift bound in \eqref{eq:multi-mismatch-budget} extends the gain to a neighborhood of matched prediction. Risk continuity gives the neighboring nonorthogonal designs. A scalar regime without the small-noise restriction is given in Supplementary Section~\ref{app:primitive-proof}.

\subsection{Multiparameter logistic regression}
For logistic regression, a grouped design makes the calibration path explicit. Each coordinate estimate has a sign determined by its group count and a magnitude that decreases with calibration strength.

\begin{theorem}[Conditional logistic improvement]\label{thm:logistic-improvement}
Fix $d\ge1$ and $a>0$. The labeled design has three rows $ae_j$ for each coordinate unit vector $e_j$ and zero rows otherwise; the unlabeled design has $M$ rows $ae_j$ per coordinate and zeros otherwise. Let $p=3/n$, $u=M/N$, $0<u<p$, $d(p+u)<1$, and use fixed predictions $f_L=f_U=1/2$. Responses are independent Bernoulli with means $\sigma(x_i^\top\beta)$. Use the unconstrained full estimator ($\Theta=\mathbb R^d$) with positive covariance ridge and the prescribed spectral protection. If $K_j$ is the labeled success count in group $j$, condition on $\mathcal E=\{K_j\in\{1,2\}\text{ for all }j\}$.

There is a deterministic strictly decreasing positive path $b_\lambda$, given in \eqref{eq:logistic-exact-path}, such that the unique finite fit has coordinates $\widehat\theta_{\lambda,j}=(2K_j-3)b_\lambda$. For $h=1/2$ and $R^c_\beta(\lambda)=E_\beta[\|\widehat\theta_\lambda-\beta\|^2\mid\mathcal E]$,
\begin{equation*}\begin{split}
 R^c_\beta(0)-R^c_\beta(h)
 =(b_0-b_h)\left\{d(b_0+b_h)-2\sum_j\beta_j\tanh(a\beta_j/2)\right\}.
\end{split}\end{equation*}
In particular, $\|\beta\|^2<d(b_0+b_h)/a$ guarantees strict conditional risk improvement.
\end{theorem}

In probability coordinates $z_j=\sigma(a\theta_j)$, the grouped design makes both moments affine and the criterion quadratic. On $\mathcal E$, calibration preserves the signs $2K_j-3$ and reduces the common magnitude $b_\lambda$. Near $\beta=0$, this reduces variation; as the signal grows, it also increases bias. The exact conditional risk identity gives the balance and allows a nonzero mismatch between the fixed predictions $1/2$ and the true probabilities.

At $\beta=0$, the pilot-existence event has probability $(3/4)^d$. The theorem compares fits conditional on that event. Supplementary Section~\ref{supp:logistic-risk} gives the covariance reduction and path; Sections~\ref{supp:logistic-global-certificate} and~\ref{supp:logistic-solver-equivalence} address numerical accuracy under their own conditions. Fixed production tolerances alone do not establish these certificates.

\subsection{Supporting confidence intervals}
\label{subsec:new-inference}
For linear projections we construct marginal intervals at nominal level $1-\alpha$, $0<\alpha<1$, by bootstrapping the full fit. Each draw recomputes the pilot, covariance and protected weight. Let $\widehat\beta=\widehat\theta_{\rm CC}$, let $a$ be the unlabeled prediction-regression anchor, and let $s_j$ be the observed OLS normalization scale. If $q_l,q_u$ are the conditional $\alpha/2$ and $1-\alpha/2$ quantiles of the centered, studentized refit, define
\begin{equation}\begin{aligned}
 J_j=[L_j,U_j]&=[\widehat\beta_j-q_us_j,\widehat\beta_j-q_ls_j],\\
 I_j&=[L_j-\gamma(L_j-a_j)_+,\ U_j+\gamma(a_j-U_j)_+].
\end{aligned}\label{eq:new-ci}\end{equation}
The first interval is the bootstrap parent. The second expands an endpoint toward the anchor when it falls outside that interval. We use $\gamma=3/4$, fixed before validation. Supplementary Section~\ref{supp:interval} specifies the centered refit, HC3 adjustment \citep{mackinnonwhite1985} and exponential bootstrap weights \citep{rubin1981,praestgaard1993}.

\begin{theorem}[Bootstrap validity for linear projections]\label{thm:hc3-iid}
Suppose the labeled pairs and unlabeled covariates are independent iid samples with the same covariate law, $f$ is fixed or independently trained, $d$ is fixed, and $n/N\to\rho\in(0,\infty)$. Assume $(\|X\|+|Y|+|f(X)|)$ has a finite moment of order greater than eight, $E(XX^\top)\succ0$, and the linear projection $\beta=\{E(XX^\top)\}^{-1}E(XY)$ has positive residual variance. The limiting joint covariance of $\sqrt n\,\widehat g(\beta)$ must be positive definite.
With ideal conditional quantiles, $P\{\beta_j\in J_j\}\to1-\alpha$ and $\liminf P\{\beta_j\in I_j\}\ge1-\alpha$ for each fixed coordinate $j$.
\end{theorem}

The theorem allows heteroskedasticity and a nonlinear response mean. Expansion preserves every parent inclusion and adds $\gamma\operatorname{dist}(a_j,J_j)$ to its length. Under the same regular conditions, if the prediction-regression anchor has probability limit $a_{0,j}\ne\beta_j$, then $|I_j|\to_p\gamma|a_{0,j}-\beta_j|$: the expanded interval need not shrink as the sample size grows. Coverage gains must therefore be assessed together with interval length. The proof, the width limit and the finite-$B$ quantile calculation for our $B=399$ implementation are in Supplementary Section~\ref{supp:interval}. Sections~\ref{subsec:new-interval-experiment} and~\ref{subsec:new-real-interval} examine interval performance under their respective evaluation protocols.

\section{Simulation Studies}
\label{sec:simulation}

\subsection{Settings and comparisons}
\label{subsec:sim-setup}
For linear regression, we generate $X\sim N(0,\Sigma_X)$ with $(\Sigma_X)_{ij}=0.6^{|i-j|}$ and $Y=X^\top\beta+\epsilon$, where $\epsilon\sim N(0,1)$ independently and $\|\beta\|=1$. The predictor is $f(x)=x^\top(\beta+r\mathbf1/\sqrt{\mathbf1^\top\Sigma_X\mathbf1})$, so its mean squared discrepancy from $X^\top\beta$ is $r^2$. We use $d\in\{2,5,10,20\}$, $r\in\{.2,.5,1\}$, $N=5n$ and labeled sizes from 50 to 10,000. Each of two fixed design pairs has 1,000 response repetitions. Logistic simulations use the same covariate distribution, an expit mean and predictor, and 200 responses per design. Full grids and slope choices are in Supplementary Section~\ref{supp:simulation}.

We compare CC-PPI with supervised estimation, PPI, PPI++, PSPA, PDC and RePPI in the general benchmarks. Methods share covariates, predictions and response draws. Uncalibrated GMM appears only in the focused calibration comparison below. Point-estimation entries are ratios of mean squared parameter errors to the common supervised MSE, multiplied by 100. Interval coverage and length average over trials and coordinates. Supplementary Section~\ref{supp:simulation} gives protocols and diagnostics.

\subsection{Point estimation and calibration}
\label{subsec:new-point-results}\label{subsec:new-default-experiment}
The linear results show a strong dimension and sample-size dependence (Table~\ref{tab:main-linear-grid}). At $d=20,n=50$, CC-PPI has 53.3\% of supervised MSE. At $d=2$, however, its ratio exceeds 100 by $n=500$. The small-label gains therefore do not amount to uniform improvement over supervised estimation.
\begin{table}[H]\centering\small
\caption{Linear MSE ratios to the full paired supervised MSE (percent), $r=.2$, averaged across two fixed designs. Supervised is 100 by construction.}
\label{tab:main-linear-grid}
\begin{adjustbox}{max width=\linewidth}\begin{tabular}{rrrrrrrr}\toprule
$d$&$n$&$n/d$&Sup.&PPI++&PSPA&PDC&CC\\\midrule
2&50&25&100.0&96.4&$\ge10^3$&113.5&51.5\\
2&100&50&100.0&98.6&266.4&107.5&83.3\\
2&500&250&100.0&99.8&100.8&100.2&104.4\\
2&2000&1000&100.0&100.1&100.0&100.0&100.5\\
5&50&10&100.0&100.3&107.4&105.1&73.4\\
5&100&20&100.0&100.2&104.5&106.5&67.4\\
5&500&100&100.0&100.0&100.7&100.9&81.4\\
5&2000&400&100.0&100.0&100.4&100.4&96.5\\
10&50&5&100.0&98.0&116.0&136.4&57.2\\
10&100&10&100.0&99.4&106.1&112.0&62.6\\
10&500&50&100.0&99.8&101.1&102.8&74.7\\
10&2000&200&100.0&99.9&100.1&100.4&89.5\\
20&50&2.5&100.0&98.4&119.8&257.5&53.3\\
20&100&5&100.0&98.9&105.7&134.2&52.8\\
20&500&25&100.0&99.7&100.9&104.0&65.2\\
20&2000&100&100.0&100.0&100.3&101.1&85.8\\
\bottomrule\end{tabular}\end{adjustbox}
\end{table}

\begin{table}[H]\centering
\caption{Logistic point-estimation MSE ratios to supervised estimation (\%), $r=.2$. Ratios use the full paired supervised denominator.}\label{tab:new-point}
\begingroup\small
\begin{adjustbox}{max width=\linewidth}\begin{tabular}{lrrrrrrr}
\hline
Model & $d$ & $n$ & Sup. & PPI++ & PSPA & PDC & CC \\
\hline
Logistic & 2 & 100 & 100.0 & 101.7 & 100.7 & 100.5 & 86.8 \\
Logistic & 5 & 100 & 100.0 & 99.5 & 103.8 & 105.6 & 52.0 \\
Logistic & 10 & 100 & 100.0 & 99.2 & 477.4 & 109.4 & 52.4 \\
Logistic & 20 & 500 & 100.0 & 99.7 & 101.0 & 104.2 & 52.6 \\
\hline\end{tabular}\end{adjustbox}
\endgroup

\end{table}

The logistic results also show gains for nonlinear moments (Table~\ref{tab:new-point}, $r=.2$). At $d=10,n=100$, CC-PPI has 52.4\% of supervised MSE, compared with 99.2\% for PPI++. These comparisons concern the specified numerical implementations; smaller-size solver diagnostics are reported in the supplement.

\par\medskip\noindent\textbf{Calibration under the regression theory designs}\par\nobreak\smallskip
Under the designs of Theorems~\ref{thm:linear-multiparameter} and~\ref{thm:logistic-improvement}, we compare half-strength CC-PPI with uncalibrated GMM to isolate the covariance update. Figure~\ref{fig:calibration-mechanism} reports $100\{R(0)-R(1/2)\}/R(0)$, so positive values favor calibration. The linear comparison uses 10,000 paired Gaussian responses per cell, with sample sizes and ridge fixed while noise and prediction mismatch vary.

At $\sigma=\sqrt\eta$ and exact prediction, linear gains are 0.562\% and 0.335\% for the two designs. They persist at small mismatch and reverse at sufficiently large mismatch. The logistic panel gives exact risks conditional on finite pilot existence, with gains of 15.12\% and 15.43\% at $\beta=0$. Designs, paired uncertainty and all risk values are in Supplementary Section~\ref{supp:calibration-mechanism}.
\begin{figure}[htbp]\centering
\noindent
\begin{minipage}{.32\linewidth}\centering
\begin{tikzpicture}
\begin{axis}[width=\linewidth,height=5cm,scale only axis=false,
title={Linear: noise},xlabel={$q$},ylabel={Gain (\%)},
title style={font=\footnotesize},label style={font=\footnotesize},tick label style={font=\scriptsize},
grid=major,grid style={gray!15},axis line style={gray!65},
legend style={font=\scriptsize,draw=none,at={(.5,1.24)},anchor=south,legend columns=2},
xmode=log,log basis x=10,
]
\addplot[color=blue!70!black,mark=*,mark size=1.5pt] coordinates {(0.05,0.0014780701) (0.1,0.0059099243) (0.2,0.023602089) (0.5,0.14589249) (1,0.56173642) (2,1.8914357) (5,3.9954342)};
\addlegendentry{$d=2$}
\addplot[color=orange!85!black,mark=square*,mark size=1.5pt] coordinates {(0.05,0.00086336144) (0.1,0.0034526074) (0.2,0.013797037) (0.5,0.085651343) (1,0.33463297) (2,1.2197638) (5,3.4914976)};
\addlegendentry{$d=5$}
\end{axis}\end{tikzpicture}\end{minipage}\hfill
\begin{minipage}{.32\linewidth}\centering
\begin{tikzpicture}
\begin{axis}[width=\linewidth,height=5cm,scale only axis=false,
title={Linear: mismatch},xlabel={$c$},ylabel={Gain (\%)},
title style={font=\footnotesize},label style={font=\footnotesize},tick label style={font=\scriptsize},
grid=major,grid style={gray!15},axis line style={gray!65},
legend style={font=\scriptsize,draw=none,at={(.5,1.24)},anchor=south,legend columns=2},
]
\addplot[gray,dashed,forget plot] coordinates {(0,0) (8,0)};
\addplot[color=blue!70!black,mark=*,mark size=1.5pt] coordinates {(0,0.56173642) (0.25,0.55833725) (0.5,0.54923999) (1,0.50675554) (2,0.27400562) (4,-0.90079756) (8,-5.3709727)};
\addplot[color=orange!85!black,mark=square*,mark size=1.5pt] coordinates {(0,0.33463297) (0.25,0.33445398) (0.5,0.33333123) (1,0.32734882) (2,0.29111767) (4,0.07036587) (8,-1.0056843)};
\end{axis}\end{tikzpicture}\end{minipage}\hfill
\begin{minipage}{.32\linewidth}\centering
\begin{tikzpicture}
\begin{axis}[width=\linewidth,height=5cm,scale only axis=false,
title={Logistic: conditional},xlabel={$\|\beta\|$},ylabel={Gain (\%)},
title style={font=\footnotesize},label style={font=\footnotesize},tick label style={font=\scriptsize},
grid=major,grid style={gray!15},axis line style={gray!65},
legend style={font=\scriptsize,draw=none,at={(.5,1.24)},anchor=south,legend columns=2},
]
\addplot[gray,dashed,forget plot] coordinates {(0,0) (4,0)};
\addplot[color=blue!70!black,mark=*,mark size=1.5pt] coordinates {(0,15.116966) (0.2,9.8854521) (0.4,4.2493748) (0.6,1.5514214) (1,-0.38423416) (1.5,-1.0318424) (2,-1.190383) (3,-1.1410304) (4,-0.98705407)};
\addplot[color=orange!85!black,mark=square*,mark size=1.5pt] coordinates {(0,15.428852) (0.2,12.421408) (0.4,7.547015) (0.6,4.1694127) (1,1.0174555) (1.5,-0.33542487) (2,-0.83552592) (3,-1.1225106) (4,-1.1267273)};
\end{axis}\end{tikzpicture}\end{minipage}
\caption{CC-PPI relative risk reduction over uncalibrated GMM. Linear noise varies $q=\sigma/\sqrt\eta$ at exact prediction. Linear mismatch varies $c=\|\beta-\alpha\|/(\sigma/\sqrt6)$ at $q=1$. Logistic risk is conditional on the theorem's pilot-existence event. Negative values retain adverse cases; connecting lines guide the eye.}\label{fig:calibration-mechanism}
\end{figure}

To assess the half-strength choice, we compare it with fixed strengths trained on independent responses. Eight Gaussian linear laws combine $d\in\{2,5\}$, $n\in\{40,160\}$ and $r\in\{.2,.5\}$, with $N=5n$ and 16 designs per law. Reference training and evaluation use separate response streams. Training uses the generating parameter, so these references are risk benchmarks rather than feasible tuning rules. Supplementary Section~\ref{supp:simulation} gives the grids and replication counts.

\begin{table}[htbp]\centering\small
\caption{Remaining improvement along the current calibration path. Risks are percentages of conditional OLS MSE. The gap is relative to half-strength risk.}
\label{tab:fs-path-space}
\begin{adjustbox}{max width=\linewidth}\begin{tabular}{rrr rr r}\toprule
$d$&$n$&$r$&Half&Reference&Gap (\%)\\\midrule
2&40&0.2&65.660&63.975&2.566\\
2&40&0.5&95.692&93.443&2.350\\
2&160&0.2&79.229&77.302&2.431\\
2&160&0.5&94.622&93.392&1.301\\
5&40&0.2&60.398&59.371&1.701\\
5&40&0.5&80.653&78.745&2.366\\
5&160&0.2&65.121&64.243&1.348\\
5&160&0.5&89.165&87.571&1.787\\\bottomrule\end{tabular}\end{adjustbox}
\end{table}

The independently trained reference in Table~\ref{tab:fs-path-space} gives a further 1.30--2.57\% reduction relative to half strength across the eight laws. Thus the default remains close to this fixed-strength benchmark while leaving room for improvement. Figure~\ref{fig:main-calibration-gain} displays the remaining gap; training details and uncertainty are in Supplementary Section~\ref{supp:simulation}.

\subsection{Supporting interval performance}
\label{subsec:new-interval-experiment}
Table~\ref{tab:main-repair} reports the interval experiment at nominal 95\%, with 36,000 responses over 36 fixed-design settings. Each of the 18 rows pools two design seeds. CC Wald and expanded intervals use the same observed CC-PPI point estimate; the linear benchmark is OLS $t$.
\begin{table}[H]\centering
\caption{All pooled configurations at nominal 95\%. CC entries are coordinate-average coverage (\%) / mean length as a percentage of OLS $t$ length; OLS length is 100\%. Each row contains 2,000 responses.}\label{tab:main-repair}
\begingroup\normalsize
\setlength{\tabcolsep}{4pt}
\begin{adjustbox}{max width=\linewidth}\begin{tabular}{rrrll}
\toprule
$n$ & $d$ & OLS $t$ coverage & CC Wald & Expanded\\
\midrule
50&2&94.33&73.78 / 50.6&97.15 / 68.9\\
50&5&95.00&90.31 / 70.2&93.92 / 77.2\\
50&10&95.07&87.62 / 57.5&93.34 / 68.6\\
50&20&94.70&78.23 / 45.4&93.45 / 72.0\\
100&2&95.45&78.45 / 68.6&93.85 / 80.6\\
100&5&95.29&92.40 / 71.3&93.96 / 73.8\\
100&10&95.03&91.49 / 67.3&93.12 / 69.6\\
100&20&94.78&88.73 / 58.4&92.84 / 66.0\\
500&2&94.42&92.73 / 96.6&94.42 / 102.5\\
500&5&95.23&93.58 / 85.8&94.00 / 83.9\\
500&10&94.85&93.51 / 80.8&92.47 / 75.8\\
500&20&94.73&94.02 / 77.0&92.34 / 70.8\\
1000&10&95.03&95.52 / 90.3&95.06 / 85.1\\
2000&2&94.97&94.60 / 99.1&97.00 / 138.0\\
2000&5&95.20&94.31 / 95.0&95.42 / 99.5\\
2000&10&95.28&95.12 / 94.3&95.64 / 92.0\\
2000&20&95.12&94.37 / 89.9&93.85 / 85.0\\
10000&10&95.17&94.96 / 98.6&96.56 / 110.6\\
\bottomrule
\end{tabular}\end{adjustbox}
\endgroup

\end{table}

Expanded coverage ranges from 92.34\% to 97.15\% across the pooled rows, and intervals are shorter than OLS $t$ in 15 of 18 configurations. Expansion is useful where Wald undercovers, but can add unnecessary length elsewhere. The three length disadvantages occur at $(n,d)=(500,2),(2000,2),(10000,10)$. Pooling also masks some design variation: one seed has expanded coverage of 91.42\%, compared with 92.34\% for Wald. Supplementary Section~\ref{supp:simulation} gives the seed and coordinate results.

A separate 4,000-trial validation compares CC Wald, the full-refit bootstrap parent and expansion on common trials, alongside OLS $t$. It includes four iid and four fixed-design configurations. Seven expanded intervals are shorter than OLS $t$; the parent-versus-expansion comparison identifies the contribution of the added width. Full results are in Supplementary Table~\ref{tab:new-interval}.

\section{Real-Data Applications}
\label{sec:real-data}

\subsection{Data and evaluation}
\label{subsec:real-setup}
We study California Housing \citep{pace1997sparse}, Bike Sharing \citep{fanaee2014event}, ACS Income \citep{ding2021retiring} and Wine Quality \citep{cortez2009modeling}. Table~\ref{tab:real-datasets} gives their target dimensions and evaluation pools. Each application combines a low-dimensional regression target with a separately trained 100-tree XGBoost predictor \citep{chen2016xgboost}. Predictor training and evaluation are separated by geography, time or wine type; labeled and unlabeled observations are sampled within the same evaluation frame.
\begin{table}[htbp]\centering\small
\caption{Evaluation-region sizes and target dimensions. $n+N$ denotes the full evaluation pool, before its labeled/unlabeled split.}
\label{tab:real-datasets}
\begin{adjustbox}{max width=\linewidth}\begin{tabular}{lccrl}\toprule
Dataset & task & $d$ & $n+N$ & predictor shift\\\midrule
California Housing \citep{pace1997sparse}&linear&8&10,313&spatial\\
Bike Sharing \citep{fanaee2014event}&linear&6&5,000&temporal\\
ACS Income \citep{ding2021retiring}&logistic&4&50,000&temporal\\
Wine Quality \citep{cortez2009modeling}&linear&12&4,898&red to white\\
\bottomrule\end{tabular}\end{adjustbox}
\end{table}

The linear targets are house value, hourly rental count and wine quality. ACS models the probability of income above \$50,000 with logistic regression. Each target includes an intercept. Bike and ACS predictors use additional features beyond the target covariates. Preprocessing, feature lists and exact training splits are in Supplementary Section~\ref{supp:data}.

For each labeled size, two seeds generate 50 subsets each; the complement is unlabeled. All methods use the same subsets and fixed predictions. Following \citet{ji2025predictions}, we evaluate squared parameter error and interval inclusion against the supervised coefficient on the complete evaluation frame. The fixed-strength reference is selected on 50 separate training subsets and tested on the same 100 fresh subsets as the default. We use the error and interval summaries of Section~\ref{subsec:sim-setup} with this full-frame coefficient as the reference.

Following an occupancy-sensitivity diagnostic, the California point-estimation target excludes AveOccup and has eight parameters including the intercept. The frozen predictor retains its original eight inputs, including AveOccup; the full-frame target is refitted in the reduced model. The original nine-parameter results are retained in Supplementary Table~\ref{tab:california-original9}.

\subsection{Point-estimation performance}
\label{subsec:real-results}
California and Bike show substantial reductions relative to supervised estimation: at $n=500$, the CC-PPI MSE ratios are 58.3\% and 17.4\%, respectively. For the eight-parameter California target, CC-PPI has the lowest mean MSE among the default methods at all ten tested sizes, excluding the independently trained reference. The full size curves and selected comparisons are in Figure~\ref{fig:real-data} and Table~\ref{tab:new-real-point}.

The ordering is less consistent on ACS and Wine (Tables~\ref{tab:main-detail-acs} and~\ref{tab:main-detail-wine}). RePPI has the lowest displayed ACS MSE at $n=2000$, and PDC outperforms CC-PPI on Wine at that size. The independently trained strengths measure the remaining risk difference along the calibration path; they need not improve on the default in every evaluation cell.

\begin{figure}[htbp]\centering
\includegraphics[width=\textwidth]{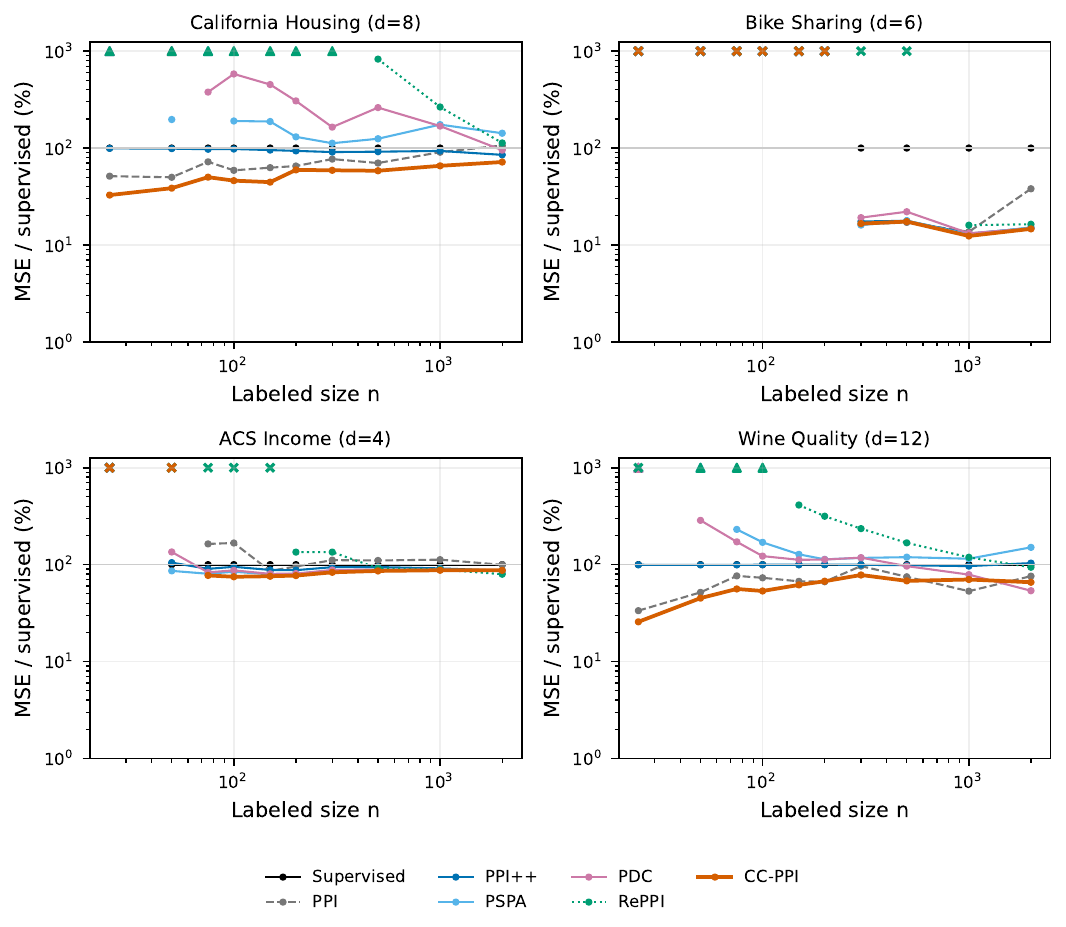}
\caption{Parameter MSE relative to supervised estimation across labeled sizes. Both axes are logarithmic. Triangles indicate finite ratios at least $10^3$; crosses indicate failed fits.}\label{fig:real-data}
\end{figure}
\begin{table}[htbp]\centering
\caption{California and Bike MSE ratios to supervised estimation (\%). Reference strength is trained independently.}\label{tab:new-real-point}
\begingroup\scriptsize
\setlength{\tabcolsep}{2.5pt}
\begin{adjustbox}{max width=\linewidth}\begin{tabular}{lrrrrrrrrr}
\hline
Data & $n$ & Sup. & PPI & PPI++ & PSPA & PDC & RePPI & CC & Ref. \\
\hline
California & 500 & 100.0 & 70.1 & 91.7 & 124.8 & 261.6 & 827.3 & 58.3 & 58.6 \\
California & 2000 & 100.0 & 106.3 & 85.1 & 142.3 & 95.1 & 112.8 & 71.8 & 71.7 \\
Bike & 500 & 100.0 & 17.1 & 17.8 & 17.8 & 22.0 & $\dagger$ & 17.4 & 17.3 \\
Bike & 2000 & 100.0 & 38.0 & 15.0 & 15.2 & 14.8 & 16.4 & 14.7 & 14.7 \\
\hline\end{tabular}\end{adjustbox}
\endgroup

\end{table}

Tables~\ref{tab:main-detail-acs} and~\ref{tab:main-detail-wine} retain the size-specific comparisons. Daggers denote failed fits or unavailable ratios; Supplementary Section~\ref{supp:data} gives the trial counts.
\begin{table}[H]\centering\small
\caption{ACS Income: MSE relative to supervised (percent).}
\label{tab:main-detail-acs}
\begin{adjustbox}{max width=\linewidth}\begin{tabular}{rrrrrrrrr}\toprule
$n$&Sup.&PPI&PPI++&PSPA&PDC&RePPI&CC&CC$^{\rm ref}$\\\midrule
75&100.0&163.5&90.4&80.4&83.2&$\dagger$&77.3&76.1\\
100&100.0&167.1&94.9&84.4&87.2&$\dagger$&74.9&71.4\\
200&100.0&95.6&87.6&80.2&80.8&134.5&77.0&76.0\\
500&100.0&110.3&94.3&89.0&88.8&91.6&86.0&85.0\\
2000&100.0&99.6&86.7&85.8&85.9&79.7&87.2&87.8\\
\bottomrule\end{tabular}\end{adjustbox}
\end{table}
\begin{table}[H]\centering\small
\caption{Wine Quality: MSE relative to supervised (percent).}
\label{tab:main-detail-wine}
\begin{adjustbox}{max width=\linewidth}\begin{tabular}{rrrrrrrrr}\toprule
$n$&Sup.&PPI&PPI++&PSPA&PDC&RePPI&CC&CC$^{\rm ref}$\\\midrule
25&100.0&33.5&99.6&$\ge10^3$&$\ge10^3$&$\dagger$&25.7&24.5\\
50&100.0&51.7&99.3&$\ge10^3$&286.4&$\ge10^3$&45.1&43.4\\
100&100.0&73.1&99.6&169.8&122.4&$\ge10^3$&53.4&51.8\\
500&100.0&74.9&98.5&119.3&96.5&167.8&67.9&65.9\\
2000&100.0&76.0&103.7&150.7&53.8&94.0&65.8&62.8\\
\bottomrule\end{tabular}\end{adjustbox}
\end{table}

\subsection{Supporting interval performance}
\label{subsec:new-real-interval}
The retained interval experiment covers 19 size cells and 1,900 subsamples. Its five California cells use the original nine-parameter target, not the revised eight-parameter point-estimation target; they are sensitivity diagnostics and do not establish interval performance for the reduced model. Wine at $n=2000$ illustrates the width cost: expanded inclusion is 99.00\%, with length 0.255 versus 0.168 for PDC. Supplementary Section~\ref{supp:data} identifies the targets and reports all results.

These are coordinate-average full-frame inclusions. The supervised linear benchmark is HC3 normal; ACS uses a sandwich benchmark and provides an empirical logistic extension of the interval procedure. Its five failed fits at $n=75$ and one at $n=100$ remain in the all-trial denominators. Complete results and successful-fit summaries are in Supplementary Section~\ref{supp:data}.

\section{Discussion}
\label{sec:conclusion}
The complete GMM criterion retains an auxiliary penalty that can affect risk even among first-order equivalent fits. CC-PPI adjusts that penalty through covariance calibration. These are two separate comparisons: the full criterion can improve on bounded scalar corrected-score roots, and calibration can further improve on the uncalibrated fit under the stated regression conditions.

The gains can reverse as prediction error grows. Independently trained strengths improve on half strength in the eight-law experiment, but they use the generating parameter during training. A practical adaptive rule must estimate risk without consuming the gain available from calibration.

The strict improvement results cover a Gaussian linear low-noise regime and a structured logistic model conditional on pilot existence. Broader designs require further analysis. For inference, fixed anchor mismatch can leave the expanded bootstrap interval with nonzero asymptotic width, so its coverage and precision should be evaluated separately from those of the parent interval. Extending the risk analysis and controlling this width cost are practical questions for further work.

\clearpage
\phantomsection
\addcontentsline{toc}{section}{Supplementary Material}
\begin{center}{\LARGE\bfseries Supplementary Material}\end{center}
The supplement is included in this document. Sections S1--S7 provide method derivations, first-order theory, finite-sample risk, strength paths, interval inference, simulation details and data applications. All cross-references link within this document.
\setcounter{section}{0}
\setcounter{table}{0}\setcounter{figure}{0}
\renewcommand{\thesection}{S\arabic{section}}
\renewcommand{\theHsection}{supp.\arabic{section}}
\numberwithin{table}{section}\numberwithin{figure}{section}
\section{Method derivations and computation}
\label{supp:method}
\subsection{Joint moment covariance}
\label{subsec:covariance}

Set $M=\min(n,N)$, $c_{1,n}=M/n$, and $c_{2,n}=M/N$. At $\theta_*$, write
\begin{equation*}\Sigma_{11}=\operatorname{Var}(\psi),\quad
 \Sigma_{1h}=\operatorname{Cov}(\psi,h),\quad
 \Sigma_{hh}=\operatorname{Var}(h),\end{equation*}
with arguments suppressed. Independence gives
\begin{equation}\operatorname{Var}\{\sqrt M\,\widehat g(\theta_*)\}=S_n=
 \begin{pmatrix}
 c_{1,n}\Sigma_{11}&-c_{1,n}\Sigma_{1h}\\
 -c_{1,n}\Sigma_{1h}^\top&(c_{1,n}+c_{2,n})\Sigma_{hh}
 \end{pmatrix}.\label{eq:population-covariance}\end{equation}
The negative cross block follows from the sign of the labeled auxiliary mean. Under a stable sample-size ratio, $S_n$ has a limit $S$. Positive definiteness of this limit is a condition for the regular efficiency analysis, not a restriction on every local sequence studied later.

Obtain a supervised initial estimate $t$, typically by solving $m_1(t)=0$, and evaluate
\begin{equation}s_i=\psi(X_i,Y_i;t),\quad h_i=h(X_i;t),\quad h_j^U=h(\widetilde X_j;t).\label{eq:initial-scores}\end{equation}
For a sample of size $k$, use the centered covariance
\begin{equation*}\widehat{\operatorname{Cov}}_k(a,b)=\frac1k\sum_{i=1}^k(a_i-\bar a)(b_i-\bar b)^\top.\end{equation*}
If $m_1(t)=0$ exactly, centering supervised scores changes neither their covariance nor the cross covariance with auxiliary scores. With a numerical regularizer $\varepsilon_M\geq0$, define
\begin{align*}
 \widehat V&=c_{1,n}\widehat{\operatorname{Cov}}_{\mathcal L}(s,s)+\varepsilon_M I_d,\\
 \widehat R&=c_{1,n}\widehat{\operatorname{Cov}}_{\mathcal L}(s,h),\\
 \widehat D&=c_{1,n}\widehat{\operatorname{Cov}}_{\mathcal L}(h,h)
 +c_{2,n}\widehat{\operatorname{Cov}}_{\mathcal U}(h^U,h^U)+\varepsilon_M I_d,\\
 \widehat S&=\begin{pmatrix}\widehat V&-\widehat R\\-\widehat R^\top&\widehat D\end{pmatrix}.
\end{align*}
Before regularization, this is a scaled labeled covariance of $(s_i^\top,-h_i^\top)^\top$ plus a positive semidefinite unlabeled contribution. It is therefore positive semidefinite; $\varepsilon_M>0$ makes it positive definite. The ridge value is part of the estimator specification. Retaining the unregularized covariance limit requires $\varepsilon_M\to0$.

\subsection{A family of covariance-weighted estimators}
\label{subsec:family}

Let $S_{0,n}(\alpha)$ be a symmetric positive definite target, where $\alpha$ indexes a specified target collection. Targets may depend on the data, initial estimate, and sample sizes. For $0\leq\lambda\leq1$, set
\begin{align}
 \widehat S_{\lambda,\alpha}&=(1-\lambda)\widehat S+\lambda S_{0,n}(\alpha),\nonumber\\
 W_{\lambda,\alpha}&=\widehat S_{\lambda,\alpha}^{-1},\nonumber\\
 Q_n(\theta;\lambda,\alpha)&=\widehat g(\theta)^\top W_{\lambda,\alpha}\widehat g(\theta),\label{eq:criterion}\\
 \widehat\theta_{\lambda,\alpha}&\in\operatorname*{arg\,min}_{\theta\in\Theta}Q_n(\theta;\lambda,\alpha).\label{eq:family-estimator}
\end{align}
If $\widehat S\succ0$, the blend is positive definite for all $\lambda\in[0,1]$. If the unregularized $\widehat S$ is singular, a positive definite target gives an invertible blend for $\lambda>0$; the endpoint $\lambda=0$ requires separate treatment.

Both tuning quantities may be data dependent. After selection, $t$, $\widehat S$, the target, and the tuning quantities are held fixed during minimization. Recomputing the weight at each candidate $\theta$ defines a different criterion with additional derivative terms.

The target determines the covariance structure to be regularized, and $\lambda$ determines the regularization strength. Their joint effect is assessed through parameter risk.

The $2d\times2d$ matrix $W_{\lambda,\alpha}$ weights the stacked moments; the $d\times d$ matrix $\Lambda$ in $m_1+\Lambda m_2=0$ is a correction coefficient. The relation between them is derived in the next section.

CC-PPI uses the target derived below. Supplementary Section~\ref{subsec:targets} defines cross-block, spherical and diagonal alternatives and a sample-size benchmark strength.

\subsection{Auxiliary penalty in the GMM criterion}
\label{subsec:finite-choice}

By Corollary~\ref{cor:common-limit}, two regular family members with vanishing covariance perturbations satisfy $\sqrt M(\widehat\theta^{(a)}-\widehat\theta^{(b)})=o_p(1)$. Comparing their finite-sample MSEs requires the higher-order contributions of covariance estimation, the sample auxiliary Jacobian and strength selection. Passing from distributional expansions to risk also requires moment control.

Higher-order GMM/GEL analyses retain contributions from moment/Jacobian dependence and estimated weights \citep{neweysmith2003}; variance estimation for linear two-step GMM may also need a finite-sample correction \citep{windmeijer2005}. These results motivate retaining same-sample weight effects.

The criterion decomposition \eqref{eq:fs-anchor-decomposition} gives the correction and penalty metrics.

For affine $m_2(\theta)=A_5\theta-b_2$, the second term has curvature $A_5^\top\Pi_WA_5$. When $A_5$ is invertible, it is a quadratic penalty around the root of $m_2=0$; otherwise it remains a penalty on the auxiliary moment, without requiring a unique anchor. Thus this penalty is a property of the full quadratic criterion. An equivalent estimating-equation representation can encode the same effect through its coefficients. Both $\Lambda_W$ and $\Pi_W$ change with the weight, so the penalty alone does not determine parameter risk. Under regular identical-covariate sampling, $A_5\to_p0$; its finite-sample contribution is consequently not characterized by the limiting Jacobian $G=(A^\top,0^\top)^\top$.

For exact risk calculations, take $\Theta=\mathbb R^d$ and condition on full-column-rank $X\in\mathbb R^{n\times d}$, arbitrary $U\in\mathbb R^{N\times d}$, and externally supplied finite predictions $f_L,f_U$. Let
\begin{align}
 Y&=X\beta+\epsilon,\quad C=(X^\top X)^{-1}X^\top,\quad t=CY,\label{eq:fs-regression}\\
 A_L&=X^\top X/n,\quad A_U=U^\top U/N,\quad
 H=\begin{pmatrix}A_L\\A_U-A_L\end{pmatrix},\nonumber\\
 b&=\begin{pmatrix}X^\top Y/n\\U^\top f_U/N-X^\top f_L/n\end{pmatrix}.
\nonumber\end{align}
Then $\widehat g(\theta)=H\theta-b$, and the complete frozen-weight solution is \eqref{eq:affine-solution}. The matrix $H$ has full column rank even if $A_U-A_L$ is singular. The predictor may be nonlinear in the covariates. Conditional parameter risk is $R(\lambda)=\mathbb E_\epsilon\|\widehat\beta_\lambda-\beta\|^2$. Random-design conclusions require averaging these conditional quantities with appropriate integrability; no dimension-growth claim is made here.

\subsection{Covariance calibration and finite-sample risk}
\label{subsec:calibrated-target}

Write $\rho=n/N$ and use the $n$-scaled covariance $\mathsf S=\widehat S/c_{1,n}$. Its scores are those in \eqref{eq:initial-scores}. With $C_L=\widehat{\operatorname{Cov}}_{\mathcal L}(h,h)$, $C_U=\widehat{\operatorname{Cov}}_{\mathcal U}(h^U,h^U)$, and $\eta=\varepsilon_M/c_{1,n}>0$,
\begin{equation*}\mathsf S=\widehat{\operatorname{Cov}}_{\mathcal L}
       \begin{pmatrix}s\\-h\end{pmatrix}
       +\operatorname{diag}(0,\rho C_U)+\eta I_{2d}.\end{equation*}
Multiplication of the full covariance by a positive scalar leaves its GMM minimizer unchanged. In this section, $Q$ denotes the $n$-scaled criterion, namely $c_{1,n}$ times the $M$-scaled criterion in \eqref{eq:criterion}; it is the normalization already used in Section~\ref{sec:setup}. In the experiments below, $N\ge n$, $c_{1,n}=1$, and $\eta=n^{-2}$.

The following characterization gives the covariance blocks used throughout this section.

The calibration retains the corrected covariance and cross covariance (P1) and replaces the labeled auxiliary covariance by its unlabeled counterpart (P2).

Write $c=m_1+m_2$. The two requirements are expressed in the following block coordinates.\begin{equation}B_c=\begin{pmatrix}I&I\\0&I\end{pmatrix},\quad T=\begin{pmatrix}I\\-I\end{pmatrix},\quad
 B_c\mathsf S B_c^\top=
 \begin{pmatrix}V_c&R_c\\R_c^\top&D_c\end{pmatrix},\quad
 \Omega=D_c-R_c^\top V_c^{-1}R_c.\label{eq:fs-corrected-basis}\end{equation}
\begin{proposition}[Channel-preserving calibration]\label{prop:fs-calibration-class}
A symmetric positive definite target $\mathsf S_* $ preserves the first block row $(V_c,R_c)$ in \eqref{eq:fs-corrected-basis} if and only if
\begin{equation}\mathsf S_* =\mathsf S+TKT^\top,\qquad K=K^\top,\qquad\Omega+K\succ0.\label{eq:fs-calibration-class}\end{equation}
For every such target, including nonlinear moments, its criterion is exactly
\begin{equation}Q_K(\theta)=c(\theta)^\top V_c^{-1}c(\theta)
       +r_c(\theta)^\top(\Omega+K)^{-1}r_c(\theta),\quad
 r_c=m_2-R_c^\top V_c^{-1}c.\label{eq:fs-residual-criterion}\end{equation}
\end{proposition}
The first term uses the original corrected-moment covariance. The second weights the auxiliary residual after linear adjustment for that moment. Calibration modifies the latter metric for both affine and nonlinear scores.

Both $V_c$ and $\Omega$ are positive definite: the calibration principle is to retain the corrected channel $(V_c,R_c)$ while changing the auxiliary residual metric.

For a general covariance pool, P2 may use a symmetric pool $F(C_L,C_U)$, giving auxiliary block $F(C_L,C_U)+\rho C_U+\eta I$.\begin{proposition}[The principles force the target]\label{prop:fs-target-forcing}
Fix $\mathsf S$ with blocks as in \eqref{eq:fs-corrected-basis}, and let $F(C_L,C_U)$ be any symmetric prediction-covariance pool. Let $\mathsf S_*$ satisfy P1 together with
\[\text{P2$'$: the auxiliary block of } B_c\mathsf S_*B_c^\top \text{ equals } F(C_L,C_U)+\rho C_U+\eta I .\]
Then the unique symmetric candidate is
\begin{equation}\mathsf S_*=\mathsf S+TK_FT^\top,\qquad K_F=F(C_L,C_U)-C_L ;\label{eq:fs-forced-k}\end{equation}
it is a positive definite target if and only if $\Omega+K_F\succ0$, and its criterion is \eqref{eq:fs-residual-criterion} with $K=K_F$. For the pool $F=C_U$ the auxiliary block becomes $(1+\rho)C_U+\eta I$, its population shape, and $K_F=C_U-C_L$.
\end{proposition}
The choice $F=C_U$ yields \eqref{eq:ccppi-raw-update} using prediction scores observed in both samples. Spectral protection ensures admissibility of the fitted update. Proofs of the block identities are in Supplementary Section~\ref{supp:method}.

\par\medskip\noindent{\bfseries Unlabeled prediction-score covariance}\par\nobreak\smallskip

\label{subsubsec:why-calibration}
The unlabeled sample supplies information for calibrating the prediction-score covariance within the auxiliary block. Under the regular pilot laws, $C_U-C_L\to_p0$, preserving the efficient limit. At finite samples, its joint dependence with the pilot and estimated weight is retained in the signed transfer analysis of Section~\ref{subsec:finite-risk-selection}. Centered divisors also enter that analysis: at an external deterministic pilot, $\mathbb EC_L=(1-1/n)\Sigma_{hh}$ and $\mathbb EC_U=(1-1/N)\Sigma_{hh}$. The target preserves the common population limit while changing the estimated auxiliary covariance at finite sample sizes.

\par\medskip\noindent{\bfseries Channel-preserving targets}\par\nobreak\smallskip

\label{subsubsec:calibration-class}

We prove Proposition~\ref{prop:fs-calibration-class}.

\begin{proof}
The transformed covariance difference must be $\operatorname{diag}(0,K)$. Since $B_cT=(0,-I)^\top$, transforming back gives \eqref{eq:fs-calibration-class}, with unique $K$. Positive definiteness is equivalent to $\Omega+K\succ0$ by the Schur complement. The block inverse gives \eqref{eq:fs-residual-criterion}.
\end{proof}

For a fixed covariance pool, the following two requirements determine a unique raw target. They define the construction and impose no risk-optimality or data-allocation criterion.

\begin{proof}
For Proposition~\ref{prop:fs-target-forcing}, among symmetric candidates P1 is equivalent to $\mathsf S_*=\mathsf S+TKT^\top$ for a unique symmetric $K$, by transforming the covariance difference into corrected coordinates. This algebraic characterization does not assume positive definiteness. Admissibility is checked separately by $\Omega+K\succ0$. Since $B_cT=(0,-I)^\top$, the corrected blocks of $\mathsf S+TKT^\top$ are $(V_c,R_c,D_c+K)$, so P2$'$ is equivalent to $D_c+K=F(C_L,C_U)+\rho C_U+\eta I$. With $D_c=C_L+\rho C_U+\eta I$ this gives \eqref{eq:fs-forced-k}, and positive definiteness is the Schur condition $\Omega+K_F\succ0$.
\end{proof}

Once the pool is fixed, P1 and P2 determine the target. We take $F=C_U$. Replacing the labeled prediction-covariance contribution $C_L$ by $C_U$, while preserving the corrected channel, gives the raw correction
\begin{equation}K=C_U-C_L,\qquad \Delta=T(C_U-C_L)T^\top,\label{eq:fs-raw-calibration}\end{equation}
and the target \eqref{eq:fs-target-construction} is \eqref{eq:fs-raw-calibration} after the spectral protection that follows.

Apply the spectral protection of \eqref{eq:ccppi-spectral-protection} to the raw correction. Here $h=1/2$ denotes the default strength and $\kappa=1/2$ the spectral cap. Protection retains the form $\Delta_p=TK_pT^\top$, so the frozen criterion along the implemented path is

\begin{equation}Q_\lambda=c^\top V_c^{-1}c+r_c^\top(\Omega+\lambda K_p)^{-1}r_c.\label{eq:fs-protected-criterion}\end{equation}

The low-rank protection formula and its admissibility properties are given in Supplementary Section~\ref{supp:protected-path}. Under the covariance consistency conditions of Section~\ref{subsec:regularity}, the calibrated path satisfies

\begin{equation}\sup_{0\le\lambda\le1}\|\widehat S_\lambda-\widehat S\|\to_p0.\label{eq:fs-uniform-efficient-path}\end{equation}

Corollary~\ref{cor:common-limit} therefore gives the common efficient limit for the half-strength estimator.

\label{app:calibration-proofs}

\subsection{Finite-sample rates}
\label{subsubsec:fs-scales}

The following rate statement concerns random designs; subsequent exact risk calculations condition on a realized full-rank design. Assume fixed dimension, independent identically distributed covariate samples with $N/n$ bounded below, and finite second moments of $XX^\top$ and $Xf(X)$. These are fourth-moment requirements on the relevant quadratic design scores. For the affine regression moments, empirical coefficient bounds then hold uniformly on compact parameter sets. Assume also a regular positive definite covariance limit and a consistent root-$n$ baseline fit. For the risk expansion below, impose the stated $L^2$ remainder bounds; distributional orders alone do not imply risk orders.

\begin{proposition}[Scales of the finite-sample choice]\label{prop:fs-scales}
Under these conditions:
\begin{enumerate}
\item[(a)] $A_5=O_p(n^{-1/2})$ and $m_2(\theta)=O_p(n^{-1/2})$ uniformly on compact sets. For weights with bounded Schur complements, $m_2^\top\Pi_Wm_2=O_p(n^{-1})$.
\item[(b)] Let the whitened update have amplitude at most a deterministic $\gamma_n<1$, bounded away from one. On a uniformly well-conditioned affine fit path, $\widehat\beta_\lambda-\widehat\beta_0=O_p(\gamma_n n^{-1/2})$ uniformly for $0\le\lambda\le1$.
\item[(c)] Write $e_0=\widehat\beta_0-\beta$ and $\widehat\beta_\lambda=\widehat\beta_0+\lambda v+\lambda^2w+O_{L^2}(\lambda^3)$, with integrable loss remainders. Then the linear risk coefficient is $2\mathbb E\langle e_0,v\rangle$ and the quadratic coefficient is $\mathbb E\|v\|^2+2\mathbb E\langle e_0,w\rangle$. If $\|e_0\|_{L^2}=O(n^{-1/2})$, $\|v\|_{L^2}=O(\gamma_n n^{-1/2})$ and $\|w\|_{L^2}=O(\gamma_n^2n^{-1/2})$, these coefficients are bounded by $O(\gamma_n/n)$ and $O(\gamma_n^2/n)$, respectively.
\end{enumerate}
\end{proposition}

The rate bounds locate the finite-sample effects generated by the auxiliary Jacobian and covariance update. Their signed coefficients are examined through the exact conditional-risk identities in the main paper. Section~\ref{subsec:convergent-risk} provides a convergent approximation with an explicit remainder, and Supplementary Section~\ref{app:scales-proof} gives the rate argument.

\subsection{Target and strength parameterization}
\label{subsubsec:target-strength}
For a symmetric $K$, the channel-preserving family has additive form
\begin{equation}\mathsf S_{\lambda,K}=(1-\lambda)\mathsf S+\lambda(\mathsf S+TKT^\top)
 =\mathsf S+\lambda TKT^\top .\label{eq:fs-family-additive}\end{equation}
Its residual metric is $\Omega+\lambda K$. At unit strength, admissibility requires $\Omega+K\succ0$; the other targets in Section~\ref{subsec:targets} need not preserve the corrected channel.

\label{subsubsec:calibrated-target}
CC-PPI uses the protected correction $K_p$ below, with $S_{0,n}^{\rm cal}=c_{1,n}(\mathsf S+TK_pT^\top)$ and default strength $h=1/2$. Protection clips eigenvalues in relative covariance coordinates, preserving the channel and bounding the update; entrywise clipping does not provide this guarantee. Section~\ref{subsec:strength-optimality} studies the risk cost of the default.

\subsection{Protection and the common efficient path}
\label{supp:protected-path}

The protection is the one announced in \eqref{eq:fs-target-construction}: it uses the whitened coordinates of the current covariance.
\begin{align}
 E_p&=\operatorname{clip}_{[-\kappa,\kappa]}
       (\mathsf S^{-1/2}\Delta\mathsf S^{-1/2}),\quad 0<\kappa<1,\label{eq:fs-protection}\\
 \Delta_p&=\mathsf S^{1/2}E_p\mathsf S^{1/2},\qquad
 S_{0,n}^{\rm cal}=c_{1,n}(\mathsf S+\Delta_p),\nonumber\\
 \widehat S_\lambda&=c_{1,n}(\mathsf S+\lambda\Delta_p),\qquad 0\le\lambda\le1.
 \label{eq:fs-implemented-path}
\end{align}
Here clipping acts on symmetric eigenvalues, and is performed once before varying strength. The target subscript $0,n$ follows Section~\ref{sec:setup}; it does not denote the endpoint $\lambda=0$. We use $\kappa=1/2$ and the default $h=1/2$.

Protection preserves the calibration class. Indeed, with $F=T^\top\mathsf S^{-1}T\succ0$,
\begin{equation}\Delta_p=TK_pT^\top,\qquad
 K_p=F^{-1/2}\operatorname{clip}_{[-\kappa,\kappa]}
               (F^{1/2}KF^{1/2})F^{-1/2}.\label{eq:fs-protected-lowrank}\end{equation}
To verify this, $\mathsf S^{-1/2}TF^{-1/2}$ has orthonormal columns, and the clipping function vanishes at zero. The nonzero spectral update stays on that subspace. The resulting criterion is \eqref{eq:fs-protected-criterion}.

If $K_p\succeq0$, increasing strength weakens this residual penalty; if $K_p\preceq0$, it strengthens it. Indefinite corrections have no uniform ordering. These are objective statements, not MSE statements. The residual root, when it exists, need not coincide with the root of $m_2$.

Within the linear, congruence-equivariant pooling family, Proposition~\ref{prop:pooling-characterization} gives $F=(1-\xi)C_L+\xi C_U$. CC-PPI takes the unlabeled pool $F=C_U$, directing the calibration toward the additional covariance information while retaining the corrected channel.

Under the covariance consistency conditions of Section~\ref{subsec:regularity}, both $C_L$ and $C_U$ converge to $\Sigma_{hh}$. At a regular positive definite limit, \eqref{eq:fs-protection} is continuous and $\Delta_p\to_p0$. This proves \eqref{eq:fs-uniform-efficient-path}.

Corollary~\ref{cor:common-limit} applies to half strength and to bounded data-dependent selections on this path. A nonvanishing strength therefore need not sacrifice efficiency: the target perturbation itself vanishes. This statement excludes singular limiting covariances.

\subsection{Amplitude, admissibility and local response}
\label{supp:calibration-properties}

The strength and spectral cap jointly control the relative update amplitude. Write $F=T^\top\mathsf S^{-1}T$ and $Q_T=\mathsf S^{-1/2}TF^{-1/2}$, so that $Q_T^\top Q_T=I$ because $T^\top\mathsf S^{-1}T=F$; by the block identity of \eqref{eq:fs-corrected-basis}, $F=\Omega^{-1}$. For symmetric $K$ put
\begin{equation*}\gamma_K=\|F^{1/2}KF^{1/2}\|_{\rm op}=\|\Omega^{-1/2}K\Omega^{-1/2}\|_{\rm op},\end{equation*}
the two expressions agreeing because $Q_T$ has orthonormal columns.

\begin{proposition}[Amplitude composition]\label{prop:fs-amplitude}
(i) $\|\mathsf S^{-1/2}TKT^\top\mathsf S^{-1/2}\|_{\rm op}=\gamma_K$. (ii) The reachable set of targets at relative amplitude at most one is $\{\mathsf S+TKT^\top:\gamma_K\le1,\ \Omega+K\succ0\}$. (iii) For every $\lambda\in[0,1]$,
\begin{equation}(1-\lambda\gamma_K)\mathsf S\preceq\mathsf S+\lambda TKT^\top\preceq(1+\lambda\gamma_K)\mathsf S .\label{eq:fs-amplitude-sandwich}\end{equation}
\end{proposition}
\begin{proof}
Since $T^\top\mathsf S^{-1}T=F$, we have $\mathsf S^{-1/2}T=Q_TF^{1/2}$, so \[\mathsf S^{-1/2}TKT^\top\mathsf S^{-1/2}=\allowbreak Q_TF^{1/2}KF^{1/2}Q_T^\top\]; taking norms and using $\|Q_TAQ_T^\top\|_{\rm op}=\|A\|_{\rm op}$ for $Q_T^\top Q_T=I$ gives (i). Part (ii) is the definition of the set, and (iii) is (i) in Loewner form, since $\|E\|_{\rm op}\le a$ for symmetric $E$ gives $-a\mathsf S\preceq\mathsf S^{1/2}E\mathsf S^{1/2}\preceq a\mathsf S$.
\end{proof}

In \eqref{eq:fs-amplitude-sandwich}, strength and cap enter through $\lambda\gamma_K$. Thus $(\lambda,\kappa)=(1/2,1/2)$ and $(1,1/4)$ have the same amplitude bound, although their updates need not coincide. The cap ensures admissibility.

\begin{proposition}[Admissible caps]\label{prop:fs-admissible-caps}
(i) For every $0<\kappa<1$, every $\lambda\in[0,1]$ and every symmetric $K$ with $\gamma_K\le\kappa$, the blend $\mathsf S+\lambda TKT^\top$ is positive definite. (ii) For the two-parameter family $K(a,b)=aK_p+b\gamma\Omega$, $|a|+|b|\le1$, with $\gamma=\gamma_{K_p}$, positive definiteness follows from $1+b\gamma>|a|\gamma$; this holds for every admissible $(a,b)$ whenever $|a|+|b|>0$ and $\gamma<1/(|a|+|b|)$; at $(a,b)=(0,0)$ the target is unchanged. It fails at $(a,b)=(1/2,-1/2)$ when $\gamma=1$ and the negative extremal direction is attained.
\end{proposition}

Part (i) gives the positive-definiteness guarantee provided by the cap: it delivers the invariant $\gamma\le\kappa$ of \eqref{eq:fs-protection} and hence positive definiteness along the whole path at every strength, and at the default strength $h=1/2$ with $\kappa=1/2$ it gives the relative interval $[3/4,5/4]$ in \eqref{eq:fs-amplitude-sandwich}. Part (ii) shows the role of the strict amplitude bound: at unit amplitude the borderline combination of the calibrated and isotropic directions leaves the class. No optimality of $\kappa=1/2$ is claimed; any cap strictly below one gives the positive-definiteness guarantee in part (i).

The criterion depends on the residual metric only through its compression to the span of the realized residuals.

\begin{proposition}[Only the realized residual span matters]\label{prop:fs-subspace-irrelevance}
Let $r_c:\Theta\to\mathbb R^d$ be arbitrary, put $V=\operatorname{span}\{r_c(\theta):\theta\in\Theta\}$, and let $\Pi$ be the orthogonal projector onto $V$. Then for every symmetric positive definite $M$ and every $\theta\in\Theta$,
\begin{equation}r_c(\theta)^\top M^{-1}r_c(\theta)=r_c(\theta)^\top\Pi M^{-1}\Pi\,r_c(\theta).\label{eq:fs-irrelevance}\end{equation}
Consequently the criterion \eqref{eq:fs-residual-criterion} depends on the target only through the compression of $(\Omega+K)^{-1}$ to $V$, and two targets whose compressions agree define the same estimator whenever that criterion has a unique minimizer. For affine $r_c(\theta)=J\theta-j$ and $\Theta$ containing a nonempty open subset, equality of the criteria for every $\theta\in\Theta$ implies equality of the compressions. Equality of their minimizers alone does not imply this converse. If the realized residual set is thin, even the criteria can agree while the compressions differ.
\end{proposition}

The following local-in-strength result applies to smooth moments. For the affine regression moments in \eqref{eq:fs-regression}, its coefficients also follow from the exact path.

\begin{proposition}[Local risk response to a covariance update]\label{prop:fs-dependence}
Let $\theta_\lambda$ be a smooth local branch of interior minimizers of $Q_{W_\lambda}$, with $W_\lambda=(\mathsf S+\lambda\Delta)^{-1}$ and $\mathsf S\succ0$. Put $H=\tfrac12\nabla^2Q_{\mathsf S^{-1}}(\theta_0)\succ0$. Assume a third-order path expansion and domination sufficient to expand expected squared loss. Then
\begin{equation}v=\left.\partial_\lambda\theta_\lambda\right|_0
 =H^{-1}J(\theta_0)^\top\mathsf S^{-1}\Delta\mathsf S^{-1}\widehat g(\theta_0).\label{eq:fs-dependence-displacement}\end{equation}
With $e_0=\theta_0-\theta_*$ and $w=\tfrac12\left.\partial_\lambda^2\theta_\lambda\right|_0$,
\begin{equation}R(\lambda)-R(0)=2\lambda\mathbb E\langle e_0,v\rangle
 +\lambda^2\{\mathbb E\|v\|^2+2\mathbb E\langle e_0,w\rangle\}
 +O(\lambda^3).\label{eq:fs-dependence-risk}\end{equation}
If $\Delta$ is independent of the full baseline map, including $H$ and $J(\theta_0)$, and $\mathbb E\Delta=0$, the linear coefficient vanishes. No sign of the quadratic coefficient follows from this independence alone.
\end{proposition}

The first-order risk alignment is linear in the update, whereas the acceleration term is quadratic. Both contribute to the local response in \eqref{eq:fs-dependence-risk}. The exact signed identity in the main paper retains the complete dependence induced by the data-based calibration; Supplementary Section~\ref{app:dependence} also gives an independent centered-update example illustrating the acceleration contribution.

\subsection{The calibration as a quadratic score-covariance map}
\label{supp:quadratic-target}

The target has a concrete moment interpretation in regression. For each sample, put $B_i=x_ix_i^\top-A_D$ and $a_i=x_if_i-\overline{xf}_D$, where $D=L,U$. Then
\begin{equation}C_D(t)=\operatorname{ave}_D\{(B_it-a_i)(B_it-a_i)^\top\}.\label{eq:fs-target-quadratic}\end{equation}
The update depends on fourth-order design moments, prediction-weighted cross moments and prediction-score second moments through a quadratic function of $t$. Covariate means, covariances or prediction MSE alone do not determine the sign in Theorem~\ref{thm:fs-signed-transfer}. Theorem~\ref{thm:fs-primitive-improvement} illustrates the dependence on design: $p+u<1$ is favorable and $p+u>1$ adverse. Theorem~\ref{thm:linear-multiparameter} supplies a further Gaussian multiparameter low-noise regime. Improvement conditions for arbitrary designs remain open.

\subsection{Quadratic decomposition and channel preservation}

Expanding the first term of \eqref{eq:fs-anchor-decomposition} gives cross term $2m_1^\top W_{12}m_2$ and lower term $m_2^\top W_{21}\allowbreak W_{11}^{-1}W_{12}m_2$. Adding the Schur complement restores $m_2^\top W_{22}m_2$. Positive definiteness of $W$ implies $W_{11}\succ0$ and $\Pi_W\succ0$. For $m_2=A_5\theta-b_2$, its quadratic coefficient is $A_5^\top\Pi_WA_5$. The Hessian of that term is twice this coefficient. No inverse of $A_5$ is used unless an anchor root is explicitly defined.

An invertible simultaneous transformation of moments and covariance leaves the criterion invariant. In the corrected basis, channel preservation makes the symmetric difference $\operatorname{diag}(0,K)$. Since $B_cT=(0,-I)^\top$, transformation back proves uniqueness of $TKT^\top$. The factorization
\begin{equation*}\begin{pmatrix}V_c&R_c\\R_c^\top&D_c+K\end{pmatrix}
 =\begin{pmatrix}I&0\\R_c^\top V_c^{-1}&I\end{pmatrix}
   \operatorname{diag}(V_c,\Omega+K)
   \begin{pmatrix}I&V_c^{-1}R_c\\0&I\end{pmatrix}\end{equation*}
proves both the definiteness test and \eqref{eq:fs-residual-criterion}. Along the protected path, only $(\Omega+\lambda K_p)^{-1}$ changes. Its derivative is $-(\Omega+\lambda K_p)^{-1}K_p(\Omega+\lambda K_p)^{-1}$, establishing the stated ordering for positive or negative corrections. This ordering alone does not order the minimizers' MSEs.

\subsection{Pooling characterization and its boundaries}

\begin{proposition}[Pooling characterization]\label{prop:pooling-characterization}
If a raw pool $\mathcal F(A,D)$ is linear on pairs of symmetric matrices, equivariant under all invertible congruences, positive definite for positive definite inputs, and satisfies $\mathcal F(A,A)=A$, then
\begin{equation*}\mathcal F(A,D)=(1-\lambda)A+\lambda D,\qquad\lambda\in[0,1].\end{equation*}
\end{proposition}

\begin{proof}
Write a linear pool as $\mathcal F(A,D)=\mathcal L_1(A)+\mathcal L_2(D)$. Each component is congruence-equivariant by setting the other input to zero in the equivariance identity. For $d\ge2$, let $E_{11}=e_1e_1^\top$ and $Q_t=\operatorname{diag}(1,tI_{d-1})$. Because $Q_tE_{11}Q_t^\top=E_{11}$ for every $t>0$, invariance forces the off-diagonal and lower blocks of $\mathcal L_j(E_{11})$ to vanish. Thus $\mathcal L_j(E_{11})=a_jE_{11}$. For $d=1$, every linear component already has the scalar form $\mathcal L_j(A)=a_jA$. For $d\ge2$ and any nonzero $x$, choose an invertible $Q$ with $Qe_1=x$. Equivariance gives $\mathcal L_j(xx^\top)=a_jxx^\top$. Rank-one spectral decompositions and linearity yield $\mathcal L_j(A)=a_jA$ for every symmetric $A$; the coefficients do not depend on the argument values. The common-input condition gives $a_1+a_2=1$. It remains to exclude a negative coefficient, as follows. Suppose instead that $a_1<0$, so $a_2=1-a_1>1>0$; fix any $A\succ0$ and set $D=-(a_1/a_2)A$, which is positive definite because $-a_1/a_2>0$. Then $\mathcal F(A,D)=a_1A+a_2D=0$, contradicting the positivity axiom. Hence $a_1\ge0$, and exchanging the roles gives $a_2\ge0$. With $a_1+a_2=1$ both coefficients lie in $[0,1]$, so $\lambda:=a_2$ satisfies $\mathcal F(A,D)=(1-\lambda)A+\lambda D$. The converse follows directly.
\end{proof}

These assumptions constrain the construction, not all uses of the available information. The characterization is sharp: the harmonic pool $2(A^{-1}+D^{-1})^{-1}$ is positive, congruence-equivariant and common-input consistent, but nonlinear. For $A=\operatorname{diag}(1,4)$ and $D=\operatorname{diag}(4,1)$ it equals $\operatorname{diag}(1.6,1.6)$, requiring different scalar mixture coefficients in the two coordinates. It nevertheless remains in the same residual mechanism if it is inserted through Proposition~\ref{prop:fs-calibration-class} and the resulting covariance is positive definite.

If targets share one protected update, $S_j=\mathsf S+\lambda_j\Delta_p$, any convex average is $\mathsf S+(\sum_jw_j\lambda_j)\Delta_p$. Separately protecting scaled raw targets need not have this property: eigenvalue clipping is nonlinear. Structural equivalence, equality of covariance paths, and equality of estimates are therefore distinct.

\subsection{Protection, convergence, and its exact support}

Let $F=T^\top\mathsf S^{-1}T$ and $Q_T=\mathsf S^{-1/2}TF^{-1/2}$. Then $Q_T^\top Q_T=I$ and
\begin{equation*}\mathsf S^{-1/2}TKT^\top\mathsf S^{-1/2}
       =Q_T(F^{1/2}KF^{1/2})Q_T^\top.\end{equation*}
The orthogonal complement has eigenvalue zero. Applying a spectral function that vanishes at zero and transforming back proves \eqref{eq:fs-protected-lowrank}. In particular $\operatorname{rank}(\Delta_p)\le d$ and channel preservation survives active clipping. Since $-\kappa I\preceq E_p\preceq\kappa I$,
\begin{equation*}(1-\lambda\kappa)\mathsf S\preceq\mathsf S+\lambda\Delta_p
                 \preceq(1+\lambda\kappa)\mathsf S.\end{equation*}
At $\lambda=1/2$, $\kappa=1/2$, the relative covariance interval is $[3/4,5/4]$ and the relative inverse interval is $[4/5,4/3]$.

Under the pilot uniform laws of Supplementary Section~\ref{app:asymptotic-proofs}, $C_L-C_U\to_p0$. The regular limit $\mathsf S\to_p S/c_1\succ0$ makes square roots, inverse square roots and clipping continuous there. Consequently $\Delta_p\to_p0$. Since $0\le\lambda\le1$ and $c_{1,n}\le1$, the supremum in \eqref{eq:fs-uniform-efficient-path} is at most $\|\Delta_p\|$. The common efficiency corollary then applies to any measurable selected strength, without requiring convergence of that strength.

\subsection{Growth, weak derivatives, and Gaussian integrability}
\label{app:fit-regularity}

Here the design, prediction values, sample sizes and positive ridge are fixed. Put $S_\lambda=\mathsf S+\lambda\Delta_p$ and $b=b(Y)$. Score vectors are affine in $Y$ through $t=CY$, so for fixed constants $c_i$,
\begin{equation*}\|\mathsf S(Y)\|\le c_0(1+\|Y\|^2),\quad
 \|b(Y)\|\le c_1(1+\|Y\|),\quad
 \lambda_{\min}(S_\lambda)\ge(1-\kappa)\eta.\end{equation*}
For $G_\lambda=S_\lambda^{-1/2}H$, the smallest singular value is at least $s_{\min}(H)/\sqrt{\|S_\lambda\|}$. The least-squares pseudoinverse formula therefore gives
\begin{equation*}\|\widehat\beta_\lambda\|
 \le\frac{\sqrt{(1+\kappa)/(1-\kappa)}}{s_{\min}(H)}
      \sqrt{\|\mathsf S\|/\eta}\,\|b\|
 \le c_2(1+\|Y\|^2),\end{equation*}
uniformly in strength. Finite fourth moments of the errors suffice for square-integrable fits and hence Theorem~\ref{thm:fs-signed-transfer}. They alone do not establish the derivative moment conditions used in the Gaussian results.

For completeness, symmetric clipping is the Frobenius projection onto the closed convex set $\{A=A^\top:-\kappa I\preceq A\preceq\kappa I\}$. Diagonalize the input: the projection inequality follows coordinatewise from scalar clipping and the bounds on the diagonal entries of every feasible comparison matrix. Applying this inequality to two projected inputs shows that the projection is nonexpansive. Thus clipping is Lipschitz, including at eigenvalue coincidences and thresholds.

For $S\succeq\eta I$, the differential $Z=D(S^{1/2})[D S]$ solves $S^{1/2}Z+ZS^{1/2}=D S$. In an eigenbasis,
\begin{equation*}Z_{ij}=(D S)_{ij}/(\sqrt{s_i}+\sqrt{s_j}),\qquad
 \|Z\|_F\le\|D S\|_F/(2\sqrt\eta).\end{equation*}
Inverse-root derivatives follow by differentiating the inverse. Combining these bounds with polynomial score derivatives proves local Lipschitzness and polynomial-growth weak derivative bounds for $E_p,\Delta_p,S_\lambda^{-1}$. Furthermore
$\lambda_{\min}(H^\top S_\lambda^{-1}H)\ge s_{\min}(H)^2/\|S_\lambda\|$, so its inverse and derivatives have polynomial bounds too. The actual fit and $v_\lambda=\partial_\lambda\widehat\beta_\lambda$ consequently have first weak response derivatives with polynomial-growth bounds, uniformly in $\lambda\in[0,1]$. Gaussian responses have all polynomial moments. This proves the required integrability and interchanges below for the implemented fixed-design pipeline. Constants can depend badly on the design and ridge; no useful uniform-in-$n$ rate is asserted.

\subsection{The raw covariance correction as a quadratic polynomial}

Centering $x_ix_i^\top t-x_if_i$ gives $B_it-a_i$ as defined in \eqref{eq:fs-target-quadratic}. For a unit vector $w$, put
\begin{align*}
 M(w)&=\operatorname{ave}_U(B_i^\top ww^\top B_i)-\operatorname{ave}_L(B_i^\top ww^\top B_i),\\
 l(w)&=\operatorname{ave}_U\{B_i^\top w(w^\top a_i)\}-\operatorname{ave}_L\{B_i^\top w(w^\top a_i)\},\\
 k(w)&=\operatorname{ave}_U(w^\top a_i)^2-\operatorname{ave}_L(w^\top a_i)^2.
\end{align*}
Expanding gives $w^\top K(t)w=t^\top M(w)t-2l(w)^\top t+k(w)$. If $e_2=\sup_w\|M(w)\|$, $e_1=\sup_w\|l(w)\|$ and $e_0=\sup_w|k(w)|$, over unit $w$, then
\begin{equation*}\|K(t)\|\le e_2\|t\|^2+2e_1\|t\|+e_0,\qquad
 \|D_tK(t)[v]\|\le(2e_2\|t\|+2e_1)\|v\|.\end{equation*}
This supplies a computable moment-based magnitude bound. It does not prove a favorable risk direction, and differentiating the protected correction additionally requires the whitening and unwhitening derivatives.

\subsection{Proof of the finite-sample rate bounds}
\label{app:scales-proof}

Proposition~\ref{prop:fs-scales} is a random-design statement. In affine regression, the coefficients of $m_2(\theta)$ are the differences of sample averages of $XX^\top$ and $Xf(X)$. Their population means cancel. Finite second moments of these coefficients and $N/n$ bounded below give the root-$n$ coefficient orders. Compact parameter sets make the affine coefficient bounds uniform:
\begin{equation*}m_2(\theta)=(A_U-A_L)\theta-
 \{\overline{Xf}_U-\overline{Xf}_L\}=O_p(n^{-1/2}).\end{equation*}
The same difference $A_U-A_L$ is its Jacobian. A bounded $\Pi_W$ proves the auxiliary penalty order. These conclusions do not follow merely from pointwise second moments for arbitrary nonlinear moment classes.

For part (b), whiten the affine fit as in Supplementary Section~\ref{app:path-proofs}. Stable covariance and design singular values give $\|L_Y\|=O_p(1)$, and the root-$n$ baseline residual gives $\|z_Y\|=O_p(n^{-1/2})$. The exact displacement is $\lambda L_Y(I+\lambda E_pP_Y)^{-1}E_pz_Y$. Its norm is at most $\|L_Y\|\gamma_n\|z_Y\|/(1-\gamma_n)$, uniformly on the strength interval.

For part (c), the exact affine coefficients are $v=L_YE_pz_Y$ and $w=-L_YE_pP_YE_pz_Y$. Expanding squared loss under the imposed integrable remainder gives
\begin{equation*}R(\lambda)-R(0)=2\lambda\mathbb E\langle e_0,v\rangle
 +\lambda^2\{\mathbb E\|v\|^2+2\mathbb E\langle e_0,w\rangle\}
 +O(\lambda^3).\end{equation*}
Cauchy--Schwarz and the stated $L^2$ bounds give $O(\gamma_n/n)$ and $O(\gamma_n^2/n)$. The second displacement contains two update factors. Neither these upper bounds nor signed cancellations identify a universal gain rate.

\subsection{Admissible caps}
\label{app:admissible-caps}

We prove Proposition~\ref{prop:fs-admissible-caps}. Recall $F=T^\top\mathsf S^{-1}T=\Omega^{-1}$, so $F^{1/2}=\Omega^{-1/2}$ and $\gamma_K=\|F^{1/2}KF^{1/2}\|_{\rm op}=\|\Omega^{-1/2}K\Omega^{-1/2}\|_{\rm op}$ by Proposition~\ref{prop:fs-amplitude}.

\emph{(i)} The whitened blend is $\mathsf S^{-1/2}(\mathsf S+\lambda TKT^\top)\mathsf S^{-1/2}=I+\lambda E$ with $E=\mathsf S^{-1/2}TKT^\top\mathsf S^{-1/2}$, and $\|E\|_{\rm op}=\gamma_K\le\kappa$ by Proposition~\ref{prop:fs-amplitude}(i). Hence $\|\lambda E\|_{\rm op}\le\kappa<1$ and the eigenvalues of $I+\lambda E$ lie in $[1-\kappa,1+\kappa]\subset(0,\infty)$, so the blend is positive definite.

\emph{(ii)} For $K(a,b)=aK_p+b\gamma\Omega$ the whitened update is $Q_T(aF^{1/2}K_pF^{1/2}+b\gamma F^{1/2}\Omega F^{1/2})Q_T^\top$. Since $F=\Omega^{-1}$ we have $F^{1/2}\Omega F^{1/2}=I$, so the two pieces have operator norms $|a|\gamma$ and $|b|\gamma$, and the update has norm at most $(|a|+|b|)\gamma$. Positive definiteness of $\Omega+K(a,b)=(1+b\gamma)\Omega+aK_p$ is equivalent to that of its $\Omega^{-1/2}$-conjugate $(1+b\gamma)I+a\Omega^{-1/2}K_p\Omega^{-1/2}$, which holds whenever $1+b\gamma>|a|\gamma$; this is implied by $\gamma<1/(|a|+|b|)$, because then $(|a|+|b|)\gamma<1$ and $|a|\gamma<1-|b|\gamma\le1+b\gamma$. At $(a,b)=(1/2,-1/2)$ and $\gamma=1$ the condition reads $1/2>1/2$, and when the negative extremal direction is attained the conjugate has a zero eigenvalue, so the borderline combination leaves the class.

\subsection{Realized-residual-span irrelevance}
\label{app:residual-span}

We prove Proposition~\ref{prop:fs-subspace-irrelevance}. Since $r_c(\theta)\in V$ for every $\theta$, we have $\Pi r_c(\theta)=r_c(\theta)$ and, using $\Pi=\Pi^\top=\Pi^2$,
\[r_c(\theta)^\top M^{-1}r_c(\theta)=(\Pi r_c(\theta))^\top M^{-1}(\Pi r_c(\theta))=r_c(\theta)^\top\Pi M^{-1}\Pi\,r_c(\theta),\]
which is \eqref{eq:fs-irrelevance}; the criterion \eqref{eq:fs-residual-criterion} uses $M=\Omega+K$ and therefore depends on the target only through that compression.

For the converse, let $\Delta M^{-1}=M_1^{-1}-M_2^{-1}$ and suppose \[r_c(\theta)^\top\Delta M^{-1}r_c(\theta)=0\] for all $\theta\in\Theta$. If $r_c(\theta)=J\theta-j$ with $\Theta$ containing a nonempty open subset, the left side is \[\theta^\top J^\top\Delta M^{-1}J\theta\allowbreak-2\theta^\top J^\top\Delta M^{-1}j+j^\top\Delta M^{-1}j\], so its vanishing forces $J^\top\Delta M^{-1}J=0$, $J^\top\Delta M^{-1}j=0$ and $j^\top\Delta M^{-1}j=0$, which is exactly $\Pi\Delta M^{-1}\Pi=0$ on $V=\operatorname{colspan}([J,j])$.

Equality of minimizers is weaker than equality of criteria: with scalar $c(\theta)=r_c(\theta)=\theta$, $V_c=1$, and residual inverse metrics $1$ and $2$, the criteria are $2\theta^2$ and $3\theta^2$. Both have the unique minimizer zero, although their compressions differ. The range hypothesis for the criterion converse also cannot be dropped. Take $\Theta$ with $r_c(\Theta)=\{(1,0),(0,1)\}$, $M_2^{-1}=I$ and $M_1^{-1}=I+t[[0,1],[1,0]]$ with $|t|<1$, so that both $M_i$ are positive definite. Then $\Delta M^{-1}=t[[0,1],[1,0]]$ vanishes at both realized residuals, while $\Pi\Delta M^{-1}\Pi=\Delta M^{-1}\ne0$ on their span. Two different metrics therefore give the same criteria on a thin realized set.

\subsection{Local risk response and the acceleration term}
\label{app:dependence}

Differentiate $J(\theta_\lambda)^\top W_\lambda\widehat g(\theta_\lambda)=0$ along a smooth minimizer branch. Since $\partial_\lambda W_\lambda=-W_\lambda\Delta W_\lambda$, the derivative at zero is $Hv-J(\theta_0)^\top\mathsf S^{-1}\Delta\mathsf S^{-1}\widehat g(\theta_0)=0$. This proves \eqref{eq:fs-dependence-displacement}. Expanding $\|e_0+\lambda v+\lambda^2w\|^2$ and using domination proves \eqref{eq:fs-dependence-risk}. The Hessian includes the second derivatives of the moments; in the affine case it reduces to $J^\top\mathsf S^{-1}J$.

If the update is independent of the full baseline map and centered, $v$ is linear in that independent update and $\mathbb E\langle e_0,v\rangle=0$. The acceleration is quadratic in the update and need not have zero error alignment. For example, take $\mathsf S=I_2$, $H=(1,0)^\top$, true parameter zero, centered $b=(b_1,b_2)^\top$ with covariance $[[1,.9],[.9,1]]$, and an independent $\Delta=kTT^\top$, $k=\pm.1$ equiprobably. The fit is $b_1+\lambda k b_2/(1+\lambda k)$. Its expected risk is
$\mathbb E_k[1+1.8c_k+c_k^2]$, $c_k=\lambda k/(1+\lambda k)$, and changes from $1$ to approximately $.9921233$ at strength one. The update is centered, independent, and within relative cap $.5$. Thus independence alone does not require deterioration. This is an affine weighting example, not a universal gain claim for the sample-derived calibration.

\subsection{Additional covariance targets and a benchmark strength}
\label{subsec:targets}

We give three alternative covariance targets. The cross-block target is
\begin{equation*}S_{0,n}^{\mathrm{cross}}=\operatorname{diag}(\widehat V,\widehat D).\end{equation*}
It retains within-block covariances and yields
\begin{equation*}\widehat S_\lambda^{\mathrm{cross}}=
 \begin{pmatrix}\widehat V&-(1-\lambda)\widehat R\\
 -(1-\lambda)\widehat R^\top&\widehat D\end{pmatrix}.\end{equation*}
Thus estimated cross covariance is discounted while the two diagonal blocks remain unchanged.

The spherical target is
\begin{equation*}S_{0,n}^{\mathrm{sph}}=\operatorname{diag}\left\{
 \frac{\operatorname{tr}(\widehat V)}d I_d,
 \frac{\operatorname{tr}(\widehat D)}d I_d\right\}.\end{equation*}
It changes both the cross block and within-block eigenvalue dispersion. This target depends on the units of moment coordinates; rescaling those coordinates can change the estimator.

The diagonal target is
\begin{equation*}S_{0,n}^{\mathrm{diag}}=\operatorname{diag}\{\operatorname{Diag}(\widehat V),\operatorname{Diag}(\widehat D)\},\end{equation*}
where $\operatorname{Diag}(K)$ extracts the diagonal into a diagonal matrix, and $\operatorname{diag}(K_1,K_2)$ denotes block concatenation. This target preserves marginal variances but depends on the coordinate axes.

Other choices include metric-based targets and their combinations. For the targets above, a sample-size benchmark strength is
\begin{equation*}\lambda_n=\min\{\lambda_{\max},c\,d/n\},\qquad 0<\lambda_{\max}<1,\quad c>0,\end{equation*}
with $c=1$ and $\lambda_{\max}=0.8$ as a concrete choice. This benchmark depends on dimension and labeled sample size. It does not estimate prediction quality or parameter risk; its first-order limit and finite-sample performance require separate analyses.

\subsection{Optional one-step computation}
\label{supp:one-step}

As an optional computational approximation to the estimator defined in \eqref{eq:family-estimator}, a single Gauss--Newton update from $t$ is
\begin{equation}\widehat\theta_{\lambda,\alpha}^{\mathrm{one}}=t-
 \{\widehat J(t)^\top W_{\lambda,\alpha}\widehat J(t)\}^{-1}
 \widehat J(t)^\top W_{\lambda,\alpha}\widehat g(t).\label{eq:one-step}\end{equation}
It equals the unconstrained solution for affine moments of full column rank. For nonlinear moments it is a separate finite-sample implementation; first-order equivalence requires the regularity and initialization conditions of the asymptotic analysis. Unless stated otherwise, the family denotes minimization of the frozen criterion, rather than this single update.

\begin{proposition}[One-step first-order equivalence]\label{prop:one-step}
Suppose (R2)--(R4), the joint moment limit, and $\sqrt M(t-\theta_*)=O_p(1)$ hold. The one-step update in \eqref{eq:one-step} has the expansion \eqref{eq:influence-general}. If a full minimizer also satisfies Theorem~\ref{thm:general-weight}, the update and that minimizer differ by $o_p(M^{-1/2})$.
\end{proposition}

\begin{proof}
Uniform Jacobian convergence and the pilot rate imply $\widehat g(t)=\widehat g(\theta_*)+G(t-\theta_*)+o_p(M^{-1/2})$. The update coefficient converges to $K(W)$, and $K(W)G=I_d$. Substitution cancels the pilot error to first order and gives the claimed expansion.
\end{proof}

For regular logistic fits, Corollary~\ref{cor:logistic-one-step-rate} strengthens this comparison to an $O_p(n^{-1})$ distance under its additional explicit conditions. For nonlinear moments, first-order equivalence allows finite-sample differences between the two procedures. Risk comparisons also require moment bounds beyond the distributional expansion. The finite-sample analysis therefore uses the specified estimator and solver throughout.

This supplement gives the details underlying Sections~\ref{sec:setup}--\ref{sec:finite-sample}. Expectations in the asymptotic analysis are conditional on the external predictor. Finite-sample regression expectations additionally condition on $X,U,f_L,f_U$. Vector norms are Euclidean and matrix norms without a subscript are operator norms. A covariance consistency assertion is distinct from a risk assertion. Classical GMM and extremum-estimation foundations are provided by \citep{hansen1982large,neweymcfadden1994}; the derivations below specialize them to the stated moment structure and frozen, possibly same-data weights.

Equivalent protected update.
\begin{equation*}F=T^\top\mathsf S^{-1}T,\end{equation*}
and apply the same eigenvalue clipping to obtain
\begin{equation}K_p=F^{-1/2}\operatorname{clip}_{[-\kappa,\kappa]}\bigl(F^{1/2}(C_U-C_L)F^{1/2}\bigr)F^{-1/2},
 \qquad \Delta_p=TK_pT^\top.\label{eq:fs-target-construction}\end{equation}
This equals \eqref{eq:ccppi-spectral-protection} because $\mathsf S^{-1/2}TF^{-1/2}$ has orthonormal columns and clipping maps zero to zero. In particular, protection retains the channel-preserving structure of the raw update.

The calibration of $\widehat S_{\lambda_{\rm CC}}$ is designed for GMM parameter estimation; minimizing covariance-estimation error is a different objective. Section~\ref{subsec:calibrated-target} formalizes the two principles defining the raw target. With the pool fixed to $C_U$, they determine the raw target, and spectral protection completes the method specification. Under the regular conditions of Section~\ref{sec:asymptotic}, $C_U-C_L\to_p0$, so the calibration vanishes and CC-PPI retains first-order efficiency. Section~\ref{sec:finite-sample} studies its finite-sample parameter risk, which depends on the design and predictor.

\subsection{A worked linear-regression example}\label{supp:worked-linear}
Consider a scalar regression without an intercept and a supplied linear predictor $f(x)=\alpha x$. Write $a_L=n^{-1}\sum_iX_i^2>0$, $a_U=N^{-1}\sum_kU_k^2$, $b_X=a_U-a_L$, and let $t$ be the supervised OLS slope. The sample moments reduce to
\begin{equation*}
 m_1(\beta)=a_L(\beta-t),\qquad m_2(\beta)=b_X(\beta-\alpha).
\end{equation*}
The auxiliary moment depends on the empirical second-moment discrepancy and the difference between the candidate slope and the supplied prediction slope.

For any frozen weight, abbreviate $w=W_{11}$, $\ell=\Lambda_W$ and $\pi=\Pi_W$, and put $c_W=a_L+\ell b_X$. Completing the square as above and differentiating in $\beta$ gives
\begin{equation}
 \widehat\beta(W)=\alpha+
 \frac{w c_Wa_L}{w c_W^2+\pi b_X^2}(t-\alpha).
 \label{eq:worked-linear-fit}
\end{equation}
The denominator is positive because $W\succ0$ and the stacked design has full column rank. In \eqref{eq:worked-linear-fit}, $\ell$ modifies the corrected-score coefficient and $\pi b_X^2$ contributes the auxiliary penalty.

To obtain CC-PPI, evaluate the pilot scores $s_i=X_i(X_it-Y_i)$, $h_i=X_i^2(t-\alpha)$ and $h_k^U=U_k^2(t-\alpha)$. Form their centered covariances, protect the update $C_U-C_L$, and use $W=\mathsf S_{1/2}^{-1}$ in \eqref{eq:worked-linear-fit}. The uncalibrated comparison uses $W=\mathsf S^{-1}$. Thus the observations, pilot and predictions are identical in the two fits; calibration can change the coefficient multiplying $t-\alpha$ through the estimated metric. Because the pilot scores depend on the observed responses, that coefficient is response dependent as well.

If $C_U=C_L$, CC-PPI equals uncalibrated GMM. If $a_U=a_L$, the scalar auxiliary moment vanishes and every positive definite frozen weight gives $t$. Otherwise calibration can change the coefficient on $t-\alpha$, although a nonzero update may leave the fit unchanged in special configurations. The multivariate counterpart is \eqref{eq:affine-solution}.

\section{First-order theory}
\label{app:asymptotic-proofs}

Conditional on the externally constructed predictor $f$, we compare limiting parameter covariances in the Loewner order for the moments in Section~\ref{sec:setup}. Efficiency is defined relative to the classical GMM bound \citep{hansen1982large}, with auxiliary information fixed.

\subsection{Regularity and the joint moment limit}
\label{subsec:regularity}

Throughout this section, $d$ and $f$ are fixed, $n,N\to\infty$, and $n/N\to r\in[0,\infty)$. With the normalization of Section~\ref{subsec:covariance}, write $c_1=\lim M/n>0$ and $c_2=\lim M/N\geq0$. Thus $c_2/c_1=r$. All expectations and probability limits below are conditional on $f$ when it is random. Define
\begin{equation*}\mu(\theta)=\mathbb E\{\psi(X,Y;\theta)\},\qquad
 g(\theta)=\begin{pmatrix}\mu(\theta)\\0\end{pmatrix},\qquad
 A=\left.\frac{\partial\mu(\theta)}{\partial\theta^\top}\right|_{\theta=\theta_*},\qquad
 G=\begin{pmatrix}A\\0\end{pmatrix}.\end{equation*}
The zero lower block follows from the identical covariate distributions; it is not a statement that the sample derivative $A_5$ vanishes.

We use the following regular conditions.
\begin{enumerate}
 \item[(R1)] $\Theta$ is compact, $\theta_*$ is an interior point, and $\mu$ is continuous with unique zero $\theta_*$. The uniform law $\sup_{\theta\in\Theta}\|\widehat g(\theta)-g(\theta)\|\overset p\to0$ holds. Compactness may be replaced by another condition that yields consistency of the selected minimizer.
 \item[(R2)] On an open neighborhood of $\theta_*$, the sample moments are continuously differentiable and their Jacobians converge uniformly in probability to the continuous population Jacobian. Differentiation under expectation is valid, and $A$ is nonsingular.
 \item[(R3)] $\psi(X,Y;\theta_*)$ and $h(X;\theta_*)$ have finite second moments. Their joint covariance gives
 \begin{equation}S=\begin{pmatrix}V&-R\\-R^\top&D\end{pmatrix}\succ0,\qquad
 V=c_1\Sigma_{11},\quad R=c_1\Sigma_{1h},\quad
 D=(c_1+c_2)\Sigma_{hh}.\label{eq:S-limit}\end{equation}
 \item[(R4)] The frozen weight satisfies $W_n\overset p\to W\succ0$. The estimator is a measurable global minimizer of its frozen criterion. A numerical approximation is also allowed if it satisfies $Q_n(\widehat\theta)\leq\inf_{\Theta}Q_n+o_p(1)$ and $\widehat J(\widehat\theta)^\top W_n\widehat g(\widehat\theta)=o_p(M^{-1/2})$.
\end{enumerate}
These fixed-dimension conditions require a positive definite limiting moment covariance. Sequences with singular limits need a different analysis.

Under (R3), the independent-sample central limit theorem gives
\begin{equation}\sqrt M\,\widehat g(\theta_*)\Rightarrow N(0,S).\label{eq:moment-clt}\end{equation}
Indeed, the labeled contribution is the centered average of $(\psi^\top,-h^\top)^\top$, and the unlabeled contribution is the centered average of $(0^\top,h^\top)^\top$. Their variances add to \eqref{eq:S-limit}. If $c_2=0$, the unlabeled fluctuation vanishes in probability by its second moment bound.

The consistency of the covariance estimator requires a further condition when it is evaluated at the random pilot $t$. We assume $t\overset p\to\theta_*$, uniform laws for the score first and second products near $\theta_*$, continuity of their expectations, and $\varepsilon_M\to0$. These imply $\widehat S\overset p\to S$. A sufficient route is continuity of the scores together with square-integrable local envelopes for $\psi$ and $h$. This covariance condition is imposed wherever the feasible inverse-covariance weight or a plug-in covariance is used.

\subsection{Asymptotic distribution for a general weight}
\label{subsec:general-weight}

\begin{theorem}[General frozen weight]\label{thm:general-weight}
Under (R1)--(R4), the estimator is consistent and
\begin{align}
 \sqrt M(\widehat\theta-\theta_*)&=-K(W)\sqrt M\,\widehat g(\theta_*)+o_p(1),\label{eq:influence-general}\\
 K(W)&=(G^\top WG)^{-1}G^\top W,\nonumber\\
 \sqrt M(\widehat\theta-\theta_*)&\Rightarrow N\{0,\mathcal V(W)\},\nonumber\\
 \mathcal V(W)&=K(W)SK(W)^\top.
\nonumber\end{align}
The weight may be estimated from the same inference samples; independence between $W_n$ and the moments is not required.
\end{theorem}

The detailed proof is given below.

Use the $M$-scaled covariance $S=[V,-R;-R^\top,D]$ and its blocks from \eqref{eq:S-limit}. For a limiting weight $W$, retain $K(W)=(G^\top WG)^{-1}G^\top W$ and $\mathcal V(W)=K(W)SK(W)^\top$. A correction $L$ has influence matrix $A^{-1}(I_d,L)$. The unrestricted optimum is\begin{equation}L(W)=W_{11}^{-1}W_{12},\qquad
 L_*=RD^{-1},\qquad C=V-RD^{-1}R^\top.\label{eq:optimal-correction}\end{equation}
\begin{align}
 K(W)&=A^{-1}\begin{pmatrix}I_d&L(W)\end{pmatrix},\label{eq:K-block}\\
 \mathcal V_*&=A^{-1}CA^{-\top},\nonumber\\
 \mathcal V(W)-\mathcal V_*&=
 A^{-1}\{L(W)-L_*\}D\{L(W)-L_*\}^\top A^{-\top}.\label{eq:correction-gap}
\end{align}
In particular, a scalar correction restricts $L$ to $\ell I_d$, whereas a coordinatewise correction restricts it to a diagonal matrix. They attain the full-matrix bound when these classes contain $L_*=RD^{-1}$. Otherwise the displayed positive semidefinite gap identifies the efficiency lost by the restriction; individual contrasts can still have equality. The comparison holds the auxiliary score information fixed.

Full-matrix combination determines first-order efficiency, whereas covariance calibration changes finite-sample risk within the common efficient class. The matrix identity and equality conditions are proved in Supplementary Section~\ref{subsec:loewner}.

\subsection{Loewner optimality and equality conditions}
\label{subsec:loewner}

\begin{theorem}[Optimal parameter covariance]\label{thm:loewner}
For $S\succ0$ and $G$ of full column rank, define
\begin{equation*}K_*=(G^\top S^{-1}G)^{-1}G^\top S^{-1},\qquad
 \mathcal V_*=(G^\top S^{-1}G)^{-1}.\end{equation*}
For every symmetric positive definite $W$,
\begin{equation}\mathcal V(W)-\mathcal V_*
 =\{K(W)-K_*\}S\{K(W)-K_*\}^\top\succeq0.\label{eq:loewner-gap}\end{equation}
Equality of the covariance matrices holds if and only if $K(W)=K_*$. For a specified contrast $a^\top\theta$, equality of asymptotic variances holds if and only if $a^\top\{K(W)-K_*\}=0$.
\end{theorem}

The detailed proof is given below.

The unrestricted matrix class attains the bound. Scalar and diagonal restrictions can attain it only when their influence maps agree with the unrestricted optimum. Otherwise the covariance gap is positive semidefinite, though it may vanish for some contrasts. Several weight matrices can therefore attain the same bound.

Equations~\eqref{eq:optimal-correction}--\eqref{eq:correction-gap} give the block correction and covariance gap, with $A^{-\top}=(A^{-1})^\top$. First-order efficiency depends on the limiting correction coefficient; the finite-sample criterion depends on the full weight.

A scalar or diagonal \emph{stacked GMM weight} has $W_{12}=0$ and hence $L(W)=0$, giving the supervised first-order covariance. A scalar or diagonal \emph{correction coefficient} instead restricts $L$ to $\ell I_d$ or to a diagonal matrix. For any such coefficient $L$, the asymptotic covariance of a regular estimator with influence map $A^{-1}(I_d,L)$ obeys \eqref{eq:correction-gap}. The full correction attains equality; a restricted correction does so exactly when it can represent $RD^{-1}$. The auxiliary information is held fixed throughout this comparison.

The nonuniqueness is explicit. For any $P,T\succ0$, the matrix
\begin{equation}W(P,T)=\begin{pmatrix}
 P&PL_*\\L_*^\top P&L_*^\top PL_*+T
 \end{pmatrix}\succ0\label{eq:efficient-weight-class}\end{equation}
has $L(W)=L_*$ and attains $\mathcal V_*$. Its Schur complement is $T$. Different choices may give different finite-sample criteria, without changing first-order efficiency. Equation~\eqref{eq:efficient-weight-class} describes efficient limiting weights in the general GMM class; it does not assert that every such weight belongs to a prespecified collection of shrinkage targets in Section~\ref{subsec:family}.

\subsection{Efficiency-preserving covariance weights}
\label{subsec:preserve-efficiency}

\begin{corollary}[A common inverse-covariance limit]\label{cor:common-limit}
Suppose the covariance consistency condition holds and the selected family member satisfies
\begin{equation}\Delta_n=\lambda_n\{S_{0,n}(\alpha_n)-\widehat S\}
 \overset p\longrightarrow0.\label{eq:vanishing-perturbation}\end{equation}
Then $\widehat S_{\lambda_n,\alpha_n}\overset p\to S$ and $W_{\lambda_n,\alpha_n}\overset p\to S^{-1}$. Under the remaining regular conditions, its limiting parameter covariance is $\mathcal V_*$. Any two such estimators computed from the same data satisfy
\begin{equation}\sqrt M\{\widehat\theta^{(a)}-\widehat\theta^{(b)}\}=o_p(1).\label{eq:first-order-equivalence}\end{equation}
\end{corollary}

\begin{proof}
The blended covariance is $\widehat S+\Delta_n$. Continuity of inversion at $S\succ0$ gives the weight limit. Theorem~\ref{thm:general-weight} gives the common influence map $K_*$. Subtracting the two expansions proves \eqref{eq:first-order-equivalence}.
\end{proof}

Condition~\eqref{eq:vanishing-perturbation} is satisfied if $\lambda_n\overset p\to0$ and the selected target is $O_p(1)$. Boundedness must hold for the \emph{selected} target when its index is data dependent. No independence of the tuning quantities is needed. Alternatively, if $S_{0,n}(\alpha_n)\overset p\to S$, every $\lambda_n\in[0,1]$ satisfies the condition, even when $\lambda_n$ does not converge to zero. If the target diverges, vanishing $\lambda_n$ alone is insufficient.

For the cross-block, spherical and diagonal targets in Supplementary Section~\ref{subsec:targets}, covariance consistency gives $O_p(1)$ targets. At fixed $d$, $\lambda_n=\min\{0.8,d/n\}$ preserves $S^{-1}$ and $\mathcal V_*$ when $\varepsilon_M\to0$. This establishes efficiency preservation only; increasing dimension or a fixed nonvanishing ridge requires another limit analysis.

Sharing $S^{-1}$ is sufficient but not necessary for sharing the efficiency bound. If $W_n\overset p\to\overline W\succ0$, equations~\eqref{eq:K-block}--\eqref{eq:correction-gap} show that efficiency is preserved precisely when
\begin{equation}\overline W_{11}^{-1}\overline W_{12}=RD^{-1}.\label{eq:efficient-limit-condition}\end{equation}
As a concrete distinction, if the cross-block target is used and $\lambda_n\overset p\to\lambda_0$, then the limiting covariance has cross block $-(1-\lambda_0)R$. Its inverse satisfies $L(\overline W)=(1-\lambda_0)RD^{-1}$, so
\begin{equation}\mathcal V(\overline W)-\mathcal V_*
 =\lambda_0^2 A^{-1}RD^{-1}R^\top A^{-\top}.\label{eq:fixed-cross-gap}\end{equation}
Thus a nonvanishing cross-block shrinkage loses first-order efficiency when $R\ne0$, although some contrasts may retain equality. If $R=0$, it preserves the bound. For other target collections, the equality test is \eqref{eq:efficient-limit-condition}, rather than a universal requirement that $\lambda_n\to0$.

\subsection{Feasible inference}
\label{subsec:inference}

Under the covariance consistency condition, let $\widehat G=\widehat J(t)$ with $t\overset p\to\theta_*$, and let $\widehat S_{\rm raw}=\widehat S-\varepsilon_M I_{2d}$ be the unregularized $M$-scaled empirical joint moment covariance. The ridge remains in the covariance used to construct the frozen weight $W_n$. Define
\begin{equation}\widehat K=(\widehat G^\top W_n\widehat G)^{-1}\widehat G^\top W_n,
 \qquad \widehat{\mathcal V}=\widehat K\widehat S_{\rm raw}\widehat K^\top.\label{eq:feasible-sandwich}\end{equation}
Then $\widehat{\mathcal V}\overset p\to\mathcal V(W)$. The middle matrix is the consistent estimate of the \emph{actual joint moment covariance}, rather than an arbitrary shrinkage target. For efficient family members, $(\widehat G^\top\widehat S^{-1}\widehat G)^{-1}$ is also consistent for $\mathcal V_*$. A Wald interval for a fixed nonzero contrast is
\begin{equation*}a^\top\widehat\theta\ \pm\ z_{1-\eta/2}
 \sqrt{a^\top\widehat{\mathcal V}a/M},\qquad \eta\in(0,1).\end{equation*}
Under the stated regular conditions, this interval has asymptotic coverage $1-\eta$, conditional on the externally trained predictor $f$ when it is random, with probability taken over the inference samples. This is not a coverage statement conditional on the realized inference designs. Replacing $\widehat S_{\rm raw}$ by $\widehat S$ adds $\varepsilon_M\widehat K\widehat K^\top$ to $\widehat{\mathcal V}$ and is first-order equivalent when $\varepsilon_M\to0$. The finite-sample containment argument in Theorem~\ref{thm:balanced-wald}, however, applies specifically to the unregularized middle matrix in \eqref{eq:feasible-sandwich}.

The optional one-step approximation is first-order equivalent under the pilot conditions in Supplementary Section~\ref{supp:one-step}. The estimators in this paper are defined by the frozen-criterion minimization in Section~\ref{subsec:implementation}.

Supplementary Section~\ref{app:asymptotic-proofs} gives detailed proofs of random-pilot covariance consistency, the influence expansion, all equality conditions, and the inference and one-step statements; further regular extremum-estimation background is provided by \citep{neweymcfadden1994}.
\subsection{Identification, covariance, and the joint central limit theorem}

Identical covariate marginals imply $\mathbb E m_2(\theta)=0$ for each fixed $\theta$. Under valid differentiation under the expectation, its population derivative is zero. This does not set its observed sample derivative to zero. At $\theta_*$, let
\begin{equation*}a_i=\begin{pmatrix}\psi_i\\-(h_i-\mathbb Eh)\end{pmatrix},\qquad
 b_j=\begin{pmatrix}0\\h_j^U-\mathbb Eh\end{pmatrix}.\end{equation*}
Then $\widehat g(\theta_*)=n^{-1}\sum_i a_i+N^{-1}\sum_j b_j$. Independence adds their covariances and gives \eqref{eq:population-covariance}. The two ordinary multivariate sample-mean central limit theorems, multiplied by $\sqrt{M/n}$ and $\sqrt{M/N}$, give \eqref{eq:moment-clt}. If $M/N\to0$, the second contribution tends to zero in $L^2$. All means and cross-block signs in this calculation use the sign convention of \eqref{eq:m2}.

At a random consistent pilot $t$, let $\widehat F_n(\theta)$ be any empirical score mean or score second-product matrix used in the covariance. The imposed uniform laws imply
\begin{equation*}\|\widehat F_n(t)-F(\theta_*)\|
 \le\sup_{\theta\in\mathcal N}\|\widehat F_n(\theta)-F(\theta)\|
       +\|F(t)-F(\theta_*)\|=o_p(1),\end{equation*}
on $t\in\mathcal N$, an event with probability tending to one. Centered covariance consistency follows by subtracting products of the convergent means. The vanishing ridge then gives $\widehat S\to_p S$. This argument uses uniform laws, not independence between $t$ and labeled scores.

\subsection{Consistency and the influence expansion}

Under (R1) and $W_n\to_p W$, $\widehat g$ is uniformly bounded in probability on $\Theta$ and
\begin{equation*}\sup_\Theta|\widehat g^\top W_n\widehat g-g^\top Wg|=o_p(1).\end{equation*}
The continuous limiting criterion has unique minimizer $\theta_*$: its value is zero there and positive elsewhere. Compactness gives a positive separation outside every neighborhood. Uniform convergence and the $o_p(1)$ objective allowance in (R4) therefore imply consistency, as in the extremum argument of \citep{neweymcfadden1994}.

Write $\delta_n=\widehat\theta-\theta_*$ and
$\overline J_n=\int_0^1\widehat J(\theta_*+s\delta_n)ds$. The segment is eventually inside the regular neighborhood. The fundamental theorem of calculus gives $\widehat g(\widehat\theta)=\widehat g(\theta_*)+\overline J_n\delta_n$, and (R2) and consistency give $\overline J_n\to_p G$ and $\widehat J(\widehat\theta)\to_p G$. The first-order residual condition yields
\begin{equation*}[\widehat J(\widehat\theta)^\top W_n\overline J_n]\delta_n
 =-\widehat J(\widehat\theta)^\top W_n\widehat g(\theta_*)+o_p(M^{-1/2}).\end{equation*}
The bracket converges to the nonsingular $G^\top WG$. Thus $\delta_n=O_p(M^{-1/2})$, and replacing convergent coefficients gives Theorem~\ref{thm:general-weight}. In particular, dependence of $W_n$ on $\widehat g$ creates no first-order term when $W_n-W=o_p(1)$; this does not make it negligible at finite sample sizes.

\subsection{Loewner ordering and equality conditions}

Set $D_K=K(W)-K_*$. Since $K(W)G=K_*G=I$ and $K_*S=\mathcal V_*G^\top$, $D_KSK_*^\top=0$. Expanding gives $\mathcal V(W)=\mathcal V_*+D_KSD_K^\top$. The last term equals $(D_KS^{1/2})(D_KS^{1/2})^\top$. Hence it is zero exactly when $D_K=0$; for one contrast its quadratic form is zero exactly when $a^\top D_K=0$. This proves all statements in Theorem~\ref{thm:loewner}, including why a nonzero covariance gap need not be positive definite.

For $G=(A^\top,0^\top)^\top$,
\begin{equation*}(A^\top W_{11}A)^{-1}A^\top W_{11}=A^{-1},\qquad
 K(W)=A^{-1}(I,L),\quad L=W_{11}^{-1}W_{12}.\end{equation*}
Direct covariance multiplication gives
\begin{equation*}K(W)SK(W)^\top=A^{-1}
   \{V-RL^\top-LR^\top+LDL^\top\}A^{-\top}.\end{equation*}
Completing the square at $L_*=RD^{-1}$ proves \eqref{eq:correction-gap}. Every efficient positive definite limiting weight has $W_{12}=W_{11}L_*$; its Schur complement is positive definite. Conversely these two properties generate exactly \eqref{eq:efficient-weight-class}. A scalar or diagonal correction can attain equality only when it represents $L_*$, whereas a scalar or diagonal stacked weight has $L=0$. These are algebraic restrictions, not claims about different auxiliary functions.

\subsection{Perturbations, inference, and one-step computation}

For the selected target, $\widehat S_{\lambda_n,\alpha_n}=\widehat S+\Delta_n\to_p S\succ0$. Continuity of inversion gives $W_n\to_p S^{-1}$. Applying the influence expansion twice to the same moments cancels the leading random term and proves \eqref{eq:first-order-equivalence}. For a general weight limit, the previous square-completion formula gives the equality test \eqref{eq:efficient-limit-condition}. For a fixed cross-block blend, its limiting covariance cross block is $-(1-\lambda_0)R$; block inversion gives $L=(1-\lambda_0)RD^{-1}$ and then \eqref{eq:fixed-cross-gap}.

The covariance consistency argument and $\widehat J(t)\to_p G$ yield $\widehat K\to_p K(W)$. Continuous multiplication gives the sandwich limit in \eqref{eq:feasible-sandwich}. Positive definiteness and full rank ensure nonzero limiting variance for any nonzero fixed contrast, so Slutsky's theorem yields the Wald coverage statement. No finite-sample coverage claim follows.

For the one-step result, $\sqrt M(t-\theta_*)=O_p(1)$ and uniform Jacobian convergence give $\widehat g(t)=\widehat g(\theta_*)+G(t-\theta_*)+o_p(M^{-1/2})$. The update coefficient tends to $K(W)$ and its product with $G$ tends to $I$. Substitution gives the same expansion as the full minimizer. For affine moments this cancellation is exact at finite samples because the Jacobian is constant. For nonlinear moments the equality is first-order only.

\subsection{Proof of conditional finite-sample separation}
\label{supp:finite-separation}

\begin{proposition}[Conditional finite-sample separation]\label{prop:finite-separation}
Let the common covariate law assign probabilities $1/2,1/4,1/4$ to $0,\sqrt2,-\sqrt2$, respectively. Let $Y=\beta X+\epsilon$, with independent $\epsilon\sim N(0,\sigma^2)$, $\sigma>0$, and fixed $f(x)=\alpha x$, $\mu=\beta-\alpha\ne0$. For independent iid labeled and unlabeled samples with $N=5n$, compare the unconstrained full frozen-criterion estimators $G_\tau$ with $W_\tau=\operatorname{diag}(1,\tau)$, $1\le\tau\le4$, and the one-step corrected estimator $P=t$, where $t$ is OLS and $W_{12}=0$. On a zero-rank labeled sample define all returns as $\alpha$.

For every fixed $\tau$, all these estimators have the same efficient first-order limit
\[
 \sqrt n(G_\tau-\beta)\Rightarrow N(0,\sigma^2),\qquad
 \sqrt n(P-\beta)\Rightarrow N(0,\sigma^2).
\]
For $n$ divisible by four, condition instead on a finite design with half the labeled rows and one quarter of the unlabeled rows nonzero. Then $A_L=1$, $A_U=1/2$, and exactly
\begin{equation*}
 \frac{R_X(G_\tau)}{R_X(P)}
   =\frac{1+\tau^2 n\mu^2/(16\sigma^2)}{(1+\tau/4)^2},
 \qquad R_X(P)=\frac{\sigma^2}{n},
\end{equation*}
where $R_X$ conditions on both designs. In particular,
\[
 \frac{R_X(G_1)}{R_X(P)}=.64+.04\frac{n\mu^2}{\sigma^2},\qquad
 \frac{R_X(G_4)}{R_X(P)}=.25+.25\frac{n\mu^2}{\sigma^2}.
\]
For every prescribed $M>0$, a sufficiently large finite $n$ in this conditional-design family has $R_X(G_4)/R_X(P)>M$, although $|X|\le\sqrt2$, $|U|\le\sqrt2$, $A_L,A_U\in[1/2,1]$ and the weight condition number is at most four.
\end{proposition}

Write $L=A_L$, $B=A_U-A_L$ and $v=t-\alpha$. The scalar moments are $m_1=L(\theta-t)$ and $m_2=B(\theta-\alpha)$. Full minimization with the positive definite weight $\operatorname{diag}(1,\tau)$ gives
\begin{equation*}
 G_\tau=\alpha+\phi_\tau(t-\alpha),\qquad
 \phi_\tau=\frac{L^2}{L^2+\tau B^2}.
\end{equation*}
The one-step correction has $\Lambda_W=0$ and therefore equals $t$. Conditional on a nonzero-rank labeled design, $t-\beta\sim N\{0,\sigma^2/(nL)\}$. Hence
\begin{equation*}
 R_X(G_\tau)=\phi_\tau^2\frac{\sigma^2}{nL}+(1-\phi_\tau)^2\mu^2.
\end{equation*}
On the stated event, $L=1$, $B=-1/2$, so $\phi_\tau=(1+\tau/4)^{-1}$. Substitution proves the ratio, absolute difference and numerical examples. Any $n$ divisible by four with $n\mu^2/\sigma^2>4M-1$ makes the $G_4$ ratio exceed $M$ (when $4M-1<0$, any positive $n$ suffices). All event probabilities are positive because the corresponding binomial counts have positive mass.

Under the original common-law iid sampling, $L\to_p1$ and $B=O_p(n^{-1/2})$, since $X^2$ is bounded. Also $t-\alpha=\mu+O_p(n^{-1/2})$. Uniformly for $1\le\tau\le4$, $G_\tau-t=-(1-\phi_\tau)(t-\alpha)=O_p(n^{-1})$, on a labeled design whose rank is positive with probability tending to one. The zero-rank event has probability $2^{-n}$ and does not alter weak convergence. The OLS CLT gives the common limit.

At truth, the scores are $s=-X\epsilon$ and $h=X^2\mu$. They have $\operatorname{Var}(s)=\sigma^2$, $\operatorname{Var}(h)=\mu^2$ and $\operatorname{Cov}(s,h)=0$. The nondegenerate stacked covariance is therefore $\operatorname{diag}\{\sigma^2,(1+1/5)\mu^2\}$, and the optimal correction is zero. All displayed limits attain the GMM efficiency bound for these specified moments. We do not obtain this claim by equating the full weights: their first-order influence maps coincide, although the weights differ.

The conditional event retains a fixed empirical second-moment gap and its probability tends to zero by the law of large numbers. No assertion is made about its frequency at a practical sample size, or about unconditional risk ratios. In particular, the example does not assert that empirical covariance weighting or CC-PPI produces the specified frozen weights. Its purpose is to quantify what common first-order efficiency alone cannot rule out.

\subsection{Numerical equivalence for regular logistic fits}
\label{supp:logistic-solver-equivalence}
This subsection supplies sufficient computational conditions for logistic fits under the $n$-scaled criterion $Q_n=\widehat g^\top W_n\widehat g/2$. Assume fixed dimension, $n,N\to\infty$, $n/N\to\rho\in[0,\infty)$, iid labeled pairs and an independent iid unlabeled sample with the same covariate marginal. Responses are binary and the fixed external predictor takes values in $[0,1]$. Suppose $E\|X\|^3<\infty$, a finite $\theta_*$ solves the population supervised equation, $A_*=E[XX^\top\sigma'(X^\top\theta_*)]\succ0$, the pilot satisfies $\sqrt n(t_n-\theta_*)=O_p(1)$, and the frozen weight satisfies $W_n\to_p W\succ0$. Pilot failure has probability tending to zero. Statistics are extended finitely and measurably on exceptional events; this convention does not authorize deleting failed experimental trials.

\begin{proposition}[Accuracy conditions for numerical logistic fits]\label{prop:logistic-solver-equivalence}
Let deterministic $\delta_n\ge0$ tend to zero and $\tau_n>0$ satisfy $\sqrt n\tau_n\to0$. Suppose a measurable return $z_n$ satisfies, with probability tending to one,
\[
 Q_n(z_n)\le Q_n(t_n)+\delta_n,\qquad
 \|J(z_n)^\top W_n\widehat g(z_n)\|\le\tau_n.
\]
Then with probability tending to one the frozen criterion on $\mathbb R^d$ has a unique finite global minimizer $\widehat\theta_n$, and
\[
 \|z_n-\widehat\theta_n\|=O_p(\tau_n),\qquad
 \sqrt n(z_n-\theta_*)=-(G^\top WG)^{-1}G^\top W\sqrt n\widehat g(\theta_*)+o_p(1),
 \quad G=\begin{pmatrix}A_*\\0\end{pmatrix}.
\]
The exact minimizer is defined by a finite measurable extension when uniqueness or existence fails. With the additional regular covariance conditions and $W=\Omega^{-1}$, this is the efficient limit of Theorem~\ref{thm:new-efficiency}.
\end{proposition}
\begin{proof}
The moment assumption and bounded responses and predictions give the joint moment CLT, $\widehat g(\theta_*)=O_p(n^{-1/2})$, and finite limits for empirical second- and third-moment envelopes. Thus $j,H_g,h_L=O_p(1)$ in Proposition~\ref{prop:logistic-global-certificate}. The LLN and the Lipschitz bound for $A$ give $A(t_n)\to_p A_*$ and $J(t_n)\to_p G$; Taylor expansion gives $\widehat g(t_n)=O_p(n^{-1/2})$.

Set $r_n=(n^{-1}+\delta_n)^{1/4}$. Then $n^{-1/2}=o(r_n)$ and $n^{-1}+\delta_n=o(r_n^2)$. Apply the global certificate at $t_n$. Its $a_r\to_p\lambda_{\min}(A_*)$, $L_r/r_n\to_p\lambda_{\min}(A_*)$, $b_r=o_p(1)$, and $m_r\to_p\lambda_{\min}(W)\lambda_{\min}(A_*)^2>0$. The outside lower bound has order $r_n^2$ and exceeds $Q_n(t_n)+\delta_n=O_p(n^{-1})+\delta_n$ with probability tending to one. Hence the global minimizer exists uniquely and every accepted return lies in the same ball. Both are consistent. Strong convexity gives $m_r\|z_n-\widehat\theta_n\|\le\tau_n$.

The interior first-order condition, with $\overline J_n$ the integrated Jacobian between $\theta_*$ and $\widehat\theta_n$, gives
\[
 [J(\widehat\theta_n)^\top W_n\overline J_n](\widehat\theta_n-\theta_*)=-J(\widehat\theta_n)^\top W_n\widehat g(\theta_*).
\]
Both Jacobians tend to $G$, and $G^\top WG$ is invertible. This first gives the root-$n$ rate and then the influence expansion. The numerical distance is $o_p(n^{-1/2})$, proving the result. Also $0\le Q_n(z_n)-\inf Q_n\le Q_n(t_n)+\delta_n=o_p(1)$, supplying the objective allowance in (R4).
\end{proof}

\begin{corollary}[One-step distance to the exact logistic fit]\label{cor:logistic-one-step-rate}
Under the same sampling, pilot and weight assumptions, the frozen Gauss--Newton update
\[
 c_n=t_n-(J_t^\top W_nJ_t)^{-1}J_t^\top W_n\widehat g(t_n)
\]
satisfies $\|c_n-\widehat\theta_n\|=O_p(n^{-1})$ and has the same first-order limit.
\end{corollary}
\begin{proof}
The update is $O_p(n^{-1/2})$. Its affine gradient vanishes exactly. With $d_n=c_n-t_n$, the full gradient equals
\begin{align*}
 &(J(c_n)-J_t)^\top W_n\{\widehat g(t_n)+J_td_n\}\\
 &\quad+J(c_n)^\top W_n\{\widehat g(c_n)-\widehat g(t_n)-J_td_n\}=O_p(n^{-1}),
\end{align*}
using the $O_p(1)$ second-derivative envelope. The global certificate at $c_n$ with radius $n^{-1/4}$ holds with probability tending to one by the preceding argument. Its curvature stays bounded below, so its error bound gives $O_p(n^{-1})$.
\end{proof}
These conditions certify numerical accuracy, not termination of a particular optimizer. An infinity-norm gradient threshold is converted to a Euclidean threshold by a factor $\sqrt d$. The fixed tolerances used in the reported experiments do not by themselves establish the required asymptotic rates or certify their individual outputs. The one-step result is an optional refinement with the weight held fixed, not a change to those experiments. The probability bounds above are not $L^2$ risk bounds; finite-sample risk transfer requires the certificate and failure accounting in Supplementary Section~\ref{supp:logistic-global-certificate}.

\section{Finite-sample risk analysis}
\label{app:risk-proofs}
For the fixed-design linear model of Section~\ref{sec:finite-sample}, write $C=(X^\top X)^{-1}X^\top$, $\Gamma=C\operatorname{Cov}(\epsilon)C^\top$ and use the following explicit projection coefficients.

Write $J=\Gamma^{1/2}$, $z=\Gamma^{-1/2}(t-\beta)$, $g=\widehat\beta_0-t$ and $u=\widehat\beta_h-\widehat\beta_0$. Define
\begin{align*}
 a_g&=Eg,\quad A_g=E(gz^\top),\quad a_u=Eu,\quad A_u=E(uz^\top),\\
 g_\perp&=g-a_g-A_gz,\quad u_\perp=u-a_u-A_uz,\\
 Q_g&=E\|g_\perp\|^2,\quad Q_u=E\|u_\perp\|^2,\\
 \mathcal D&=2a_g^\top a_u+\|a_u\|^2+2\langle J+A_g,A_u\rangle_F+\|A_u\|_F^2.
\end{align*}

\subsection{Conditional projection refinements}
\label{supp:conditional-projections}

Under $\epsilon\sim N(0,\sigma^2 I)$, $r=(I-XC)Y$ is independent of $t$. Conditional projections given $r$ yield the same theorem with all projection coefficients computed conditionally and then averaged over $r$. Write $Q_g^{\rm cond}=\mathbb E\|g-a_g(r)-A_g(r)z\|^2$ and $Q_u^{\rm cond}=\mathbb E\|u-a_u(r)-A_u(r)z\|^2$ for these conditional residual energies, distinct from the unconditional $Q_g,Q_u$ in the main paper. For square-integrable first weak derivatives,
\begin{equation}Q_g^{\rm cond}\le\tfrac12\{\mathbb E\|D_zg\|_F^2-\mathbb E_r\|A_g(r)\|_F^2\},\quad
 Q_u^{\rm cond}\le\tfrac12\{\mathbb E\|D_zu\|_F^2-\mathbb E_r\|A_u(r)\|_F^2\}.\label{eq:fs-weak-gradient-tail}\end{equation}
The conditional Hermite expansion below proves these bounds. Positive ridge and Lipschitz spectral clipping give the required weak derivatives with polynomial Gaussian growth; classical second derivatives need not exist at clipping thresholds \citep{belleczhang2021}.

\subsection{Exact signed projection argument}

With $z=\Gamma^{-1/2}(t-\beta)$, $\mathbb Ez=0$ and $\mathbb Ezz^\top=I$. For any square-integrable vector $w$, its affine projection onto $z$ is $\mathbb Ew+\mathbb E(wz^\top)z$: direct multiplication proves zero residual mean and zero residual cross moment with $z$. These properties do not require Gaussian errors. Applying them to $g,u$ gives
\begin{align*}
 \mathbb E\langle\widehat\beta_0-\beta,u\rangle
 &=a_g^\top a_u+\langle J+A_g,A_u\rangle_F
                +\mathbb E\langle g_\perp,u_\perp\rangle,\\
 \mathbb E\|u\|^2&=\|a_u\|^2+\|A_u\|_F^2+Q_u.
\end{align*}
Subtracting squared losses proves the exact identity in Theorem~\ref{thm:fs-signed-transfer}. Cauchy--Schwarz yields both its upper bound and the lower bound $\mathcal D+Q_u-2\sqrt{Q_gQ_u}$. In particular, the sufficient margin is a statement about exact population expectations, not Monte Carlo estimates or samplewise domination.

Equivalently, if $Q_h=\mathbb E\|g_\perp+u_\perp\|^2$, then $R(h)-R(0)=\mathcal D+Q_h-Q_g$, giving the additional upper bound $\mathcal D+\min\{Q_h,Q_u+2\sqrt{Q_gQ_u}\}$. Neither formula assumes the nonlinear covariance is positive.

\subsection{Proof of the design-based improvement theorem}
\label{app:primitive-proof}

\begin{theorem}[Design-dependent improvement and prediction-error tolerance]
\label{thm:fs-primitive-improvement}
Consider scalar regression without an intercept, with fixed labeled and unlabeled covariates in $\{0,-a,a\}$, $a>0$. Let $m\ge2$ and $M_U$ be their numbers of nonzero rows, and put $p=m/n$, $u=M_U/N$, with $0<u<p<1$. Let $Y_i=\beta x_i+\epsilon_i$, where the errors are independent $N(0,\sigma^2)$, $\sigma>0$, and let $f(x)=\alpha x$. Use unconstrained fits, the centered divisors $n,N$, the OLS pilot, any fixed $\eta>0$, cap $\kappa=1/2$, and the complete fit along \eqref{eq:fs-calibrated-member}. Write $\mu=\beta-\alpha$ and $R_\mu(\lambda)=\mathbb E_\mu(\widehat\beta_\lambda-\beta)^2$.

At $\mu=0$, the risk is strictly decreasing in $\lambda\in[0,1]$ if $p+u<1$, strictly increasing if $p+u>1$, and constant if $p+u=1$. In particular, the first case gives $\mathcal G_{\rm ref}=R_0(0)-R_0(h)>0$ for $h=1/2$.

For this favorable case, set $s^2=\sigma^2/(ma^2)$ and define, under $\mu=0$,
\begin{equation*}H_h^2=\mathbb E_0\{[(\widehat\beta_h-\alpha)^2-(\widehat\beta_0-\alpha)^2]^2\},\qquad
 J_h^2=\mathbb E_0(\widehat\beta_h-\widehat\beta_0)^2.\end{equation*}
These constants are finite and have explicit Gaussian/chi-square integral expressions in Supplementary Section~\ref{app:primitive-proof}. With $q=\mu^2/s^2$, put
\begin{equation}\varepsilon_h(\mu)=H_h\sqrt{\cosh(q)-1}
                  +2|\mu|J_h\sqrt{\sinh(q)}.\label{eq:fs-shift-budget}\end{equation}
Then
\begin{equation*}R_\mu(0)-R_\mu(h)\ge \mathcal G_{\rm ref}-\varepsilon_h(\mu).\end{equation*}
Consequently $\varepsilon_h(\mu)<\mathcal G_{\rm ref}$ is a quantitative sufficient condition for strict improvement. It holds on a nonempty neighborhood of $\mu=0$.
\end{theorem}
In the favorable regime of Theorem~\ref{thm:fs-primitive-improvement}, $\mathcal G_{\rm ref}>0$ is the gain at the reference predictor. For fixed $s>0$, the transfer budget has expansion\begin{equation*}
 \varepsilon_h(\mu)=
 \left\{\frac{H_h}{\sqrt2\,s^2}+\frac{2J_h}{s}\right\}\mu^2+o(\mu^2).
\end{equation*}
The bound loses the reference margin at quadratic order in prediction error. Strict improvement therefore persists throughout a nonzero neighborhood of the reference slope.

For Theorem~\ref{thm:fs-primitive-improvement}, write $v=t-\alpha$, $\rho=n/N$ and
\begin{align*}
 L_a&=a^2p,\quad U_a=a^2u,\quad B_a=U_a-L_a<0,\\
 V_L&=a^4p(1-p),\quad V_U=a^4u(1-u),\quad k=V_U-V_L,\\
 c&=n^{-1}\sum_i x_i^2r_i^2,\quad C_a=c+\eta,\quad
 D_a=v^2(V_L+\rho V_U)+\eta.
\end{align*}
Here $r=Y-Xt$. Since $x_i^3=a^2x_i$ and $\sum_i x_ir_i=0$, the centered covariance between $s_i=-x_ir_i$ and $h_i=x_i^2v$ is exactly zero. Hence $\mathsf S=\operatorname{diag}(C_a,D_a)$. The Gaussian OLS decomposition gives independent laws
\begin{equation}v\sim N(\mu,s^2),\qquad
 c\sim\frac{\sigma^2a^2}{n}\chi^2_{m-1},\qquad
 s^2=\frac{\sigma^2}{ma^2}.\label{eq:app-threepoint-laws}\end{equation}
To check the second law, restrict to the $m$ active rows: their residual is an orthogonal projection of $m$ independent Gaussian errors onto the complement of their one-dimensional design. Inactive rows have zero score and do not contribute to $c$.

Put $F_a=C_a^{-1}+D_a^{-1}$ and
\begin{equation*}k_p=\frac{\operatorname{clip}_{[-\kappa,\kappa]}(v^2kF_a)}{F_a},\quad
 D_0=L_a^2D_a+B_a^2C_a,\quad
 F_0=L_aB_a(L_aD_a-B_aC_a)<0.\end{equation*}
Formula~\eqref{eq:fs-protected-lowrank} gives this exact protected update. Inverting the two-by-two covariance in the normal equation gives
\begin{equation}\widehat\beta_\lambda-\alpha=v\phi_\lambda(v,c),\qquad
 \phi_\lambda=
 \frac{L_a(L_aD_a+\lambda U_a k_p)}{D_0+\lambda U_a^2 k_p}.\label{eq:app-threepoint-fit}\end{equation}
The fit and its moments are functions of the two independent variables in \eqref{eq:app-threepoint-laws}, with their dependence through the estimated covariance unchanged.

The bound $|k_p|\le\kappa C_aD_a/(C_a+D_a)\le\kappa D_a$ implies $L_aD_a+\lambda U_a k_p\ge(L_a-\kappa U_a)D_a>0$. Also
$U_a^2C_aD_a\le(C_a+D_a)D_0$, since the difference is $(L_aD_a-B_aC_a)^2$. Thus $|U_a^2k_p|\le\kappa D_0$, proving positivity of the denominator and $\phi_\lambda>0$ even under active clipping. Differentiation in strength gives
\begin{equation}\partial_\lambda\phi_\lambda
 =-\frac{U_a k_pF_0}{(D_0+\lambda U_a^2k_p)^2},\qquad
 \left.\partial_\lambda(\widehat\beta_\lambda-\beta)^2\right|_{\mu=0}
 =-\frac{2v^2\phi_\lambda U_a k_pF_0}{(D_0+\lambda U_a^2k_p)^2}.\label{eq:app-threepoint-sign}\end{equation}
Its sign for $v\ne0$ is the sign of $k=a^4(u-p)(1-u-p)$. Gaussian moment bounds and protected denominators justify integration. This proves strict monotonicity, including unique endpoint optima when $k\ne0$; at $k=0$ the update is zero for every $\mu$.

For the favorable case and $h=1/2$, the exact constants used in the main theorem are
\begin{align}
 \mathcal G_{\rm ref}&=\mathbb E_0[v^2(\phi_0^2-\phi_h^2)],\nonumber\\
 H_h^2&=\mathbb E_0[v^4(\phi_h^2-\phi_0^2)^2],\qquad
 J_h^2=\mathbb E_0[v^2(\phi_h-\phi_0)^2].
 \label{eq:app-threepoint-integrals}
\end{align}
Equations~\eqref{eq:app-threepoint-laws}--\eqref{eq:app-threepoint-fit} specify the distributions and integrands needed to evaluate these expectations.

An analytic lower bound is obtained as follows. Choose $0<c_-<c_+$ and $0<v_-<v_+$, put $D_X=V_L+\rho V_U$, and define\begin{align*}
 C_\pm&=c_\pm+\eta,\quad D_\pm=v_\pm^2D_X+\eta,\quad D_{0,+}=L_a^2D_++B_a^2C_+,\\
 k_-&=\min\{|k|v_-^2,\kappa C_-D_-/(C_-+D_-)\},\\
 F_-&=L_a|B_a|(L_aD_-+|B_a|C_-),\quad
 \phi_- =\frac{L_a(L_a-\kappa U_a)D_-}{(1+\kappa)D_{0,+}},\\
 j_-&=\frac{2v_-^2\phi_-U_a k_-F_-}{(1+\kappa)^2D_{0,+}^2},\\
 \underline G_h&=h j_-\Pr(c_-\le c\le c_+)\\
 &\quad{}\times\Pr(v_-\le|N(0,s^2)|\le v_+)>0.
\end{align*}
On this rectangle the absolute derivative in \eqref{eq:app-threepoint-sign} is at least $j_-$ uniformly in strength. Integration over $[0,h]$ proves $\mathcal G_{\rm ref}\ge\underline G_h$. Moreover $0<\phi_\lambda\le B_\phi=(1+\kappa U_a/L_a)/(1-\kappa)$, so
\begin{equation*}H_h\le\sqrt3 B_\phi^2s^2,\qquad J_h\le B_\phi s.\end{equation*}
Substituting these upper bounds and $\underline G_h$ into \eqref{eq:fs-shift-budget} yields an explicit sufficient inequality using only Gaussian/chi-square probabilities and model/design parameters. Its left side tends to zero with $\mu$ whereas its right side is strictly positive. The resulting interval can be conservative; the sharper constants in \eqref{eq:app-threepoint-integrals} are preferable when rigorous integration bounds are available.

\par\medskip\noindent{\bfseries Numerical evaluation of the primitive improvement regime.}\par\nobreak\smallskip

The numerical example $n=40$, $N=200$, $a=\sigma=1$, $p=0.4$, $u=0.2$, and ridge $\eta=n^{-2}=1/1600$ has numerically integrated $\mathcal G_{\rm ref}\simeq9.38\times10^{-4}$. At $\mu=0.0625$, \eqref{eq:fs-shift-budget} is approximately $2.94\times10^{-4}$, leaving a positive lower margin of approximately $6.44\times10^{-4}$; the integrated actual gain is approximately $9.01\times10^{-4}$. These are quadrature evaluations of a proved inequality, not certified integration bounds. Supplementary Section~\ref{app:primitive-proof} additionally gives a conservative analytic lower margin and moment upper bounds that establish a nonempty tolerance range without numerical integration. Stability under small design or prediction perturbations is established in Supplementary Section~\ref{app:perturbation-proof}; no uniform or high-dimensional radius is asserted.

\subsection{Gaussian-shift transfer in every fixed dimension}
\label{app:shift-proof}

Let the fixed external predictor be $f(x)=x^\top\alpha$, $\mu=\beta-\alpha$, $V_t=\sigma^2(X^\top X)^{-1}$ and $v=t-\alpha$. Under Gaussian errors, $v\sim N(\mu,V_t)$ independently of $r=(I-XC)Y$. Write the actual fit as $\widehat\beta_\lambda=\alpha+a_\lambda(v,r)$. Its centered affine moments have right-hand side $(A_Lv,0)^\top$. Joint reflection $(v,r)\mapsto(-v,-r)$ changes the sign of all pilot scores, preserves their centered covariance, ridge and protected update, and changes the sign of this right-hand side. Consequently $a_\lambda(-v,-r)=-a_\lambda(v,r)$.

For deterministic strengths $s,b$, define $h_{sb}=\|a_s\|^2-\|a_b\|^2$, $j_{sb}=a_s-a_b$, $H_{sb}=\|h_{sb}\|_{L^2(\mathbb P_0)}$ and $J_{sb}=\|j_{sb}\|_{L^2(\mathbb P_0)}$. The maps are evaluated under the zero-shift reference law $\mathbb P_0$; they contain no unknown $\beta$ argument. Put $q_\mu=\mu^\top V_t^{-1}\mu$. The likelihood ratio is $\ell_\mu(v)=\exp(\mu^\top V_t^{-1}v-q_\mu/2)$. Its reflection-even and reflection-odd parts satisfy
\begin{equation*}\mathbb E_0(\ell_{\rm even}-1)^2=\cosh(q_\mu)-1,\qquad
 \mathbb E_0\ell_{\rm odd}^2=\sinh(q_\mu).\end{equation*}
Evenness of $h_{sb}$ and oddness of $j_{sb}$ give
\begin{equation*}R_\mu(s)-R_\mu(b)
 =\mathbb E_0[h_{sb}\ell_{\rm even}]-2\mu^\top\mathbb E_0[j_{sb}\ell_{\rm odd}].\end{equation*}
Subtract the zero-shift contrast and apply Cauchy--Schwarz to obtain
\begin{equation*}|R_\mu(s)-R_\mu(b)-R_0(s)+R_0(b)|
 \le H_{sb}\sqrt{\cosh(q_\mu)-1}
     +2\|\mu\|J_{sb}\sqrt{\sinh(q_\mu)}.\end{equation*}
All moments are finite by the fixed-design Gaussian growth bounds. For a fixed design this perturbation is $O(\|\mu\|^2)$ near zero. Setting $s=h,b=0$ proves the main theorem. In higher dimensions this transfers a reference gain when one is separately established; Theorem~\ref{thm:linear-multiparameter} supplies one multiparameter reference sign, whereas a favorable sign for arbitrary designs remains open (Supplementary Section~\ref{app:literature-protocol}). No data-selected strength can be substituted without accounting for its selection map.

\subsection{Stability of an established improvement margin}
\label{app:perturbation-proof}

Consider two fixed design/prediction configurations with the same $n,N,d,\beta,\sigma,\eta>0$ and cap, and couple their responses using the same standard Gaussian noise. Superscripts $0,1$ distinguish their actual fitted maps. For $a\in\{0,h\}$ put $d_a=\|\widehat\beta_a^1-\widehat\beta_a^0\|_{L^2}$. The exact loss expansion and Cauchy--Schwarz give
\begin{equation*}|R_a^1-R_a^0|\le 2\sqrt{R_a^0}\,d_a+d_a^2.\end{equation*}
Hence an established margin $G^0=R_0^0-R_h^0>0$ survives whenever
\begin{equation*}2\sqrt{R_0^0}\,d_0+d_0^2+2\sqrt{R_h^0}\,d_h+d_h^2<G^0.\end{equation*}
At fixed positive ridge, the actual fit is continuous in finite design and prediction entries near any full-rank labeled design: the pilot, positive-definite square roots, spectral clipping and normal-equation inverses are continuous. On a sufficiently small compact neighborhood the labeled singular value is bounded below and the polynomial Gaussian growth envelope in Supplementary Section~\ref{app:fit-regularity} is uniform. Dominated convergence therefore makes $d_0,d_h\to0$ as the configuration converges. Every strictly favorable reference configuration has an open neighborhood of favorable configurations of the same dimension, including predictions not exactly linear on the observed rows. This proves existence of a neighborhood, not a useful uniform numerical radius, a change of dimension, or improvement after averaging over arbitrary random designs.

\subsection{Conditioning, Hermite tails, and affine residual budgets}

For homoscedastic Gaussian errors, $r=(I-XC)\epsilon$ and $t-\beta=C\epsilon$ are jointly Gaussian with cross covariance zero, since $C(I-XC)=0$ and $XC$ is an orthogonal projector. They are independent. Given $r$, $z\sim N(0,I_d)$. Use conditional coefficients $a_g(r),A_g(r),a_u(r),A_u(r)$ and conditional projection residuals; then all preceding identities hold after averaging over $r$. The resulting tail energies are $Q_g^{\rm cond},Q_u^{\rm cond}$ of \eqref{eq:fs-weak-gradient-tail}; the unconditional tails also include variation of the conditional affine coefficients.

Let $\{H_\alpha\}$ be normalized multivariate Hermite polynomials for standard Gaussian measure, indexed by multi-indices $\alpha\in\mathbb N^d$, with coefficients $a_\alpha(r)=\mathbb E_z[gH_\alpha(z)\mid r]$. Orthogonality, completeness, and $\partial_jH_\alpha=\sqrt{\alpha_j}H_{\alpha-e_j}$ imply, for vector-valued Gaussian Sobolev functions,
\begin{equation*}\mathbb E_z\|g\|^2=\sum_\alpha\|a_\alpha(r)\|^2,\qquad
 \mathbb E_z\|D_zg\|_F^2=\sum_\alpha|\alpha|\|a_\alpha(r)\|^2.\end{equation*}
For polynomials this follows term by term; for square-integrable weak derivatives it follows by completion in the Gaussian Sobolev norm. Degrees zero and one are precisely the affine projection, with degree-one coefficient matrix $A_g(r)$. Every remaining degree contributes at least twice its energy in the second identity. This proves \eqref{eq:fs-weak-gradient-tail}, including its factor $1/2$. The same calculation applies to $u$. The weak-derivative Gaussian integration-by-parts framework is consistent with \citep{belleczhang2021}; this tail decomposition is derived here rather than attributed as a theorem about our target.

First-order Stein also gives $A_g(r)=\mathbb E_z[D_zg\mid r]$ and $A_u(r)=\mathbb E_z[D_zu\mid r]$. In coordinates, $D_zg=(D_t\widehat\beta_0-I)J$, $D_zu=(D_t\widehat\beta_h-D_t\widehat\beta_0)J$, holding $r$ fixed. Perturbing $Y$ by $Xv$ changes $t$ by $v$ and leaves $r$ unchanged, so these derivatives retain the entire same-response pipeline.

For any conditional affine approximation $a(r)+A(r)z$, orthogonality gives
\begin{equation*}\mathbb E_z\|g-a-Az\|^2
 =\mathbb E_z\|g_\perp\|^2+\|a_g(r)-a\|^2+\|A_g(r)-A\|_F^2.\end{equation*}
Thus independently fitted affine residuals bound the exact tail in population. Finite validation averages need uncertainty control before they become confidence bounds. These coefficients are theoretical diagnostics depending on unknown $\beta$, not an automatically feasible selector.

\subsection{Proof of full-GMM improvement over bounded scalar corrections}
\label{supp:full-gmm-scalar}
We prove Theorem~\ref{thm:full-gmm-scalar}, conditioning throughout on the fixed designs. Write $A=aI_d$, $D=uI_d$, $B=D-A$, and $t=(X^\top X)^{-1}X^\top Y$. Use the actual centered covariance in Section~\ref{subsec:implementation} with $\eta=n^{-2}$; no calibration is applied.

\par\medskip\noindent\textbf{An exact pointwise lower bound for the comparator.}\par\smallskip
The matrix $A+\ell B=\{a-\ell(a-u)\}I_d$ is invertible on $[0,1]$. Therefore
\[
 \widehat\beta^P_\ell=\alpha+c_\ell(t-\alpha),\qquad
 c_\ell=\frac{a}{a-\ell(a-u)}\in[1,C],\quad C=a/u.
\]
At $\beta=\alpha$ this implies $\|\widehat\beta^P_\ell-\beta\|^2\ge\|t-\beta\|^2$ for every response. It remains true when the coefficient is selected from that response, even by an oracle. Since $\ell=0$ attains equality and $\operatorname{Cov}(t)=\sigma^2(na)^{-1}I_d$, equation~\eqref{eq:full-gmm-scalar-floor} follows. Independence between $\widehat\ell$ and $t$ is never used.

\par\medskip\noindent\textbf{The actual covariance is a positive quadratic perturbation.}\par\smallskip
At $\beta=\alpha$, put $w=(X^\top X)^{-1}X^\top Z$ and $P=X(X^\top X)^{-1}X^\top$. Then $t=\alpha+\sigma w$, $P$ is an orthogonal projection, and
\[
 s_i=\sigma X_i\{(PZ)_i-Z_i\},\quad
 h_i=\sigma X_i(PZ)_i,\quad h_k^U=\sigma U_kU_k^\top w.
\]
Thus $S=\eta I+\sigma^2S_2(Z)$, with $S_2\succeq0$. The following explicit bound will also produce a finite noise threshold. Let
\[
 L_X=\max_i\|X_i\|^2,\quad L_U=\max_k\|U_k\|^2,\quad
 C_2=L_X/n+uL_U/(Na).
\]
Centering reduces the trace of a covariance. The labeled contribution has trace at most
\[
 \frac{L_X}{n}\{\|(I-P)Z\|^2+\|PZ\|^2\}=\frac{L_X}{n}\|Z\|^2.
\]
The unlabeled contribution has trace at most $(n/N)uL_U\|w\|^2\le uL_U\|Z\|^2/(Na)$, since $\|w\|\le\|Z\|/\sqrt{na}$. Positivity then gives
\begin{equation}\label{eq:gmm-class-C2}
 \|S_2(Z)\|\le\operatorname{tr}S_2(Z)\le C_2\|Z\|^2.
\end{equation}
This calculation retains the full joint covariance and the common responses.

\par\medskip\noindent\textbf{Global control of the fitted-parameter remainder.}\par\smallskip
Let $H=(A^\top,B^\top)^\top$, $v_Z=((Aw)^\top,0^\top)^\top$, and $h_0=a^2+(u-a)^2$. The exact normal equation gives
\[
 \widehat\beta_G-\alpha=\sigma(H^\top WH)^{-1}H^\top Wv_Z.
\]
Write $l_0=(H^\top H)^{-1}H^\top v_Z=kw$ and $r_Z=v_Z-Hl_0$. Because $H^\top r_Z=0$,
\begin{equation}\label{eq:gmm-class-remainder-identity}
 e_\sigma:=\widehat\beta_G-\alpha-\sigma kw
 =\sigma(H^\top WH)^{-1}H^\top(W-\eta^{-1}I)r_Z.
\end{equation}
Since $S_2\succeq0$,
\[
 \|(H^\top WH)^{-1}\|\le\frac{\eta+\sigma^2\|S_2\|}{h_0},\qquad
 \|W-\eta^{-1}I\|\le\frac{\sigma^2\|S_2\|}{\eta^2}.
\]
Also $\|H\|=\sqrt{h_0}$ and $\|r_Z\|\le\|v_Z\|\le\sqrt{a/n}\|Z\|$. Hence for $0<\sigma\le1$,
\[
 \|e_\sigma\|\le\sigma^3D_0\{\eta\|Z\|^3+C_2\|Z\|^5\},\qquad
 D_0=\frac{C_2\sqrt{a/n}}{\sqrt{h_0}\eta^2}.
\]
For $Z\sim N(0,I_n)$, define $M_6=n(n+2)(n+4)$ and $M_{10}=n(n+2)(n+4)(n+6)(n+8)$. Minkowski's inequality yields
\begin{equation}\label{eq:gmm-class-C3}
 (E\|e_\sigma\|^2)^{1/2}\le C_3\sigma^3,\qquad
 C_3=D_0\{\eta\sqrt{M_6}+C_2\sqrt{M_{10}}\}<\infty.
\end{equation}
The cross product with $\sigma kw$ is $O(\sigma^4)$, proving~\eqref{eq:full-gmm-scalar-risk}. More explicitly,
\[
 \sqrt{R_G(\alpha)}\le\sigma k\sqrt{v_0}+C_3\sigma^3.
\]
One valid, conservative threshold is
\begin{equation}\label{eq:gmm-class-threshold}
 \sigma_0=\min\left\{1,\sqrt{\frac{(1-k)\sqrt{v_0}}{2C_3}}\right\}>0.
\end{equation}
For $0<\sigma<\sigma_0$, $R_G(\alpha)\le\sigma^2v_0\{(1+k)/2\}^2<v_0\sigma^2$. Thus the superiority holds on an actual finite-noise interval, not only as a formal expansion. The bound can be extremely conservative and is not a recommended application cutoff. No limit-expectation exchange without domination is used.

\par\medskip\noindent\textbf{A common mismatch neighborhood for all selection rules.}\par\smallskip
Fix such a $\sigma$ and let $\mu=\beta-\alpha$. Every coefficient rule has error $\sigma c_{\widehat\ell}w+(c_{\widehat\ell}-1)\mu$, so
\begin{equation}\label{eq:gmm-class-mismatch}
 R_\beta(\widehat\beta^P_{\widehat\ell})\ge v_0\sigma^2
 -2\sigma C(C-1)E\|w\|\,\|\mu\|.
\end{equation}
This bound is uniform over the rules, including discontinuous selectors. For GMM, on a compact $\mu$ neighborhood, the scores are affine in $(\mu,Z)$ and $\|S\|$ has a common quadratic Gaussian envelope. The bounds $S\succeq\eta I$, $\|S^{-1}\|\le\eta^{-1}$ and $\|(H^\top S^{-1}H)^{-1}\|\le\|S\|/h_0$ give a common polynomial envelope for the squared error. Pointwise continuity and dominated convergence imply continuity of $R_G(\beta)$ at $\alpha$.

Let $\Delta_\sigma=v_0\sigma^2-R_G(\alpha)>0$. Choose $\delta_\sigma>0$ such that $|R_G(\beta)-R_G(\alpha)|<\Delta_\sigma/3$ and the allowance in~\eqref{eq:gmm-class-mismatch} is less than $\Delta_\sigma/3$ when $\|\mu\|<\delta_\sigma$. Then $R_G(\beta)$ is at least $\Delta_\sigma/3$ below the lower bound for the entire class. Taking the infimum proves~\eqref{eq:full-gmm-scalar-dominance}. This completes the proof.

\par\medskip\noindent\textbf{Implementation and scope.}\par\smallskip
The linear PPI++ update $t-(A+\ell B)^{-1}\ell m_2(t)$ is exactly the corrected root used above. Define its coefficient as zero if the tuning denominator vanishes. When $t\ne\alpha$, the labeled and unlabeled prediction-score means differ by $(a-u)(t-\alpha)$, so the pooled covariance trace is positive; the exceptional case $t=\alpha$ has Gaussian probability zero at $\sigma>0$. This completion does not change risk. The result keeps sample sizes and positive ridge fixed. The mismatch radius may depend on noise. It does not imply unconditional random-design dominance or a comparison with unrestricted scalar or matrix corrections. At exact prediction the covariance before ridge is singular, so this proof is separate from the regular efficiency limit.

\subsection{A Gaussian multiparameter improvement regime}
\label{supp:linear-multiparameter}

The following proof uses the original centered joint covariance, the same-response OLS pilot, and relative spectral protection. The noise scale varies with the finite designs and ridge held fixed. It is not a sample-size asymptotic argument.

\begin{proof}[Proof of Theorem~\ref{thm:linear-multiparameter}]
Write $L=a^2p$, $B=a^2(u-p)$, $U=L+B=a^2u$, $D_*=L^2+B^2$, $c_*=L^2/D_*$, and
\[
 A_*=-\frac{LB U(L-B)}{\eta D_*^2}>0,\qquad
 \tau^2=(ma^2)^{-1}.
\]
At $\beta=\alpha$, couple the responses as $Y=X\alpha+\sigma Z$, with $Z\sim N(0,I_n)$. The pilot is $t=\alpha+\sigma w$, where $w=(X^\top X)^{-1}X^\top Z\sim N(0,\tau^2I_d)$. Each nonzero design row is $ae_j$ or $-ae_j$. Orthogonality of the OLS residual within each coordinate group implies that the centered outcome--prediction score cross covariance vanishes exactly. Thus
\begin{equation*}
 \mathsf S=\eta I_{2d}+\sigma^2 S_2(Z),\qquad
 S_2=\operatorname{diag}\{C_2,D_2\}\succeq0,
\end{equation*}
where $C_2$ is diagonal and
\begin{align*}
 D_2&=a^4\{(p+\rho u)\operatorname{diag}(w_j^2)
                    -(p^2+\rho u^2)ww^\top\},\\
 K&=\sigma^2K_2,\qquad
 K_2=a^4\{(u-p)\operatorname{diag}(w_j^2)
                    -(u^2-p^2)ww^\top\}.
\end{align*}
In particular, global centering produces nonzero off-diagonal terms. They are retained throughout. Both $S_2$ and $K_2$ have norm bounded by a fixed multiple of $\|Z\|^2$.

Let $H=(LI_d,BI_d)^\top$ and center the parameter at $\alpha$. The affine moment right-hand side is $\sigma y$, $y=(Lw,0)^\top$. Without clipping, write $M_\lambda=S_2+\lambda T K_2T^\top$. Expanding the weighted normal equation at $\sigma=0$ gives, in $L^2$ uniformly for $\lambda\in[0,1]$,
\begin{equation*}
 \widehat\beta_\lambda-\alpha
 =\sigma c_*w+\sigma^3\{b_0(Z)+\lambda A_*K_2w\}
       +O_{L^2}(\sigma^5).
\end{equation*}
Put $a_*=c_*w$ and $r_*=Ha_*-y$, so $H^\top r_*=0$. The cubic coefficient is $\eta^{-1}D_*^{-1}H^\top M_\lambda r_*$; in particular $b_0=\eta^{-1}D_*^{-1}H^\top S_2r_*$ is a fixed polynomial in $Z$. To check the sign and coefficient, its derivative with respect to strength is
\[
 \eta^{-1}(H^\top H)^{-1}H^\top T K_2T^\top(Hc_*w-y)
 =-\frac{LB U(L-B)}{\eta D_*^2}K_2w.
\]

We justify the expansion also for the protected estimator. Choose a constant $b>0$ sufficiently small that $\|Z\|\le b/\sigma$ implies
\[
 \sup_{\lambda\in[0,1]}\|\sigma^2 M_\lambda/\eta\|\le1/4,
 \qquad
 \|\mathsf S^{-1/2}\sigma^2TK_2T^\top\mathsf S^{-1/2}\|<\kappa.
\]
Such a $b$ exists because the relevant matrix norms are bounded by a constant times $\|Z\|^2$. Clipping is inactive on this event. Write
\[
 E_\lambda=\sigma^2M_\lambda/\eta,\quad
 A_\lambda=H^\top(I+E_\lambda)^{-1}H,\quad
 \delta_\lambda=\|E_\lambda\|\le1/4.
\]
The exact normal equation gives $\widehat\beta_\lambda-\alpha-\sigma a_*=-\sigma A_\lambda^{-1}H^\top(I+E_\lambda)^{-1}r_*$. Use $H^\top H=D_*I$, $H^\top r_*=0$ and $(I+E)^{-1}=I-E+E^2(I+E)^{-1}$. Since
\[
 \|A_\lambda^{-1}\|\le(1+\delta_\lambda)/D_*,\qquad
 \|A_\lambda^{-1}-D_*^{-1}I\|
 \le\frac{(1+\delta_\lambda)\delta_\lambda}{(1-\delta_\lambda)D_*},
\]
subtracting the cubic approximation yields the explicit bound
\begin{equation}
 \left\|\widehat\beta_\lambda-\alpha-\sigma a_*
       -\frac{\sigma^3}{\eta D_*}H^\top M_\lambda r_*\right\|
 \le\frac{10\sigma^5}{3\eta^2\sqrt{D_*}}
           \|M_\lambda\|^2\|r_*\|.
 \label{eq:multi-explicit-remainder}
\end{equation}
Indeed, the two remainder terms are $\sigma(A_\lambda^{-1}-D_*^{-1}I)H^\top E_\lambda r_*$ and $-\sigma A_\lambda^{-1}H^\top E_\lambda^2(I+E_\lambda)^{-1}r_*$. Their sum is bounded by $2(1+\delta_\lambda)\sigma\delta_\lambda^2\|r_*\|/\{(1-\delta_\lambda)\sqrt{D_*}\}$, giving the constant $10/3$. The right side of \eqref{eq:multi-explicit-remainder} is at most $C\sigma^5(1+\|Z\|^5)$, uniformly in strength. On the complementary Gaussian tail, the protected eigenvalue bound and the affine pseudoinverse bound give
\[
 \|\widehat\beta_\lambda-\alpha\|
 \le C\sigma\|Z\|\sqrt{1+\sigma^2\|Z\|^2/\eta}.
\]
Every polynomial tail moment is smaller than every power of $\sigma$ there. The polynomial approximation has the same tail integrability, so both its tail and that of the actual fit are negligible after division by any fixed power of $\sigma$. The $L^2$ remainder assertion follows. This argument includes samples with active clipping rather than deleting them.

Expansion of squared Euclidean loss now yields
\begin{equation}
 R_0(0)-R_0(h)
 =-2h c_* A_*\sigma^4 E(w^\top K_2w)+O(\sigma^6).
 \label{eq:multi-risk-leading}
\end{equation}
Using $Ew_j^4=3\tau^4$ and $E\|w\|^4=d(d+2)\tau^4$ gives
\begin{equation*}
 E(w^\top K_2w)=a^4\tau^4d(u-p)\{3-(d+2)(p+u)\}<0.
\end{equation*}
Consequently the coefficient in \eqref{eq:multi-risk-leading} is strictly positive. There are finite constants $C_6,\sigma_1>0$ such that the absolute remainder is at most $C_6\sigma^6$ for $0<\sigma<\sigma_1$. If $G_*$ denotes the positive coefficient of $\sigma^4$, every $\sigma<\min\{\sigma_1,\sqrt{G_*/(2C_6)}\}$ has gain at least $G_*\sigma^4/2$. This proves existence of a strictly favorable finite-sample noise range; the constants are not claimed to give a practically sharp threshold.

For each such fixed $\sigma$, Gaussian-shift transfer in Supplementary Section~\ref{app:shift-proof} gives the stated mismatch bound and a nonempty neighborhood of $\mu=0$. The risk-continuity argument in Supplementary Section~\ref{app:perturbation-proof} further gives an open neighborhood of the reference design and prediction entries with the same strict gain. Generic perturbations in that neighborhood have nonorthogonal columns. Thus improvement is not confined to exactly decoupled coordinate groups, although no uniform neighborhood radius is asserted.
\end{proof}

The sufficient condition $(d+2)(p+u)<3$ concerns the expected sign; it need not make $K_2$ negative semidefinite for every response. A covariance ordering is not used as a substitute for a parameter-risk argument. The result fixes $n,N,d,\eta$ and then restricts $\sigma$. It includes the experimental ridge $\eta=n^{-2}$ at fixed $n$, but does not certify the noise level or Gaussian designs of the main simulation tables. The scalar theorem supplies an additional regime without a small-noise requirement.

\par\medskip\noindent\textbf{Prediction mismatch.}\par\smallskip
For a noise level with $G_\sigma=R_0(0)-R_0(h)>0$, the quantitative prediction-mismatch bound is as follows. For every fixed noise level in this range, let $V_t=\sigma^2(X^\top X)^{-1}$, $q=\mu^\top V_t^{-1}\mu$, $a_\lambda=\widehat\beta_\lambda-\alpha$ under $\mu=0$, and
\[
 H_h^2=E_0(\|a_h\|^2-\|a_0\|^2)^2,\qquad
 J_h^2=E_0\|a_h-a_0\|^2.
\]
These constants are finite and
\begin{equation}
 R_\mu(0)-R_\mu(h)\ge G_\sigma
 -H_h\sqrt{\cosh(q)-1}-2\|\mu\|J_h\sqrt{\sinh(q)}.
 \label{eq:multi-mismatch-budget}
\end{equation}
The right side is positive on a nonempty neighborhood of zero mismatch. Strict improvement also persists under sufficiently small design and prediction perturbations, including nonorthogonal designs.

\subsection{Multiparameter logistic risk: curvature and an exact design family}
\label{supp:logistic-risk}

We first bound the local curvature error for logistic moments, then establish exact conditional parameter-risk improvement in a structured design. The latter concerns log-odds squared error conditional on exact pilot existence. Neither result gives unconditional dominance for arbitrary designs or covers numerical returns on separated samples.

\par\medskip\noindent\textbf{A local curvature comparison.}\par\nobreak\smallskip
Use the logistic moments of Section~\ref{sec:setup}, with independent Bernoulli responses of means $\sigma(x_i^\top\beta)$ and fixed predictions in $[0,1]$. For this subsection set $Q_\lambda(\theta)=\widehat g(\theta)^\top W_\lambda\widehat g(\theta)/2$, $J=D_\theta\widehat g$, and $W_\lambda=(\mathsf S+\lambda\Delta_p)^{-1}$. Conditional on each response, the pilot, covariance and protected update remain frozen during fitting. They are recomputed for another response. Let $\mathcal B=\{\theta:\|\theta-\beta\|\le r\}$ and suppose $X$ has full column rank. Define $B_{L,k}=n^{-1}\sum_i\|x_i\|^k$ and $B_{U,k}=N^{-1}\sum_\ell\|u_\ell\|^k$. Valid global bounds for $\|J\|$, $\|D^2\widehat g\|$ and $\|\widehat g\|$ are
\begin{align*}
 j&=\tfrac14\{B_{L,2}^2+(B_{L,2}+B_{U,2})^2\}^{1/2},\\
 H_g&=\tfrac14\{B_{L,3}^2+(B_{L,3}+B_{U,3})^2\}^{1/2},\\
 B_g&=\{B_{L,1}^2+(B_{L,1}+B_{U,1})^2\}^{1/2}.
\end{align*}
These follow from $|\sigma'|,|\sigma''|\le1/4$; sharper verified local bounds may replace them. Put
\begin{align*}
 M_r&=\max_i\{|x_i^\top\beta|+r\|x_i\|\},\quad
 a_r=\sigma'(M_r)\lambda_{\min}(X^\top X/n),\\
 w_-&=\{(1+\kappa)\lambda_{\max}(\mathsf S)\}^{-1},\quad
 w_+=\{(1-\kappa)\lambda_{\min}(\mathsf S)\}^{-1},\\
 b_r&=\min\{B_g,\|\widehat g(\beta)\|+jr\},\quad
 m=w_-a_r^2-w_+b_rH_g.
\end{align*}

\begin{proposition}[Local logistic curvature transfer]\label{prop:logistic-transfer}
If $m>0$ and $jw_+\|\widehat g(\beta)\|<mr$, every $\lambda\in[0,1]$ has a unique stationary point $\widehat\theta_\lambda$ in $\mathcal B$, the interior minimizer over that ball. At $\widehat\theta_0$ write $g_0=\widehat g(\widehat\theta_0)$, $J_0=J(\widehat\theta_0)$, $j_0=\|J_0\|$, $b_0=\|g_0\|$, and define
\begin{equation*}
 \widetilde\theta_\lambda=\widehat\theta_0-(J_0^\top W_\lambda J_0)^{-1}J_0^\top W_\lambda g_0,
 \quad q_\lambda=\|\widetilde\theta_\lambda-\widehat\theta_0\|.
\end{equation*}
If $\widetilde\theta_\lambda\in\mathcal B$, then
\begin{equation*}
 \|\widehat\theta_\lambda-\widetilde\theta_\lambda\|\le e_\lambda,
 \quad e_\lambda=\frac{w_+H_g}{m}\{b_0q_\lambda+(j+j_0/2)q_\lambda^2\}.
\end{equation*}
On any common event $\mathcal E$ where these conditions hold, let $R_{\mathcal E}(\lambda)=E[\mathbf1_{\mathcal E}\|\widehat\theta_\lambda-\beta\|^2]$ and $\widetilde G_{\mathcal E}(\lambda)=E[\mathbf1_{\mathcal E}\{\|\widehat\theta_0-\beta\|^2-\|\widetilde\theta_\lambda-\beta\|^2\}]$. Then
\begin{equation*}\begin{split}
 R_{\mathcal E}(0)-R_{\mathcal E}(\lambda)&\ge\widetilde G_{\mathcal E}(\lambda)\\
 &\quad-E[\mathbf1_{\mathcal E}\{2\|\widetilde\theta_\lambda-\beta\|e_\lambda+e_\lambda^2\}].
\end{split}\end{equation*}
\end{proposition}

\begin{proof}
The full Hessian is $J^\top W_\lambda J+\sum_a(W_\lambda\widehat g)_aD^2\widehat g_a$. Its first term is at least $w_-a_r^2I$, and the second has quadratic form bounded in absolute value by $w_+b_rH_g$. Thus it is at least $mI$. On the boundary, $(\theta-\beta)^\top\nabla Q_\lambda(\theta)\ge mr^2-r\|\nabla Q_\lambda(\beta)\|>0$. A compact-ball minimizer is interior and unique. At the affine surrogate, subtract its zero affine gradient from the exact gradient. Taylor bounds $\|J(\widetilde\theta)-J_0\|\le H_gq$ and $\|\widehat g(\widetilde\theta)-g_0-J_0(\widetilde\theta-\widehat\theta_0)\|\le H_gq^2/2$ give gradient norm at most $w_+H_g\{b_0q+(j+j_0/2)q^2\}$. Strong monotonicity divides this by $m$. Expanding squared loss proves the risk bound.
\end{proof}

The curvature conditions are sufficient and can be conservative. Improvement requires the surrogate gain to exceed the nonlinear budget, which vanishes for affine moments. A numerical output in the certified ball with Euclidean gradient error $\varepsilon$ is within $\varepsilon/m$ of the root; locating that output in the ball and controlling complementary-event loss remain separate requirements.

\par\medskip\noindent\textbf{Exact covariance reduction for the design in Theorem~\ref{thm:logistic-improvement}.}\par\nobreak\smallskip
Write $\rho=n/N$, $v=1/6$, $\ell=ap$, $\upsilon=au$, $b=\upsilon-\ell<0$, and
\begin{align*}
 C&=a^2p(2/9)+\eta,\quad V_L=a^2p(1-dp),\quad V_U=a^2u(1-du),\\
 D&=v^2(V_L+\rho V_U)+\eta,\quad k=v^2(V_U-V_L)<0,\\
 k_p&=\frac{\operatorname{clip}_{[-\kappa,\kappa]}\{k(C^{-1}+D^{-1})\}}{C^{-1}+D^{-1}},
 \quad D_0=\ell^2D+b^2C.
\end{align*}
The deterministic scalar path used in that theorem is
\begin{equation}
 \phi_\lambda=\frac{\ell(\ell D+\lambda\upsilon k_p)}{D_0+\lambda\upsilon^2k_p},
 \quad b_\lambda=a^{-1}\operatorname{logit}(1/2+v\phi_\lambda).
 \label{eq:logistic-exact-path}
\end{equation}
The scalar $b$ in these constants is distinct from the positive path coefficient $b_\lambda$.

\begin{proof}[Proof of Theorem~\ref{thm:logistic-improvement}]
On $\mathcal E$, let $s_j=2K_j-3\in\{-1,1\}$ and use probability coordinates $z_j=\sigma(a\theta_j)$. The pilot is $t_j=s_j\log2/a$, and the exact moments are $m_1=\ell(z-\mathbf1/2-vs)$, $m_2=b(z-\mathbf1/2)$. The mean outcome score is zero, its covariance is $a^2p(2/9)I_d$, and its cross covariance with the prediction score is zero. Global centering gives
\begin{equation*}
 C_L=a^2v^2(pI_d-p^2ss^\top),\qquad
 C_U=a^2v^2(uI_d-u^2ss^\top).
\end{equation*}
The off-diagonal terms must be retained. In direction $s/\sqrt d$, the joint covariance reduces to $\operatorname{diag}(C,D)$ and the raw auxiliary difference to $k$. Orthogonal directions have difference $a^2v^2(u-p)$. Whitening and clipping preserve these subspaces. In coordinates $z-\mathbf1/2$, the right-hand side lies in direction $s$. Its full normal equation therefore uses the two-by-two covariance $\operatorname{diag}(C,D)+\lambda k_p(1,-1)(1,-1)^\top$, giving $z=\mathbf1/2+v\phi_\lambda s$.

Protection gives $|k_p|\le\kappa CD/(C+D)$ and
\[ (C+D)D_0-\upsilon^2CD=(\ell D-bC)^2\ge0. \]
Thus the denominator in \eqref{eq:logistic-exact-path} is at least $(1-\kappa)D_0>0$. Its numerator is positive since $k_p\ge-\kappa D$ and $\ell>\upsilon$. Denominator minus numerator is $b^2C+\lambda\upsilon b k_p>0$, so $0<\phi_\lambda<1$. Moreover
\begin{equation*}
 \phi'_\lambda=-\frac{\upsilon k_pF_0}{(D_0+\lambda\upsilon^2k_p)^2}<0,
 \quad F_0=\ell b(\ell D-bC)<0.
\end{equation*}
The unrestricted quadratic minimizer in $z$ is inside $(0,1)^d$, so its logit is the unique finite global logistic fit. Its parameter coordinates are $s_jb_\lambda$ with positive strictly decreasing $b_\lambda$.

At truth $\beta$, let $q_j=\sigma(a\beta_j)$. The group counts are independent binomials with three observations. Hence $P_\beta(\mathcal E)=\prod_j3q_j(1-q_j)$ and $E_\beta(s_j\mid\mathcal E)=2q_j-1=\tanh(a\beta_j/2)$. Expanding squared parameter error gives
\begin{equation*}
 R^c_\beta(\lambda)=db_\lambda^2-2b_\lambda M_\beta+\|\beta\|^2,
 \quad M_\beta=\sum_j\beta_j\tanh(a\beta_j/2).
\end{equation*}
Taking the difference proves the theorem. Finally $0\le M_\beta\le a\|\beta\|^2/2$ gives its explicit improvement neighborhood.
\end{proof}

At the reference truth, the successful exact-pilot event has probability $(3/4)^d$. A covariance ridge does not resolve separation of the unpenalized pilot \citep{albert1984existence}. A score-tolerance routine may nevertheless return a large finite approximation on separated samples; the exact theorem does not cover those returns. Conversely a positive-probability numerical failure scored with infinite loss prevents a finite unconditional MSE comparison. The conditional theorem neither changes the estimator nor removes failures from the experimental denominators.

The risk formula also clarifies the strength question. With $\tau_\beta=M_\beta/d$, the optimum solves $b_{\lambda_*}=\operatorname{clip}_{[b_1,b_0]}(\tau_\beta)$. On $\beta=t\mathbf1$, $t\ge0$, the quantity $t\tanh(at/2)$ ranges continuously from zero to infinity. Consequently both endpoint and interior conditional expected-risk optima occur. Over $\tau_\beta\in[b_1,b_0]$, the minimax strength instead solves $b_{\lambda^\dagger}=(b_0+b_1)/2$, since excess risk is $d(b_\lambda-\tau_\beta)^2$. It need not be exactly $1/2$. This is a conditional-risk example, not a general unconditional endpoint or selection-impossibility theorem.

\par\medskip\noindent\textbf{A numerical-fit certificate for the exact logistic design.}\par\nobreak\smallskip

The exact design permits a global certificate beyond the local curvature proposition. Use $Q_\lambda=\widehat g^\top W_\lambda\widehat g/2$. For a numerical return $\bar\theta_\lambda$ on $\mathcal E$, suppose its probabilities $\bar z_j=\sigma(a\bar\theta_{\lambda,j})$ lie in $[\delta,1-\delta]$, where $0<\delta\le1/3$, and $\|\nabla_\theta Q_\lambda(\bar\theta_\lambda)\|\le\varepsilon_\lambda$. Put
\begin{equation*}
 \underline m_z=\frac{\ell^2+b^2}{(1+\kappa)\lambda_{\max}(\mathsf S)},\qquad
 d_\lambda=\frac{\varepsilon_\lambda}{a^2\delta^2(1-\delta)^2\underline m_z}.
\end{equation*}
Then $\|\bar\theta_\lambda-\widehat\theta_\lambda\|\le d_\lambda$. Indeed, in probability coordinates the objective is a strictly convex quadratic with Hessian $H_z^\top W_\lambda H_z\succeq\underline m_z I$, $H_z=(\ell I,bI)^\top$. The chain rule bounds its probability-gradient residual by $\varepsilon_\lambda/\{a\delta(1-\delta)\}$, and strong convexity bounds distance to its unique minimizer. That minimizer lies in $[1/3,2/3]^d$ by $0<\phi_\lambda<1$. The probability-to-parameter map $z\mapsto a^{-1}\operatorname{logit}(z)$ has Lipschitz constant $1/\{a\delta(1-\delta)\}$ on the connecting segment, proving the certificate. Thus it controls the global exact fit for this design, including active protection. The factor $1/2$ in $Q$ must be included when converting a program's gradient residual.

If the outputs at strengths $0,h$ exist on all of $\mathcal E$ and have uniform bounds $d_0,d_h$, their conditional risk difference is at least
\begin{equation*}
 \bar R^c_\beta(0)-\bar R^c_\beta(h)
 \ge R^c_\beta(0)-R^c_\beta(h)
       -\sum_{\lambda\in\{0,h\}}
        \{2d_\lambda\sqrt{R^c_\beta(\lambda)}+d_\lambda^2\}.
\end{equation*}
This follows by expanding each squared loss and applying Cauchy--Schwarz. A theoretical improvement larger than this numerical budget survives approximate solution. If certificates or returns fail on part of $\mathcal E$, the displayed conclusion requires a separate loss bound on that part; deleting failed samples changes the risk target. General logistic designs still require the local certificate and complementary-event control above. The existing production stopping checks alone do not establish these certificates.

\subsection{An observable global certificate for logistic fits}
\label{supp:logistic-global-certificate}
Fix the observed arrays and the positive definite frozen weight $W$, and use $Q(\theta)=\widehat g(\theta)^\top W\widehat g(\theta)/2$ on $\mathbb R^d$. Write $A(c)=n^{-1}\sum_i x_ix_i^\top\sigma'(x_i^\top c)$ and $h_L=B_{L,3}/4$, with $j,H_g$ as in the local curvature comparison above. For a numerical center $c$ and radius $r>0$, define
\begin{align*}
 a_r&=\lambda_{\min}\{A(c)\}-h_Lr,& b_r&=\|\widehat g(c)\|+jr,\\
 L_r&=a_rr-\|m_1(c)\|,& m_r&=\lambda_{\min}(W)a_r^2-\lambda_{\max}(W)b_rH_g,\\
 \varepsilon_c&=\|J(c)^\top W\widehat g(c)\|.
\end{align*}
\begin{proposition}[Global logistic solution certificate]\label{prop:logistic-global-certificate}
If $a_r,L_r,m_r>0$ and $Q(c)<\lambda_{\min}(W)L_r^2/2$, then $Q$ has a unique finite global minimizer $\widehat\theta$ in the open ball $B(c,r)$, and $\|\widehat\theta-c\|\le\varepsilon_c/m_r$. Every $z$ in that ball also satisfies $\|z-\widehat\theta\|\le\|\nabla Q(z)\|/m_r$.
\end{proposition}
\begin{proof}
The bound $|\sigma''|\le1/4$ gives $\|A(\theta)-A(c)\|\le h_L\|\theta-c\|$, so $A(\theta)\succeq a_rI$ in the ball. Everywhere else $A(\theta)\succeq0$. For any unit vector $v$ and $s\ge r$,
\[
 v^\top m_1(c+sv)=v^\top m_1(c)+\int_0^s v^\top A(c+tv)v\,dt\ge L_r.
\]
Consequently every point outside the open ball has $Q\ge\lambda_{\min}(W)L_r^2/2>Q(c)$. Continuity on the closed ball supplies an interior global minimum. Within the ball, the full Hessian is bounded below by $m_rI$: the first block of $J$ is $A$, while the residual curvature term has norm at most $\lambda_{\max}(W)b_rH_g$. Strong convexity gives uniqueness. Integrating this Hessian bound between the minimizer and any point in the ball gives the stated gradient-to-distance bounds.
\end{proof}
The certificate uses observed quantities, including the chosen radius, and excludes better solutions outside the ball without assuming global convexity of the complete criterion. It is sufficient, not necessary. A failed certificate does not imply an incorrect fit. Strict numerical certification requires error-controlled bounds for the displayed quantities; ordinary floating-point checks alone are diagnostics.

For strengths $0,h$, suppose the certificates hold on a common event $\mathcal E$, with numerical returns $c_\lambda$ and bounds $d_\lambda=\varepsilon_\lambda/m_\lambda$. Put $D_\lambda^2=E[\mathbf1_{\mathcal E}d_\lambda^2]$ and $R_{\mathcal E}(\lambda)=E[\mathbf1_{\mathcal E}\|\widehat\theta_\lambda-\beta\|^2]$. Expanding squared loss and applying Cauchy--Schwarz shows that the numerical risk gain on $\mathcal E$ is at least the exact gain minus
\[
 \sum_{\lambda\in\{0,h\}}\{2\sqrt{R_{\mathcal E}(\lambda)}D_\lambda+D_\lambda^2\}.
\]
An unconditional claim additionally requires integrable loss control on $\mathcal E^c$; a small failure probability alone is insufficient. This transfers an established risk margin and does not prove general-design risk dominance.

\section{Strength-path analysis}
\label{app:path-proofs}
\subsection{The half-strength default}
\label{subsec:default-strength}
\label{subsec:strength-default}

With $\kappa=1/2$, half strength bounds relative covariance eigenvalues by $[3/4,5/4]$. Proposition~\ref{prop:fs-convexity-balance} characterizes its scalar population balance. The following quadratic decision model gives a separate minimax-regret interpretation.

For affine regression moments $\widehat g(\beta)=H\beta-b$, put $E_p=\mathsf S^{-1/2}\Delta_p\mathsf S^{-1/2}$ and define\begin{equation}G_Y=\mathsf S^{-1/2}H,\quad L_Y=(G_Y^\top G_Y)^{-1}G_Y^\top,
 \quad P_Y=I-G_YL_Y,\quad z_Y=\mathsf S^{-1/2}(H\widehat\beta_0-b).\label{eq:fs-path-coordinates}\end{equation}
Let $q_k$ be a complete orthonormal eigenbasis of $P_YE_pP_Y$, with eigenvalues $\rho_k$. The frozen normal equations give the exact representation
\begin{equation}\widehat\beta_\lambda=\widehat\beta_0+
       \sum_k\frac{\lambda}{1+\lambda\rho_k}u_k,
 \qquad u_k=(L_YE_pq_k)(q_k^\top z_Y),\quad |\rho_k|\le\kappa.\label{eq:fs-exact-strength-path}\end{equation}
Since $|\rho_k|\le\kappa=1/2$, all denominators are bounded away from zero on $[0,1]$. The coordinates depend jointly on the observed responses, pilot and covariance.

Section~\ref{subsec:new-default-experiment} compares the default with both uncalibrated GMM and an independently trained fixed strength. Derivative bounds, the risk expansion and grid-error analysis are given in Supplementary Section~\ref{app:path-proofs}.

\begin{proposition}[Midpoint as a minimax-regret strength]
\label{prop:fs-minimax-midpoint}
Fix the design and the predictor. Consider a quadratic decision model with excess-risk curves $\mathcal E_\zeta(\lambda)=B(\lambda-\zeta)^2$, where $B>0$ is known and the unknown optimal strength $\zeta$ lies in $[l,u]\subseteq[0,1]$. Then for every deterministic $\lambda\in[0,1]$ the worst-case excess risk over that interval is
\begin{equation}W(\lambda)=\sup_{\zeta\in[l,u]}\mathcal E_\zeta(\lambda)=B\max\{(\lambda-l)^2,(\lambda-u)^2\},\label{eq:fs-minimax-worst-case}\end{equation}
which is uniquely minimized at $\lambda=(l+u)/2$ with value $B(u-l)^2/4$. In particular, if the uncertainty interval is symmetric about $1/2$ then the minimax-regret strength is exactly $1/2$ and its worst-case excess risk is $B(u-l)^2/4$. If instead $\zeta$ has a prior with mean $m$ and variance $A_\pi$ on $[l,u]$, then $\mathbb E_\pi\{\mathcal E_\zeta(\lambda)\}=B\{(\lambda-m)^2+A_\pi\}$, minimized at $\lambda=m$; a uniform prior gives the midpoint again, with average excess $B(u-l)^2/12$.
\end{proposition}

For paths affine in strength, including zero scalar feedback, risk is quadratic. When its quadratic coefficient is positive and its unconstrained minimizer lies in $[0,1]$, excess risk relative to the best admissible strength has the displayed squared-distance form. If the unconstrained minimizer lies outside this interval, boundary-optimal excess risk generally also has a linear term; a zero quadratic coefficient is a separate degenerate case. The CC-PPI path generally has a rational form. Response-based tuning additionally requires control of risk-estimation error after selection \citep{tibshiranirosset2019,cauchois2021}.

Supplementary Section~\ref{supp:strength-selection} bounds the error of a paired parameter-SURE selector. Here we compare the fixed default with the best deterministic strength on the same path.

\subsection{Risk along the calibration path}
\label{subsec:remaining-path-space}

\label{subsec:regret}

Let
\begin{equation*}R_*^{\rm path}=\inf_{0\le\lambda\le1}R(\lambda),\qquad
 G_h=R(h)-R_*^{\rm path},\qquad
 G_0=R(0)-R_*^{\rm path}.\end{equation*}
Conditional on the design and predictor, $G_h$ is the deterministic excess risk of the default relative to the best fixed strength. Continuity and an integrable protected-path envelope ensure attainment of the infimum. The definitions imply\begin{equation*}R(h)-R(0)=G_h-G_0,\qquad G_h\ge0,\quad G_0\ge0.\end{equation*}
Consequently $G_h\le G_0$ holds exactly when $R(h)\le R(0)$; it is not automatic. In the adverse regime of Theorem~\ref{thm:fs-primitive-improvement}, $G_0=0$ and $G_h>0$. The recovered-gain fraction $1-G_h/G_0$ is interpretable only when $G_0>0$, and lies in $[0,1]$ only if the default improves on the original fit.

The path coordinates and exact representation are equations~\eqref{eq:fs-path-coordinates}--\eqref{eq:fs-exact-strength-path} above. For verification, first derive \[\widehat\beta_\lambda-\widehat\beta_0
 =\lambda L_Y(I+\lambda E_pP_Y)^{-1}E_pz_Y\], use $P_Yz_Y=z_Y$, and diagonalize the restriction to the residual subspace. All coordinates depend on the same observed response. Zero residual eigenvalues are retained: $\rho_k=0$ can still have $u_k\ne0$, contributing the linear term $\lambda u_k$. There is no independent-weight substitution or local approximation.

The remaining subsections derive the path's derivative envelopes, realized-loss stationary points, continuum regret and grid-error bounds. The response dependence in this representation is retained in the following risk approximation.

\par\medskip\noindent{\bfseries Risk approximation and regret bounds}\par\nobreak\smallskip

\label{subsec:convergent-risk}

The exact path also permits a finite-sample risk approximation with an explicit remainder. This controls approximation in strength at fixed sample sizes, separately from any asymptotic expansion in $n$. Abbreviate $G=G_Y$, $L=L_Y$, $P=P_Y$, $E=E_p$, $z=z_Y$, $e_0=\widehat\beta_0-\beta$, and $\gamma=\|E\|\le\kappa<1$. For an integer $m\ge1$, define
\begin{align*}
 v_j^{[m]}&=(-1)^{j-1}L(EP)^{j-1}Ez,\quad 1\le j\le m,\\
 \widehat\beta^{[m]}_\lambda&=\widehat\beta_0+\sum_{j=1}^m\lambda^jv_j^{[m]},\qquad
 U_m=\frac{\|(EP)^mEz\|}{s_{\min}(G)(1-\gamma)}.
\end{align*}
The coefficient itself does not depend on the truncation order once $m\ge j$. Let $p_m(\lambda)=\mathbb E\|\widehat\beta^{[m]}_\lambda-\beta\|^2$, a polynomial of degree at most $2m$, and set
\begin{equation}\epsilon_m(\lambda)=2\lambda^{m+1}\mathbb E\left[
  \left(\|e_0\|+\sum_{j=1}^m\lambda^j\|v_j^{[m]}\|\right)U_m\right]
  +\lambda^{2m+2}\mathbb E U_m^2,
 \qquad\bar\epsilon_m=\epsilon_m(1).\label{eq:fs-risk-truncation-budget}\end{equation}

\begin{proposition}[Improvement magnitude and fixed-path regret]
\label{prop:fs-risk-truncation}
Suppose $\mathbb E\|e_0\|^2<\infty$ and $\mathbb E Z^2<\infty$, where $Z=\|z\|/s_{\min}(G)$. Then, uniformly for $\lambda\in[0,1]$,
\begin{equation*}|R(\lambda)-p_m(\lambda)|\le\epsilon_m(\lambda)\le\bar\epsilon_m,
 \qquad\bar\epsilon_m\longrightarrow0\end{equation*}
geometrically in $m$. In particular, with $G_h=R(h)-\inf_{[0,1]}R(\lambda)$,
\begin{equation*}\begin{gathered}
 p_m(0)-p_m(h)-\epsilon_m(h)\le R(0)-R(h),\\
 R(0)-R(h)\le p_m(0)-p_m(h)+\epsilon_m(h).
\end{gathered}\end{equation*}
\begin{equation}\begin{gathered}
 \max\{0,p_m(h)-\min_{[0,1]}p_m-\epsilon_m(h)-\bar\epsilon_m\}\le G_h,\\
 G_h\le p_m(h)-\min_{[0,1]}p_m+\epsilon_m(h)+\bar\epsilon_m.
\end{gathered}\label{eq:fs-fixed-path-regret-interval}\end{equation}
The moment assumptions hold at fixed full-rank design, positive ridge and Gaussian errors for the actual estimator.
\end{proposition}

Increasing the truncation order reduces the risk-approximation error while retaining dependence between the weight and response. Numerical comparisons evaluate the exact path by Monte Carlo. Supplementary Section~\ref{app:convergent-proof} proves the remainder bound; Proposition~\ref{prop:fs-nonconvex-example} gives a nonconvex realized-loss example.

\subsubsection{Proof of the midpoint minimax-regret proposition}
\label{supp:midpoint-proof}

\begin{proof}
The map $\zeta\mapsto(\lambda-\zeta)^2$ is convex, so its supremum over an interval is attained at an endpoint; this gives \eqref{eq:fs-minimax-worst-case}. Both branches are convex and nonnegative, so their maximum is convex; it is invariant under the reflection $\lambda\mapsto l+u-\lambda$ that exchanges the branches, strictly decreasing below $(l+u)/2$ and strictly increasing above it, so the minimum is unique at $(l+u)/2$ and equals $B((u-l)/2)^2$. For the prior statement, expand $(\lambda-\zeta)^2=(\lambda-m)^2+2(\lambda-m)(m-\zeta)+(\zeta-m)^2$; the cross term has zero $\pi$-expectation and the last has expectation $A_\pi$.
\end{proof}

\subsection{Feasible strength selection and its risk budget}
\label{supp:strength-selection}

In the Gaussian regression setting with known $\sigma^2$, define the paired parameter-risk criterion
\begin{equation}\widehat q_\lambda=\|\widehat\beta_\lambda-t\|^2-\|\widehat\beta_h-t\|^2
 +2\sigma^2\langle C,D_Y(\widehat\beta_\lambda-\widehat\beta_h)\rangle_F.\label{eq:fs-paired-sure}\end{equation}
For every fixed $\lambda$, its expectation is $R(\lambda)-R(h)$. Let $d_\lambda=\|\widehat\beta_\lambda-\beta\|^2-\|\widehat\beta_h-\beta\|^2$, $e_\lambda=\widehat q_\lambda-d_\lambda$, and $v_\lambda=\partial_\lambda\widehat\beta_\lambda$. Under the indicated differentiation and moment conditions,
\begin{align}
 Z_\lambda=\partial_\lambda e_\lambda
 &=-2\{\epsilon^\top C^\top v_\lambda
              -\sigma^2\operatorname{tr}[(D_Yv_\lambda)C^\top]\},\label{eq:fs-selection-noise}\\
 \mathcal B=\int_0^1\mathbb EZ_\lambda^2\,d\lambda
 &=4\int_0^1\left\{\sigma^2\mathbb E(v_\lambda^\top CC^\top v_\lambda)
 +\sigma^4\mathbb E\operatorname{tr}([(D_Yv_\lambda)C^\top]^2)\right\}d\lambda.
 \label{eq:fs-selection-budget}
\end{align}
The second-order Stein identity \citep{belleczhang2021} gives the last equality. Its matrix-square trace need not be nonnegative pointwise; replacing it by a squared Frobenius norm yields a valid, potentially loose upper bound.

\begin{proposition}[Default-relative selection protection]\label{prop:fs-selection-protection}
Let a finite grid $\mathcal G\subset[0,1]$ contain $h$, with the grid and $\tau>0$ deterministic conditional on the designs, and assume $\mathcal B<\infty$. The selector
\begin{equation}\widehat\lambda_\tau\in\arg\min_{\lambda\in\mathcal G}
       \{\widehat q_\lambda+\tau|\lambda-h|\}\label{eq:fs-penalized-selector}\end{equation}
satisfies
\begin{align*}
 R_{\rm sel}&\le R(h)+\mathcal B/(4\tau),\\
 R_{\rm sel}&\le\min_{\lambda\in\mathcal G}
              \{R(\lambda)+\tau|\lambda-h|\}+\mathcal B/(4\tau).
\end{align*}
\end{proposition}
\begin{proof}
Since $e_h=0$, integration and Cauchy--Schwarz give $|e_\lambda|\le\sqrt{|\lambda-h|}\,A(Y)$, with $A(Y)^2=\int_0^1Z_s(Y)^2ds$ and $\mathbb EA^2=\mathcal B$. The selected objective is no larger than zero at $h$, so its loss contrast is at most $-\tau d+\sqrt d A\le A^2/(4\tau)$, where $d=|\widehat\lambda_\tau-h|$. This proves the first bound. Comparing the objective to any fixed candidate and using $\mathbb Ee_\lambda=0$ gives the second by the same square completion.
\end{proof}

This controls possible deterioration; it is not a lower bound on selection cost or a theorem that tuning is harmful. For unknown noise variance, a plug-in $\widehat\sigma^2$ adds $2(\widehat\sigma^2-\sigma^2)\langle C,D_Y(\widehat\beta_\lambda-\widehat\beta_h)\rangle_F$ to the criterion error. Its joint moments with the learned maps must be retained; the unbiased fixed-candidate argument cannot be reused unchanged. Numerical derivative error likewise needs an explicit allowance. No feasible choice is allowed to use the unknown $\beta$.

In an eight-law, 64-design validation with $d=2,5$, $n=40,160$, $N=5n$, and two predictor discrepancies, direct known-noise parameter-SURE selection had higher average risk than half in four laws and lower risk in four. The paired differences were small relative to the design-level Monte Carlo errors. In this comparison, response-based selection gave no consistent improvement over the fixed half-strength rule.

\par\medskip\noindent{\bfseries An explicit selection budget relative to the available gain}\par\nobreak\smallskip

\label{subsec:selection-scale}

The size of $\mathcal B$, not just its finiteness, determines what the selection inequality can certify. Let $I_\ell=[h-\ell,h+\ell]$, $0<\ell\le1/2$, and restrict the candidate grid to this deterministic interval. Define $\mathcal B_{I_\ell}=\int_{I_\ell}\mathbb E Z_s^2\,ds$. All derivatives below are of the complete response-dependent procedure. Write $D_t v_s$ for the derivative along $Y=Xt+r$, holding the OLS residual $r$ fixed.

\begin{proposition}[Design-scaled selection budget]
\label{prop:fs-selection-scale}
In the known-variance Gaussian model, suppose $\lambda_{\min}(X^\top X/n)\ge c_X>0$, and let
\begin{equation}a_I^2=\int_{I_\ell}\mathbb E\|v_s\|^2ds,\qquad
 b_I^2=\int_{I_\ell}\mathbb E\|D_t v_s\|_F^2ds.\label{eq:fs-selection-derivative-budgets}\end{equation}
Then
\begin{equation*}\mathcal B_{I_\ell}\le\overline{\mathcal B}_I
 :=\frac{4\sigma^2 a_I^2}{nc_X}
   +\frac{4\sigma^4 b_I^2}{n^2c_X^2}.\end{equation*}
For a finite grid $\mathcal G\subset I_\ell$ containing $h$, deterministic conditional on the designs, the selector in \eqref{eq:fs-penalized-selector} satisfies
\begin{align}
 R_{\rm sel}&\le R(h)+\overline{\mathcal B}_I/(4\tau),\label{eq:fs-scale-default}\\
 R_{\rm sel}-\min_{s\in\mathcal G}R(s)
 &\le\tau\ell+\overline{\mathcal B}_I/(4\tau).
 \label{eq:fs-scale-oracle}
\end{align}
If a positive deterministic upper budget $B_I^{\rm up}\ge\overline{\mathcal B}_I$ is available, $\tau=\tfrac12\sqrt{B_I^{\rm up}/\ell}$ makes the right side of \eqref{eq:fs-scale-oracle} at most $\sqrt{\ell B_I^{\rm up}}$.
\end{proposition}
\begin{proof}
Put $A_X=(X^\top X)^{-1}$. Then $CC^\top=A_X$, $D_t v_s=(D_Yv_s)X$, and $(D_Yv_s)C^\top=(D_t v_s)A_X$. Apply $\|A_X\|\le1/(nc_X)$ and $|\operatorname{tr}(B^2)|\le\|B\|_F^2$ to \eqref{eq:fs-selection-budget}. The anchored path proof of Proposition~\ref{prop:fs-selection-protection} works with integration over $I_\ell$. Finally $|s-h|\le\ell$ and balancing the two terms proves the assertion.
\end{proof}

For a model sequence with established fixed-default gain of order $g_n>0$, the condition $\ell B_I^{\rm up}=o(g_n^2)$ makes the oracle remainder $o(g_n)$. This rate has not been established for the current target. The budgets in \eqref{eq:fs-selection-derivative-budgets} are population quantities; a single response or a plug-in Monte Carlo mean does not provide a deterministic upper bound. Alternatively, a known margin $R(0)-R(h)\ge G>0$ and $B_I^{\rm up}/(4\tau)<G$ ensure strict improvement over the original fit by \eqref{eq:fs-scale-default}. The same anchored envelope bounds the unpenalized known-noise selector's remainder by $\sqrt{\ell\mathcal B_{I_\ell}}$, so penalization has no general superiority guarantee. Unknown variance and numerical derivative error require the joint allowances in Supplementary Section~\ref{app:risk-proofs}.

\subsubsection{Parameter SURE and its variance identity}

Let $g_\lambda=\widehat\beta_\lambda-t$ and $F_\lambda=C^\top g_\lambda$. Gaussian integration by parts gives
\begin{equation*}R(\lambda)=\sigma^2\operatorname{tr}(CC^\top)
   +\mathbb E\{\|g_\lambda\|^2+2\sigma^2\langle C,D_Yg_\lambda\rangle_F\}.\end{equation*}
Subtracting the equation at $h$ proves unbiasedness of \eqref{eq:fs-paired-sure}. This is parameter loss, not prediction loss in the observation space. The subtraction cancels the derivative of $t$.

Set $w_\lambda=\widehat\beta_\lambda-\widehat\beta_h$. Expanding the two losses around $t$ gives exactly
\begin{equation}e_\lambda=-2\{\epsilon^\top C^\top w_\lambda
        -\sigma^2\operatorname{tr}[(D_Yw_\lambda)C^\top]\}.\label{eq:app-exact-criterion-error}\end{equation}
Differentiate in strength, with all covariance operations frozen in strength but not in response, to get \eqref{eq:fs-selection-noise}. For a square-integrable vector map $F(Y)$ with square-integrable weak first derivative, the second-order Stein identity \citep{belleczhang2021} is
\begin{equation}\mathbb E\{\epsilon^\top F-\sigma^2\operatorname{div}F\}^2
 =\sigma^2\mathbb E\|F\|^2+\sigma^4\mathbb E\operatorname{tr}[(D_YF)^2].\label{eq:app-second-order-stein}\end{equation}
Apply it to $F=C^\top v_\lambda$. Cyclicity of trace gives the matrix-square trace in \eqref{eq:fs-selection-budget}. Its replacement by the Frobenius norm is valid because $|\operatorname{tr}(A^2)|\le\|A\|_F^2$. Supplementary Section~\ref{app:fit-regularity} supplies the needed Gaussian moments and weak derivatives. No classical second response derivative of the spectral clipping map is needed.

\subsubsection{Detailed selection inequality}

Because $e_h=0$ and the error path is absolutely continuous in strength,
\begin{equation*}|e_\lambda|\le\sqrt{|\lambda-h|}\,A(Y),\qquad
 A(Y)^2=\int_0^1Z_s(Y)^2ds,\quad\mathbb EA^2=\mathcal B.\end{equation*}
The finite-grid selector is measurable after a deterministic tie-breaking rule. Comparing its objective to $h$ and writing $d=|\widehat\lambda_\tau-h|$ gives
$d_{\widehat\lambda_\tau}\le-\tau d-e_{\widehat\lambda_\tau}\le-\tau d+\sqrt d A\le A^2/(4\tau)$. This proves the default-relative bound. Comparison to any fixed $s\in\mathcal G$ instead gives
\begin{equation*}d_{\widehat\lambda_\tau}-d_s
 \le\tau|s-h|-\tau d+e_s-e_{\widehat\lambda_\tau}.\end{equation*}
Taking expectations, using $\mathbb Ee_s=0$ and completing the same square proves the second bound. There is no grid-cardinality factor because the common anchored path controls every candidate simultaneously. It does not prove that the bound is tight. At $\tau=0$, the same argument and $d\le1/2$ give $R_{\rm sel}\le\min_{s\in\mathcal G}R(s)+\sqrt{\mathcal B/2}$.

For a purely discrete path, let ordered grid points contain $h$. Define $A_{\rm grid}^2=\sum_j(e_{s_{j+1}}-e_{s_j})^2/(s_{j+1}-s_j)$. Weighted Cauchy--Schwarz proves the same envelope at grid points. Each error increment has second moment given by \eqref{eq:app-second-order-stein} with $F=C^\top(\widehat\beta_{s_{j+1}}-\widehat\beta_{s_j})$ and factor four. Thus the selection proof can also be formulated with the discrete expected budget. For a differentiable path it is at most $\mathcal B$.

\subsubsection{Unknown variance and numerical derivative error}

For a response-based $\widehat\sigma^2$, define $\widetilde q_\lambda$ by substitution in \eqref{eq:fs-paired-sure}. Its error is
\begin{align*}
 \widetilde e_\lambda&=e_\lambda+2(\widehat\sigma^2-\sigma^2)
          \operatorname{tr}[(D_Yw_\lambda)C^\top],\\
 \widetilde Z_\lambda&=Z_\lambda+2(\widehat\sigma^2-\sigma^2)
          \operatorname{tr}[(D_Yv_\lambda)C^\top].
\end{align*}
The plug-in is fixed with respect to strength, not independent of response. Its fixed-candidate mean error need not vanish. Nevertheless, the pointwise default comparison still proves
\begin{equation}R_{\rm sel,plugin}\le R(h)+
 \mathbb E\left\{\frac{\int_0^1\widetilde Z_s^2ds}{4\tau(Y)}\right\}\label{eq:app-plugin-protection}\end{equation}
for any positive response-dependent penalty with the displayed expectation finite. For deterministic $\tau$, a valid variance budget is
\begin{equation*}\mathbb E\int_0^1\widetilde Z_s^2ds
 \le2\mathcal B+8\int_0^1\mathbb E\{(\widehat\sigma^2-\sigma^2)^2
                  \operatorname{tr}[(D_Yv_s)C^\top]^2\}ds.\end{equation*}
These are joint expectations; factoring them without independence is invalid. With $n>d$, the OLS estimate $\widehat\sigma^2=\|r\|^2/(n-d)$ is feasible. Polynomial bounds ensure finiteness for a constant positive penalty. For a penalty proportional to $\widehat\sigma^2$, $n-d>2$ is a sufficient inverse-residual-moment condition at fixed design: the Gaussian residual radius has a chi-square distribution with $n-d$ degrees of freedom, and its inverse first moment is finite precisely above two. This does not restore the unbiased fixed-candidate oracle comparison.

Suppose an approximate derivative-based contrast is anchored at zero for $h$ and differs uniformly from the exact contrast by at most $b(Y)$ on the grid. The objective comparisons then add $\mathbb Eb$ to the default bound and $2\mathbb Eb$ to the known-noise fixed-candidate bound. They add $\mathbb Eb$ to \eqref{eq:app-plugin-protection}. For a fixed finite grid, almost-everywhere differentiability and an integrable difference-quotient envelope give $\mathbb Eb\to0$ as a central-difference step tends to zero. Polynomial weak-derivative bounds above provide domination for the actual Gaussian pipeline. A numerical step-halving audit is a diagnostic, not proof that $b=0$ at the chosen step.

Finally, $\langle C,D_Y\widehat\beta_\lambda\rangle_F
=\sum_{j=1}^d D_{C_j^\top}\widehat\beta_{\lambda,j}$. Thus $d$ response perturbation directions suffice for the parameter divergence. Every perturbed response requires recomputing the pilot, covariance, protection and final fit. Treating the weight as response-constant here would differentiate a different estimator.

\subsubsection{Path derivative envelopes and grid comparisons}
\label{supp:path-envelopes}

Writing $\dot\beta_\lambda=\sum_k u_k/(1+\lambda\rho_k)^2$ and $\ddot\beta_\lambda=-\sum_k2\rho_ku_k/(1+\lambda\rho_k)^3$, define
\begin{equation*}B_0=\sum_k\frac{\|u_k\|}{1-\kappa},\quad
 B_1=\sum_k\frac{\|u_k\|}{(1-\kappa)^2},\quad
 B_2=\sum_k\frac{2|\rho_k|\|u_k\|}{(1-\kappa)^3}.\end{equation*}
The realized loss curvature is bounded by the modulus $\mathcal M_Y=2B_1^2+2(\|\widehat\beta_0-\beta\|+B_0)B_2$, whose expectation bounds $\sup_\lambda|R''(\lambda)|$ when finite. This supplies the curvature constant of the two-piece regret certificate, and it does not by itself certify a small risk gap: a curvature-only bound about a nonstationary reference need not hold; Theorem~\ref{thm:fs-regret-certificate} states precisely when the derivative term vanishes.

Specifically, let a grid include both endpoints with maximum spacing $a$, assume $|R''|\le\mathcal M$, and let $\widehat g_j$ estimate $g_j=R(h)-R(\lambda_j)$. On any simultaneous event $|\widehat g_j-g_j|\le b_j$,
\begin{equation}\max\{0,\max_j(\widehat g_j-b_j)\}\le G_h
 \le\max_j(\widehat g_j+b_j)+\mathcal M a^2/8.\label{eq:fs-grid-regret-bound}\end{equation}
At an interior continuum minimizer $R'=0$; Taylor expansion to the nearest grid point gives the $\mathcal Ma^2/8$ term. Endpoint minimizers are already included. The event itself needs a valid probability assessment.

\subsubsection{Strength stationarity and conditional regret results}
\label{subsec:strength-optimality}

The exact strength path distinguishes two optimization problems. For a fixed response, let $\mathcal L_Y(\lambda)=\|\widehat\beta_\lambda-\beta\|^2$ denote realized loss. The risk $R(\lambda)=\mathbb E\mathcal L_Y(\lambda)$ instead averages the response-dependent coefficients of that path. A truth-assisted minimum of realized loss is not the best fixed strength for risk.

\begin{proposition}[Realized-loss stationarity polynomial]
\label{prop:fs-strength-stationarity}
Fix the design, predictor and response, and put $m=\operatorname{rank}(E_p)\le d$. The realized loss can be written $\mathcal L_Y=N_Y/B_Y^2$ with polynomials of degrees at most $2m$ and $m$, respectively, and $B_Y$ nonzero on $[0,1]$. Define
\begin{equation*}g_Y(\lambda)=N_Y'(\lambda)B_Y(\lambda)-2N_Y(\lambda)B_Y'(\lambda).\end{equation*}
Then $\mathcal L_Y'=g_Y/B_Y^3$, and $\deg g_Y\le3m-1$ for $m\ge1$. Unless the derivative is identically zero, the minimum realized loss is attained among the endpoints and interior real roots:
\begin{equation*}\min_{[0,1]}\mathcal L_Y
 =\min\{\mathcal L_Y(0),\mathcal L_Y(1),
          \mathcal L_Y(\lambda):g_Y(\lambda)=0,\ 0<\lambda<1\}.\end{equation*}
If $g_Y\equiv0$, the loss is constant on the path. This statement does not characterize the expected-risk optimum.
\end{proposition}
\begin{proof}
If $m=0$, the update and displacement vanish and the loss is constant. For $m\ge1$, factor $E_p=U_E A_E U_E^\top$ with $U_E$ having $m$ orthonormal columns. The exact displacement is
$\lambda L_Y U_E A_E(I_m+\lambda U_E^\top P_YU_E A_E)^{-1}U_E^\top z_Y$.
The determinant of the matrix in parentheses gives $B_Y$; the determinant identity with $I+\lambda P_YE_pP_Y\succ0$ shows it is nonzero. Its adjugate has degree at most $m-1$, so the error vector has a polynomial numerator of degree at most $m$. Squaring gives $N_Y$ of degree at most $2m$. Differentiation proves the derivative and degree bound. No common factor of $N_Y$ and $B_Y$ is assumed.
\end{proof}

Equivalently, the spectral coordinates of Section~\ref{subsec:remaining-path-space} give
\begin{equation*}\begin{aligned}
 \mathcal L_Y(\lambda)&=\|e_0\|^2+2\sum_k a_k h_k(\lambda)
 +\sum_{i,j}b_{ij}h_i(\lambda)h_j(\lambda),\\
 h_k(\lambda)&=\frac{\lambda}{1+\lambda\rho_k},\qquad
 a_k=e_0^\top u_k,\qquad b_{ij}=u_i^\top u_j.
 \end{aligned}\end{equation*}
All these coefficients are response dependent. For expected risk, dominated differentiation gives $R'=\mathbb E\mathcal L_Y'$, not a single polynomial equation of the preceding degree. Section~\ref{subsec:convergent-risk} supplies population polynomial approximations with a uniform remainder. Independent training and evaluation streams can approximate the best fixed strength numerically, but require simulation truth or a justified risk estimate.

\begin{theorem}[Two-piece deterministic regret bound]
\label{thm:fs-regret-certificate}
Let $R$ be twice continuously differentiable on $[0,1]$, and let $\lambda_*$ minimize it. Put
\begin{equation*}D_*=|R'(\lambda_*)|,\qquad M=\sup_{[0,1]}|R''(\lambda)|<\infty.\end{equation*}
Then
\begin{equation*}G_h\le D_*|h-\lambda_*|+\tfrac{M}{2}(h-\lambda_*)^2.\end{equation*}
For an interior minimizer $D_*=0$; at an endpoint the derivative term may be necessary.
\end{theorem}
\begin{proof}
Taylor's formula around $\lambda_*$ gives a linear term $R'(\lambda_*)(h-\lambda_*)$ and an integral remainder bounded in absolute value by $M(h-\lambda_*)^2/2$. An interior minimizer has zero derivative.
\end{proof}

The protected path provides integrable derivative and curvature bounds. Since $0\le h_k(\lambda)\le(1-\kappa)^{-1}$,
\begin{equation}D_*\le\mathbb E\left\{
 \frac{2\sum_k|a_k|}{(1-\kappa)^2}
 +\frac{2\sum_{i,k}|b_{ik}|}{(1-\kappa)^3}\right\},\label{eq:fs-derivative-bound}\end{equation}
and $M\le\mathbb E\mathcal M_Y$ for the realized curvature modulus in Supplementary Section~\ref{supp:path-envelopes}. These bounds express the deterministic strength comparison through the complete path coefficients.

For the paired criterion of Supplementary Section~\ref{supp:strength-selection}, let a finite grid $\mathcal G$, deterministic conditional on the designs, contain $h$ and define the gain and its estimation error by
\begin{equation*}\widehat g_\lambda=-\widehat q_\lambda,\qquad
 e_\lambda^{\rm risk}=\widehat g_\lambda-\{R(h)-R(\lambda)\},\qquad
 e_h^{\rm risk}=0.\end{equation*}
This is error around a population risk contrast, distinct from the realized-loss criterion error used in the selection proof.

\begin{proposition}[Conditional event-based regret interval]
\label{prop:fs-regret-certificate}
On $\mathcal E=\{\max_{\lambda\in\mathcal G}|e_\lambda^{\rm risk}|\le\epsilon\}$,
\begin{equation*}\max\{0,\max_{\mathcal G}\widehat g_\lambda-\epsilon\}
 \le G_h\le\max_{\mathcal G}\widehat g_\lambda+\epsilon+L_{\rm gap},
 \qquad L_{\rm gap}=\min_{\mathcal G}R-\inf_{[0,1]}R\ge0.\end{equation*}
If the grid contains both endpoints, has maximum spacing $a$, and $|R''|\le M$, then $L_{\rm gap}\le Ma^2/8$.
\end{proposition}
\begin{proof}
The grid gain is $G_{\mathcal G}=R(h)-\min_{\mathcal G}R$, so $G_h=G_{\mathcal G}+L_{\rm gap}$. On $\mathcal E$, $G_{\mathcal G}$ lies between $\max\widehat g-\epsilon$ and $\max\widehat g+\epsilon$. Nonnegativity gives the lower bound. At an interior continuum minimizer Taylor's formula to a nearest grid point gives $Ma^2/8$; endpoint minima are already on the grid.
\end{proof}

The event is a separate statistical requirement. If deterministic budgets $V_\lambda\ge\mathbb E|e_\lambda^{\rm risk}|^2$ are supplied, a union bound and Markov's inequality give a valid allowance
\begin{equation}\epsilon=\sqrt{\delta^{-1}\sum_{\lambda\in\mathcal G}V_\lambda},
 \qquad\Pr(\mathcal E)\ge1-\delta,\quad0<\delta<1.\label{eq:fs-regret-epsilon}\end{equation}
The interval uses the supplied deterministic budgets. Estimated noise variance and numerical derivatives contribute their corresponding allowances; Supplementary Section~\ref{app:risk-proofs} gives the full criterion-error identity.

\subsubsection{Residual-space derivation}

Fix one response realization and abbreviate $G=G_Y$, $L=L_Y$, $P=P_Y$, $E=E_p$, $y=\mathsf S^{-1/2}b$ and $z=G\widehat\beta_0-y$. The baseline least-squares equations give $G^\top z=0$, $Pz=z$, $Lz=0$. Put $\delta_\lambda=\widehat\beta_\lambda-\widehat\beta_0$ and
$w=(I+\lambda E)^{-1}(G\delta_\lambda+z)$. The weighted normal equations imply $G^\top w=0$, hence $Pw=w$. Multiplying $G\delta_\lambda+z=(I+\lambda E)w$ by $P$ and $L$ gives
\begin{equation*}z=(I+\lambda PEP)w,\qquad
 \delta_\lambda=\lambda LEw
   =\lambda LE(I+\lambda PEP)^{-1}z.\end{equation*}
Since $\|PEP\|\le\kappa<1$, the inverse exists throughout the path. Diagonalizing $PEP$ proves \eqref{eq:fs-exact-strength-path}. The equivalent pushed-through inverse expression above follows from $Pz=z$. All matrices in this proof are fixed only within this response realization, not across repeated responses.

Differentiating the scalar rational coefficients gives the stated first and second strength derivatives. Since $|1+\lambda\rho_k|\ge1-\kappa$,
\begin{equation}\sup_\lambda\|\widehat\beta_\lambda-\widehat\beta_0\|\le B_0,
 \quad\sup_\lambda\|\dot\beta_\lambda\|\le B_1,
 \quad\sup_\lambda\|\ddot\beta_\lambda\|\le B_2.\label{eq:app-path-envelopes}\end{equation}
These remain valid for singular $A_U-A_L$, active clipping and indefinite corrections.

\subsubsection{Realized-loss stationarity and expected-risk regret}

The stationarity polynomial and its degree bound follow from the proof of Proposition~\ref{prop:fs-strength-stationarity} above. The coefficients use the unknown truth, so this is not a feasible selector.

For expected risk, dominated differentiation of the protected path yields $R'=\mathbb E\mathcal L_Y'$ and $R''=\mathbb E\mathcal L_Y''$. The domination argument is given below. Taylor expansion at its minimizing $\lambda_*$ proves Theorem~\ref{thm:fs-regret-certificate}; the first-order term vanishes at an interior optimum but can remain at an endpoint. Bounding the spectral derivative uses $h_k(\lambda)\le(1-\kappa)^{-1}$, which yields \eqref{eq:fs-derivative-bound} after taking expectations.

For the event interval, let $g_\lambda=R(h)-R(\lambda)$ and $\widehat g_\lambda=-\widehat q_\lambda$. Fixed-candidate SURE unbiasedness gives $\mathbb E\widehat g_\lambda=g_\lambda$, but its error about this expected gain is distinct from the selection criterion's error about realized loss. The simultaneous event bounds the maximum grid gain. Adding $\min_{\mathcal G}R-\inf_{[0,1]}R$ gives the continuum regret. The nearest-grid-point Taylor argument supplies $Ma^2/8$ when endpoints are included. If valid deterministic second-moment budgets are provided, Markov's inequality and a union bound prove \eqref{eq:fs-regret-epsilon}. Without those budgets the interval is an event implication, not an automatically feasible confidence certificate. The term $-2(t-\beta)^\top w_\lambda$ in \eqref{eq:app-exact-criterion-error} remains after OLS residual orthogonality.

\subsubsection{Curvature and continuum risk}

Twice differentiating the realized loss gives
\begin{equation*}\partial_\lambda^2\|\widehat\beta_\lambda-\beta\|^2
 =2\|\dot\beta_\lambda\|^2
       +2\langle\widehat\beta_\lambda-\beta,\ddot\beta_\lambda\rangle.\end{equation*}
The triangle inequality and \eqref{eq:app-path-envelopes} bound its absolute value by $\mathcal M_Y=2B_1^2+2(\|\widehat\beta_0-\beta\|+B_0)B_2$. Under Gaussian errors its expectation is finite by Supplementary Section~\ref{app:fit-regularity}, or directly by bounding $\|L\|$, $\|z\|$ and the finite sums using the ridge. Dominated differentiation gives a continuous risk curve with $\sup_\lambda|R''(\lambda)|\le\mathbb E\mathcal M_Y$ and an attained minimum on $[0,1]$. Finiteness at fixed design does not imply a small curvature bound uniformly across designs.

\begin{proposition}[Realized loss need not be convex in strength]\label{prop:fs-nonconvex-example}
A frozen rank-one update satisfying the relative spectral cap need not give convex realized loss; clipping need not be active. This example concerns the frozen affine weighting class, rather than a constructed sample-derived CC-PPI pipeline. Take the frozen affine case $d=1$ with $S=I_2$, $H=(1,0)^\top$, $b=(3,5)^\top$, $\beta=0$, and protected correction $\Delta_p=T(-\tfrac15)T^\top$ with $T=(1,-1)^\top$. Then $S^{-1/2}\Delta_pS^{-1/2}$ has eigenvalues $0$ and $-\tfrac25$, within the cap $\kappa=\tfrac12$, and the exact path \eqref{eq:fs-exact-strength-path} is $\widehat\beta_\lambda=3-\lambda/(1-\tfrac15\lambda)$. Its squared loss $\mathcal L_Y(\lambda)=\widehat\beta_\lambda^2$ has $\mathcal L_Y''(0)=-\tfrac25<0$.
\end{proposition}

\begin{proof}
Here $G=S^{-1/2}H=(1,0)^\top$, $L=(1,0)$, $P=\operatorname{diag}(0,1)$, and the restricted correction has residual eigenvalue $\rho=-\tfrac15$. With $q=(0,1)^\top$ and $z=S^{-1/2}(H\widehat\beta_0-b)=(0,-5)^\top$, the residual coordinate is $u=(L\Delta_p q)(q^\top z)=-1$ and $\widehat\beta_0=Lb=3$. Equation~\eqref{eq:fs-exact-strength-path} reduces to the stated rational fit. Differentiating $\mathcal L_Y(\lambda)=\widehat\beta_\lambda^2$ at $\lambda=0$ using $\widehat\beta_0=3$, $\widehat\beta_0'=-1$, and $\widehat\beta_0''=-\tfrac25$ gives $\mathcal L_Y''(0)=2(\widehat\beta_0')^2+2\widehat\beta_0\widehat\beta_0''=2-\tfrac{12}{5}=-\tfrac25$.
\end{proof}

\begin{proposition}[Sufficient convexity and the half-strength balance]\label{prop:fs-convexity-balance}
Consider the scalar case $d=1$, where the exact path \eqref{eq:fs-exact-strength-path} has a single term $\widehat\beta_\lambda=\widehat\beta_0+\lambda u/(1+\lambda\rho)$ with $|\rho|\le\kappa\le1/2$. Put $e_0=\widehat\beta_0-\beta$ and condition on $\rho$:
\begin{equation*}A(\rho)=\mathbb E[e_0u\mid\rho],\qquad B(\rho)=\mathbb E[u^2\mid\rho].\end{equation*}
Then
\begin{equation}R'(\lambda)=2\,\mathbb E\left\{\frac{\lambda B(\rho)+(1+\lambda\rho)A(\rho)}{(1+\lambda\rho)^3}\right\},\label{eq:fs-scalar-derivative}\end{equation}
\begin{equation}R''(\lambda)=2\,\mathbb E\left\{\frac{B(\rho)(1-2\lambda\rho)}{(1+\lambda\rho)^4}-\frac{2\rho A(\rho)}{(1+\lambda\rho)^3}\right\}.\label{eq:fs-scalar-curvature}\end{equation}
If $\rho A(\rho)\le0$ almost surely and $\mathbb E u^2>0$, then $R$ is strictly convex on $[0,1]$: its unique minimizer is $\lambda=0$ when $R'(0)\ge0$, $\lambda=1$ when $R'(1)\le0$, and otherwise the unique root of $R'(\lambda)=0$. Under this condition, half strength is the unique minimizer exactly when
\begin{equation}\mathbb E\left\{\frac{\tfrac12 B(\rho)+(1+\tfrac12\rho)A(\rho)}{(1+\tfrac12\rho)^3}\right\}=0,\label{eq:fs-half-balance}\end{equation}
a specific population balance between the alignment $A(\rho)$ and the movement moment $B(\rho)$, not a consequence of equal nominal covariance weights.
\end{proposition}

\begin{proof}
Since $\widehat\beta_\lambda-\beta=e_0+\lambda u/(1+\lambda\rho)$,
\begin{equation*}R(\lambda)-R(0)=\mathbb E\left\{\frac{2\lambda e_0u}{1+\lambda\rho}+\frac{\lambda^2u^2}{(1+\lambda\rho)^2}\right\}.\end{equation*}
Conditioning on $\rho$ and differentiating under the expectation, justified by the growth bounds in Supplementary Section~\ref{app:fit-regularity}, gives \eqref{eq:fs-scalar-derivative} and \eqref{eq:fs-scalar-curvature}. Under $\rho A(\rho)\le0$ the term $-2\rho A(\rho)$ is nonnegative, while $B(\rho)\ge0$ and $1-2\lambda\rho>0$ for $\lambda<1$ because $|\rho|\le\kappa\le1/2$; hence $R''\ge0$, with strict positivity on the interior whenever $\mathbb E u^2>0$, so $R'$ is strictly increasing and $R$ is strictly convex. Strict convexity yields the unique root/boundary classification, and $R'(\tfrac12)=0$ is \eqref{eq:fs-half-balance}.
\end{proof}

The condition $\rho A(\rho)\le0$ is not automatic. In the example of Proposition~\ref{prop:fs-nonconvex-example}, $e_0=3$, $u=-1$, and $\rho=-\tfrac15$ give $\rho A(\rho)=\tfrac35>0$, consistent with $R''(0)=-\tfrac25<0$. Protection and the rank-one covariance-PPI form therefore do not by themselves deliver the sign of $\rho A(\rho)$.

If $\rho$ is deterministic, $\rho\equiv\rho_0$, then $A(\rho)$ and $B(\rho)$ are constants and the sign of $R'(\lambda)$ is that of $\lambda(B+\rho_0A)+A$, so $R$ is unimodal without any convexity assumption. In the nondegenerate case $B+\rho_0A>0$ this recovers the constant-feedback oracle $\lambda_*=\operatorname{clip}_{[0,1]}\{-A/(B+\rho_0A)\}$, a benchmark that is not established for the actual random sample covariance path.

For completeness, a fixed midpoint also has a conditional envelope bound without a risk-sign assumption. Let
$W_Y=\sum_k\|u_k\|/\{(1-|\rho_k|)|1+h\rho_k|\}$ and $\delta=\tfrac12(\mathbb EW_Y^2)^{1/2}$. Subtracting path coefficients gives $\|\widehat\beta_\lambda-\widehat\beta_h\|_{L^2}\le\delta$. The reverse triangle inequality in $L^2$ then yields
\begin{equation*}R(h)-R_*^{\rm path}\le R(h)-[\sqrt{R(h)}-\delta]_+^2.\end{equation*}
The bound holds in every fixed dimension but can be substantially larger than the observed gains; it supplies no uniform few-percent regret guarantee.

\subsubsection{Convergent risk expansion}
\label{app:convergent-proof}
\label{supp:truncation-proof}

\begin{proof}[Proof of Proposition~\ref{prop:fs-risk-truncation}]
The Neumann series for $(I+\lambda EP)^{-1}$ gives
$\|\widehat\beta_\lambda-\widehat\beta^{[m]}_\lambda\|\le\lambda^{m+1}U_m$.
Expand the squared loss and apply Cauchy--Schwarz. Since $\|v_j^{[m]}\|\le\kappa^j Z$ and $U_m\le\kappa^{m+1}Z/(1-\kappa)$, the uniform remainder converges geometrically. Finally $p_m(0)=R(0)$, and $|\inf R-\min p_m|\le\bar\epsilon_m$. The exact remainder, polynomial coefficients and explicit geometric bound are derived below.
\end{proof}

In the notation of Proposition~\ref{prop:fs-risk-truncation}, the exact fit displacement is $\lambda L(I+\lambda EP)^{-1}Ez$. The finite geometric-series identity yields
\begin{equation}\widehat\beta_\lambda-\widehat\beta^{[m]}_\lambda
 =(-1)^m\lambda^{m+1}L(I+\lambda EP)^{-1}(EP)^mEz.\label{eq:app-exact-neumann-tail}\end{equation}
Because $\|EP\|\le\gamma<1$, the inverse norm is at most $(1-\lambda\gamma)^{-1}\le(1-\gamma)^{-1}$. Also $\|L\|=1/s_{\min}(G)$. This proves the tail-vector bound above without differentiating the clipping operation.

Set $v_j=v_j^{[m]}$ for $j\le m$ and $v_j=0$ otherwise. Retaining the complete squared loss gives
\begin{equation*}p_m(\lambda)=R(0)+\sum_{k=1}^{2m}c_k\lambda^k,\qquad
 c_k=2\mathbb E\langle e_0,v_k\rangle+
       \sum_{\substack{i+j=k\\1\le i,j\le m}}\mathbb E\langle v_i,v_j\rangle.\end{equation*}
The ordered sum includes both cross terms when $i\ne j$. In particular an order-two estimator truncation generally has a quartic risk, not merely a quadratic risk. The coefficients retain their signed expectations and all pilot/weight dependence.

Let $u_m(\lambda)$ denote \eqref{eq:app-exact-neumann-tail}. Then
\begin{equation*}R(\lambda)-p_m(\lambda)
 =2\mathbb E\langle\widehat\beta^{[m]}_\lambda-\beta,u_m(\lambda)\rangle
  +\mathbb E\|u_m(\lambda)\|^2.\end{equation*}
Apply the tail-vector bound and the triangle inequality to get \eqref{eq:fs-risk-truncation-budget}. Since $\|v_j\|\le\kappa^jZ$, the explicit uniform bound is
\begin{equation*}\bar\epsilon_m\le
 \frac{2\kappa^{m+1}}{1-\kappa}
 \mathbb E\left[\left(\|e_0\|+\frac{\kappa}{1-\kappa}Z\right)Z\right]
 +\frac{\kappa^{2m+2}}{(1-\kappa)^2}\mathbb E Z^2.\end{equation*}
Cauchy--Schwarz proves its finiteness under the displayed moment conditions. For fixed Gaussian regression and positive ridge, both $\|z\|$ and $1/s_{\min}(G)$ have polynomial growth in the response, by the bounds in Supplementary Section~\ref{app:fit-regularity}; all required moments therefore exist. No uniform-in-$n$ inverse-moment bound is inferred from this fact.

Uniform error gives $|\inf R-\min p_m|\le\bar\epsilon_m$, which proves both sides of \eqref{eq:fs-fixed-path-regret-interval}. If $\lambda_m$ minimizes the population polynomial, then $R(\lambda_m)-\inf R\le2\bar\epsilon_m$ as well. These are statements about deterministic strengths and exact population expectations.

The remaining unknown target in the envelope can be separated from observable integrands. Writing $R_{\rm OLS}=\sigma^2\operatorname{tr}(CC^\top)$,
\begin{equation}\bar\epsilon_m\le
 2\mathbb E\left[\left(\|\widehat\beta_0-t\|+\sum_{j=1}^m\|v_j\|\right)U_m\right]
 +2\sqrt{R_{\rm OLS}\mathbb E U_m^2}+\mathbb E U_m^2.\label{eq:app-truth-free-envelope}\end{equation}
This follows from $\|e_0\|\le\|\widehat\beta_0-t\|+\|t-\beta\|$. For known variance all integrands are data-computable, but their population expectations are not observed. Under the Gaussian weak-derivative conditions, a signed alternative is
\begin{equation*}R(\lambda)-p_m(\lambda)=
 2\mathbb E\{\langle\widehat\beta^{[m]}_\lambda-t,u_m(\lambda)\rangle
       +\sigma^2\langle C,D_Yu_m(\lambda)\rangle_F\}
 +\mathbb E\|u_m(\lambda)\|^2.\end{equation*}
Stein's identity applies to the complete map $u_m$, including the covariance and protection. It can preserve cancellations lost in norm envelopes. Neither this identity nor \eqref{eq:app-truth-free-envelope} supplies a confidence bound from a single realization.

\subsection{Grid and Monte Carlo uncertainty}

Let $\lambda_*$ minimize the continuous risk. If it is an endpoint, it is already on the grid. Otherwise $R'(\lambda_*)=0$. A nearest grid point is within $a/2$, and Taylor's integral remainder gives
$R(\lambda_j)-R(\lambda_*)\le\mathcal Ma^2/8$. Consequently the maximum grid gain is no larger than $G_h$ and no smaller than $G_h-\mathcal Ma^2/8$. On the stated simultaneous error event, replacing true grid gains by their upper or lower estimates proves \eqref{eq:fs-grid-regret-bound}. No convexity or unique optimizer is required.

The deterministic event implication is separate from its probability. In the reported experiment, simultaneous allowances use Gaussian Monte Carlo approximations, including an allowance for the estimated curvature mean. They are not finite-sample certified concentration bounds. The independently fitted reference is selected using simulated truth on training responses and evaluated on fresh responses. It approximates the best fixed strength, not the best response-dependent selector, and not the best target matrix. Designwise normalization by conditional OLS risk precedes averaging; ratios of pooled unnormalized risks are different quantities.

\section{Confidence-interval theory}
\label{supp:interval}
\label{sec:hc3-parb}
Wald intervals provide standard first-order inference, but their finite-sample coverage need not reflect the accuracy of the CC-PPI point estimator. We construct marginal intervals by centered joint-moment bootstrap with complete refitting and partial expansion toward a prediction anchor. The target is a population linear projection; a linear conditional response mean is not required. Throughout, $\widehat\beta$ denotes the protected half-strength CC-PPI estimator, whose affine refitting map is specified in Appendix~\ref{app:point-map}.

\subsection{Full-refit bootstrap and partial anchor expansion}

Let $P_n$ and $P_N^U$ denote the labeled and unlabeled empirical averages. Write $A_n=P_nXX'$, $p=A_n^{-1}P_nXY$, $a=(P_N^UUU')^{-1}P_N^UUf(U)$, $e_i=Y_i-X_i'p$, and
\[\widehat v=\frac{\sum_i e_i^2}{n-d},\quad
 \ell_i=X_i'(X'X)^{-1}X_i,\quad
 s_{n,j}^2=\frac{\widehat v}{n}(A_n^{-1})_{jj}.\]
Here $s_{n,j}$ is a convenient OLS normalization, not a claim that model-based OLS standard errors are robust under general responses. Put $Q=(p-a)'X'X(p-a)$ and
\begin{equation*}m=a+\sqrt{(1-d\widehat v/Q)_+}(p-a),\qquad
 \widetilde Y_i=X_i'p+e_i/(1-\ell_i),\end{equation*}
with $m=a$ when $Q=0$. The residual factor $1/(1-\ell_i)$ follows the HC3/jackknife leverage correction \cite{mackinnonwhite1985}; the radial center $m$ is our construction choice. Near-unit leverage is an explicit computational failure, not silently clipped.

For $g_n(t)=[P_nX(X't-\widetilde Y);(P_N^U-P_n)X(X't-f)]$, draw independently normalized exponential row weights in the two groups. In each draw, recompute the offset-corrected pilot, covariance blocks and protected weight $W_w$ as in Appendix~\ref{app:point-map}, with offset $o=-g_n(m)$, numerical ridge $n^{-2}I$ and half strength. Hold this fitted weight fixed when minimizing
\begin{equation}\widehat\beta^*=\mathop{\rm arg\,min}_{t\in\mathbb R^d}
 \{g_w(t)-g_n(m)\}'W_w\{g_w(t)-g_n(m)\}.\label{eq:hc3-bootstrap-objective}\end{equation}
Equivalently, $G_w'W_w\{g_w(\widehat\beta^*)-g_n(m)\}=0$, where $G_w=\partial g_w(t)/\partial t$. The $2d$ individual moment components need not vanish simultaneously. The OLS scale $s_{n,j}^*$ is computed separately from the weighted \emph{original} responses with residual divisor $n-d$. Let $q_l,q_u$ be the conditional $\alpha/2$ and $1-\alpha/2$ quantiles of $(\widehat\beta_j^*-m_j)/s_{n,j}^*$. Define
\begin{equation*}J_j=[L_j,U_j]=[\widehat\beta_j-q_us_{n,j},\widehat\beta_j-q_ls_{n,j}],\qquad
 I_j=[L_j-\gamma(L_j-a_j)_+,\ U_j+\gamma(a_j-U_j)_+],\end{equation*}
We use $\gamma=3/4$ and $B=399$. Expansion moves an endpoint only toward an anchor outside the parent interval and leaves the point estimate unchanged. The fraction was selected on separate discovery data before validation. Earlier Gaussian-parent validation is reported separately in Appendix~\ref{app:completed-results}; it does not establish optimality of the fraction. The normalized exponential weights have the Bayesian-bootstrap Dirichlet representation \cite{rubin1981}, used here for frequentist exchangeably weighted inference \cite{praestgaard1993}.

\subsection{Marginal validity and scope}

Theorem~\ref{thm:hc3-iid} is a fixed-distribution, first-order result for marginal intervals. Its proof is in Appendix~\ref{app:hc3-proof}. At the implemented $B=399$, the ideal-quantile statement is replaced by a limiting parent-coverage bracket of 94.5--95.0\%; expansion inherits its lower bound. This qualification concerns Monte Carlo quantiles, not a finite-sample coverage guarantee. The theorem does not cover singular auxiliary limits, growing dimension or conditional fixed-design inference.

Theorem~\ref{thm:finite-frame} gives conservative ideal-quantile coverage under complementary finite-frame sampling and an explicit score CLT. Theorems~\ref{thm:local-anchor} and~\ref{thm:anchor-joint-law} characterize local coverage and width, including dependence between the estimate and anchor. Theorem~\ref{thm:quantile-transfer} separates sampling, bootstrap, Monte Carlo and computational errors, conditional on bounds for the respective approximation terms. These supplementary results also clarify a cost of fixed expansion: when $a_{0,j}\ne\beta_j$, the width converges to $\gamma|a_{0,j}-\beta_j|$. Coverage gains must therefore be assessed together with length.
\subsection{Point estimator and refitting map}
\label{app:point-map}
Rows $X_i,U_k\in\mathbb R^d$ include any intercept column. The supplied scalar predictions are $f_i=f(X_i)$ and $f_k^U=f(U_k)$. All inverses below require full rank. Let $P_L,P_U$ be empirical means, optionally with separately normalized row weights, and set $\rho_n=n/N$. For a two-block offset $o=(o_1',o_2')'$, define
\[\begin{gathered}
 A_L=P_LXX',\quad A_U=P_UUU',\quad B=A_U-A_L,\\
 b_1=P_LXY-o_1,\quad b_2=P_UUf(U)-P_LXf(X)-o_2.
 \end{gathered}\]
The covariance pilot is $p_o=A_L^{-1}b_1$. Center the score rows $s=X(X'p_o-Y)$, $h=X(X'p_o-f)$ and $h_U=U(U'p_o-f(U))$ within their own samples. Covariances below use normalized empirical averages, not residual degrees-of-freedom divisors. Write
\[V=\operatorname{Cov}_L(s)+n^{-2}I,\quad M=-\operatorname{Cov}_L(s,h),\quad
 D=\operatorname{Cov}_L(h)+\rho_n\operatorname{Cov}_U(h_U)+n^{-2}I.\]
Put $V_c=V+M+M'+D$, $R_c=M+D$, $\omega=D-R_c'V_c^{-1}R_c$, and
\[Z=\omega^{-1/2}\{\operatorname{Cov}_U(h_U)-\operatorname{Cov}_L(h)\}\omega^{-1/2},\qquad
 W_r=\omega^{-1/2}[I+\tfrac12\operatorname{clip}(Z,-\tfrac12,\tfrac12)]^{-1}\omega^{-1/2}.\]
Clipping acts on eigenvalues. Let $H_r=B-R_c'V_c^{-1}A_U$, $b_c=b_1+b_2$ and $b_r=b_2-R_c'V_c^{-1}b_c$. The original protected half-strength CC point is
\[\widehat\beta(o)=
 (A_U'V_c^{-1}A_U+H_r'W_rH_r)^{-1}
 (A_U'V_c^{-1}b_c+H_r'W_rb_r).\]
The corresponding weight in the original two-moment coordinates is
\begin{equation*}\mathcal R=\begin{pmatrix}I&I\\-R_c'V_c^{-1}&I-R_c'V_c^{-1}\end{pmatrix},
 \qquad W=\mathcal R'\operatorname{diag}(V_c^{-1},W_r)\mathcal R.\end{equation*}
This weight is computed at the corrected pilot and held fixed in the quadratic minimization. Applied to the bootstrap-weighted arrays, it is $W_w$ in~\eqref{eq:hc3-bootstrap-objective}.
The observed fit has unit weights and $o=0$. Each centered bootstrap fit uses the adjusted responses and the fixed observed offset $o=b_n-G_nm=-g_n^{\widetilde Y}(m)$ in the sign convention $g_n(t)=G_nt-b_n$ of Section~\ref{sec:hc3-parb}. Thus its corrected right-hand side is $b_w-o=b_w+g_n^{\widetilde Y}(m)$, giving $g_w(t)-g_n(m)$. The executable routine calls this subtracted quantity \texttt{offset}. Covariance centering removes constant score offsets; the pilot must still be recomputed from the corrected first right-hand side. The raw-response OLS normalization is defined in Section~\ref{sec:hc3-parb}.
\subsection{A finite-sample limitation of Wald inference}
\label{sec:wald-failure}
Under the regular iid assumptions and consistent estimation of the actual joint moment covariance, the influence expansion and Slutsky's theorem give a standard-normal limit for the sandwich-studentized estimator. Thus the usual Wald interval is asymptotically valid in that regime. Estimating the weight on the inference sample is allowed. This statement does not provide a finite-sample approximation rate or conditional fixed-design validity.

The following example concerns the scalar-response linear-regression Wald construction, with $d$ regression coefficients including the intercept.

\begin{theorem}[Finite-sample regression Wald undercoverage]
\label{thm:balanced-wald}
Let $X=[\mathbf 1,Z]\in\mathbb R^{n\times d}$ have full column rank, $\mathbf 1'Z=0$, and $n>d$. Conditional on this design, let $Y=X\beta+\varepsilon$, with $\varepsilon\sim N(0,\sigma^2 I_n)$, $\sigma^2>0$, and $0<\alpha<1$. For an integer $m\ge1$, let $U$ consist of $m$ copies of the rows of $X$, so $N=mn$. Use a fixed predictor with matching repeated predictions in $U$. Use empirical covariance divisors $n,N$, joint numerical ridge $r=n^{-2}$ in the frozen weight $W$, and either zero or protected half-strength update. Let $\widehat S_{\rm raw}$ denote the unregularized $n$-scaled empirical joint moment covariance evaluated at the fitted point, and set
\begin{equation*}
 \widehat J=\widehat J(\widehat\beta),\qquad
 \widehat K=(\widehat J'W\widehat J)^{-1}\widehat J'W,\qquad
 \widehat V_{\rm W}=\frac1n\widehat K\widehat S_{\rm raw}\widehat K'.
\end{equation*}
The Wald interval is $\mathrm{CI}_{\rm W,j}=\widehat\beta_j\pm z_{1-\alpha/2}\sqrt{(\widehat V_{\rm W})_{jj}}$; no ridge is added to its middle covariance matrix. Then both observed points equal labeled OLS and, for the intercept,
\begin{equation}P\{\beta_1\in\mathrm{CI}_{\mathrm W,1}\mid X,U\}
 \le 2F_{t_{n-d}}\!\left(z_{1-\alpha/2}\sqrt{\frac{n-d}{n}}\right)-1
 <1-\alpha.\label{eq:wald-failure-bound}\end{equation}
For a fixed labeled realization, increasing $m$ cannot increase these Wald widths. The result holds for every $N=mn$; more unlabeled repetitions do not remove the failure.
\end{theorem}

A proof is given in Appendix~\ref{app:wald-proof}.

The example uses the implemented ridge and variance map. Its bound worsens as the residual degrees of freedom $n-d$ decrease and is unchanged by increasing the repeated unlabeled sample. It describes the specified conditional design, not a classification of arbitrary designs by $n,N,d$, and is compatible with regular iid asymptotic validity.\subsection{Proof of the regression Wald failure result}
\label{app:wald-proof}
\begin{proof}
Put $A=X'X/n$, $p=A^{-1}X'Y/n$, $s_i=X_i(X_i'p-Y_i)$, $h_i=X_i(X_i'p-f_i)$, and use centered empirical covariances
\[V=\operatorname{Cov}_L(s),\quad C=\operatorname{Cov}_L(s,h),\quad
 H=\operatorname{Cov}_L(h),\quad D=(1+n/N)H.\]
The unlabeled and labeled empirical laws coincide, so their auxiliary moment and derivative are identically zero for every parameter. Their covariance difference is also zero, so the protected update vanishes. The unique observed minimizer is $p$. Write
\[S=\widehat S_{\rm raw}=\begin{pmatrix}V&-C\\-C'&D\end{pmatrix},\quad
 W=(S+rI)^{-1},\quad J=\begin{pmatrix}A\\0\end{pmatrix},\quad
 B=C(D+rI)^{-1}.\]
Block inversion gives $K=(J'WJ)^{-1}J'W=A^{-1}[I,B]$. Consequently the actual unregularized-middle sandwich obeys
\begin{align}
 \widehat V_{\mathrm W}
 &=\frac1n A^{-1}\{V-CB'-BC'+BDB'\}A^{-1},\nonumber\\
 \widehat V_{\mathrm{HC0}}-\widehat V_{\mathrm W}
 &=\frac1n A^{-1}C\{(D+rI)^{-1}+r(D+rI)^{-2}\}C'A^{-1}\succeq0.
 \label{eq:wald-variance-loss}
\end{align}
Both intervals have the same center, so the Wald intercept interval is contained in the HC0 normal interval. Centering the non-intercept columns gives $\widehat V_{\mathrm{HC0},11}=\mathrm{RSS}/n^2$. Also $p_1-\beta_1\sim N(0,\sigma^2/n)$ independently of $\mathrm{RSS}/\sigma^2\sim\chi^2_{n-d}$. This yields the first inequality in \eqref{eq:wald-failure-bound}. For positive arguments, the Student CDF is smaller than the standard-normal CDF, and the argument is smaller than $z_{1-\alpha/2}$, giving strict undercoverage. Finally $D=(1+1/m)H$ decreases as $m$ increases. In an eigenbasis of $H$, each eigenvalue of $(D+rI)^{-1}+r(D+rI)^{-2}$ increases, so \eqref{eq:wald-variance-loss} increases and widths decrease.
\end{proof}
\par\medskip\noindent{\bfseries Limit notation for the proof.}\par\nobreak\smallskip

For the score covariance $\Omega$ in Theorem~\ref{thm:hc3-iid}, put $G=[A;0]$,
$K=(G'\Omega^{-1}G)^{-1}G'\Omega^{-1}$, $\Sigma=K\Omega K'$ and $c_j^2=v(A^{-1})_{jj}$.
The proof establishes the sampling limit $N(0,\Sigma)$ and the conditional bootstrap-root limit $t\mapsto\Phi(tc_j/\sqrt{\Sigma_{jj}})$ before inverting its quantiles.
\subsection{Validity of CC-PPI confidence intervals}\label{app:hc3-proof}
This appendix proves Theorem~\ref{thm:hc3-iid} for the fully refitted algorithm in Section~\ref{sec:hc3-parb}. All probabilistic statements refer to the regular iid assumptions of that theorem unless stated otherwise.

\subsubsection{Pilot, leverage, and conditional empirical moments}
The law of large numbers gives $A_n\to A$, $p\to\beta$, $\widehat v\to v$, and $a\to a_0=A^{-1}E[Xf(X)]$. Moreover $\max_i\ell_i=o_p(1)$: $\max_i\|X_i\|^2/n=o_p(1)$ by integrability of $\|X\|^2$, while $\|A_n^{-1}\|=O_p(1)$. If $\Delta_i=\widetilde Y_i-Y_i$, then
\begin{equation}P_n\|X_i\Delta_i\|^2\le
 \left(\frac{\max_i\ell_i}{1-\max_i\ell_i}\right)^2 P_n\|X_ie_i\|^2=o_p(1).\label{eq:hc3-small}\end{equation}
In fact $|\widetilde Y_i|\le |Y_i|+\|X_i\|\|p\|+|e_i|/(1-\max\ell_i)$ on the low-leverage event. Together with the assumed moments, this bounds the empirical moments of all score products used below. Their dependence on the fitted pilot is polynomial; convergence of that pilot and Cauchy--Schwarz justify replacing it by $\beta$ in second moments.

The center satisfies the deterministic inequality
\begin{equation*}(m-p)'A_n(m-p)\le d\widehat v/n.\end{equation*}
For $Q\le d\widehat v$ it is immediate; otherwise use $(1-\sqrt{1-t})^2\le t$ for $t=d\widehat v/Q\in[0,1]$. Consequently $m-p=O_p(n^{-1/2})$ and $m\to\beta$, without differentiating the radial positive-part rule.

For normalized weights $w_i=E_i/\bar E$ with iid $E_i\sim\operatorname{Exp}(1)$, exchangeability gives $E_*w_i=1$ and
\[\operatorname{Var}_*(P_wz)=\frac{P_n[(z-P_nz)(z-P_nz)']}{n+1}.\]
This formula applies conditionally to data-dependent arrays. Empirical second-moment bounds therefore yield conditional LLNs for all degree-four data products in the fitted covariances (the $8+\delta$ assumption is sufficient), as well as for $A_w,A_w^U$ and the raw-response OLS scale. Conditional invertibility follows by convergence to positive definite limits.

Write $r_i=X_i\widetilde Y_i-X_iX_i'm$. The adjusted supervised pilot uses the offset $P_nX\widetilde Y-A_nm$, hence exactly
\[p^*-m=A_w^{-1}(P_w-P_n)r.\]
The variance bound above gives $p^*-m=O_{P_*}(n^{-1/2})$ in probability. In contrast, the \emph{raw} weighted original-response pilot converges to $\beta$ and supplies $\sqrt n s_{n,j}^*\to c_j$ conditionally in probability. Also $\sqrt n s_{n,j}\to c_j$. Mixing these two pilots would define a different algorithm.

\subsubsection{Consistency of the actual protected map}
In the code's transformed coordinates let $M_n=-\widehat C$, $D_n=\widehat H_L+\rho_n\widehat H_U+n^{-2}I$, $V_n=\widehat V+n^{-2}I$, $V_{c,n}=V_n+M_n+M_n'+D_n$, $R_{c,n}=M_n+D_n$, and
\[\omega_n=D_n-R_{c,n}'V_{c,n}^{-1}R_{c,n}.\]
The invertible transformation $(z_1,z_2)\mapsto(z_1+z_2,z_2)$ sends $\Omega$ to the limiting covariance of these blocks. Its Schur complement is positive definite. By the preceding LLNs, all inverses and square roots converge. The guard
\[Z_n=\omega_n^{-1/2}(\widehat H_U-\widehat H_L)\omega_n^{-1/2}\xrightarrow{p}0.\]
The implemented residual precision is $\omega_n^{-1/2}[I+\tfrac12\operatorname{clip}(Z_n,-\tfrac12,\tfrac12)]^{-1}\omega_n^{-1/2}$. It converges to the inverse limiting Schur covariance. Spectral matrix functions are continuous here, including at repeated eigenvalues; derivatives of eigenvectors are unnecessary. Undoing the coordinate transformation shows that the fitted affine solution map $K_n$ converges to $K$. The conditional LLNs and corrected-pilot identity give the same result for $K_n^*$ in conditional probability. This argument accounts for covariance and pilot refitting, rather than freezing their finite-sample values.

\subsubsection{Affine identities and the multiplier limit}
For the original responses write $g_n^Y(t)=[P_nX(X't-Y);(P_N^U-P_n)X(X't-f)]$. Direct affine minimization gives exactly
\begin{equation}\widehat\beta-\beta=-K_n g_n^Y(\beta),\qquad
 \widehat\beta^*-m=-K_n^*\{g_w^{\widetilde Y}(m)-g_n^{\widetilde Y}(m)\}.\label{eq:hc3-affine}\end{equation}
The dependence of $K_n$ on the same sample does not invalidate these identities. The first moment vector is centered and its covariance after multiplication by $n$ tends to $\Omega$; the independent two-sample CLT and Slutsky establish the stated sampling limit.

The conditional multiplier argument below is a finite-dimensional instance of the exchangeably weighted bootstrap framework \cite{praestgaard1993}; we verify its required moment and refitting steps for the present estimator. For an empirical vector $z$, the exact multiplier representation is
\[\sqrt n(P_w-P_n)z=\frac{n^{-1/2}\sum_i(E_i-1)(z_i-P_nz)}{\bar E}.\]
Apply it independently to $z_L=(s,-h)$ and $z_U=(0,h)$. Their centered empirical covariances converge; the moment assumption implies negligible maximal squared summands relative to the sample size and conditional Lindeberg. The resulting conditional covariance is $\Omega$, with the factor $n/N$ in the unlabeled block. Replacing $\beta$ by $m$ changes scores by products of $XX'$ with $m-\beta=o_p(1)$, whose empirical squared norm is $o_p(1)$. Equation~\eqref{eq:hc3-small} handles the response adjustment. The conditional variance formula makes both replacements negligible in the multiplier sum. Conditional Slutsky in~\eqref{eq:hc3-affine}, followed by scale convergence, proves the asserted conditional normal limit. Continuity of its distribution function upgrades weak convergence to uniform CDF convergence.

Conditional quantile consistency now gives $q_l\to F_j^{-1}(\alpha/2)$ and $q_u\to F_j^{-1}(1-\alpha/2)$, where $F_j(t)=\Phi(tc_j/\sqrt{\Sigma_{jj}})$. Since $\beta_j\in J_j$ is exactly $q_l\le(\widehat\beta_j-\beta_j)/s_{n,j}\le q_u$, convergence to constant quantiles proves nominal marginal coverage. Finally $J_j\subseteq I_j$ for every defined sample, proving the lower bound without any independence assumption between quantiles and observed roots.

Undefined inverse, leverage or Schur events have probability tending to zero under the stated conditions. For formal convergence, define the statistic arbitrarily and measurably on these events; this is a mathematical extension, not a deployed fallback or permission to omit trials. For fixed $B$, a finite union bound handles all draws. For $B_n\to\infty$, quantile conclusions apply to an extended conditional law; equivalence to a batch that aborts on any failure additionally requires $B_nP_*(\text{draw failure})\to0$ in probability, or another explicit failure policy. We do not assert that rate from mere consistency.

\subsubsection{Monte Carlo quantiles and expansion cost}
For $B=399$ at 95\%, default linear quantiles interpolate ranks 10--11 and 389--390. Conditional convergence to deterministic $F_j$ implies joint convergence of the observed root and these finitely many draws to independent draws from $F_j$. Thus the base limiting coverage is bracketed by $(389-11)/400=0.945$ and $(390-10)/400=0.950$. Expansion retains the lower bound. These are asymptotic Monte Carlo bounds, not finite-sample coverage bounds or a guarantee of exactly 95\% for an interpolated quantile.

For ideal quantiles and $a_{0,j}\ne\beta_j$, $L_j,U_j\to\beta_j$ implies
\[|I_j|\xrightarrow{p}\gamma|a_{0,j}-\beta_j|.\]
Precisely one expansion becomes a positive constant. Under the centered normal limit the corresponding endpoint constraint becomes automatic, whereas the other retains its one-sided nominal limit. Hence $P\{\beta_j\in I_j\}\to1-\alpha/2$. With finite $B$, the analogous one-sided Monte Carlo bracket at 95\% is $0.9725$--$0.9750$. If $a_{0,j}=\beta_j$, both $a_j-\beta_j$ and the base endpoints are $O_p(n^{-1/2})$, so the expanded width is $O_p(n^{-1/2})$, though its exact limiting coverage need not equal nominal. This distinguishes a genuine width result from containment alone.

\subsection{Sampling-law boundaries}
The distinction between random-design and fixed-design regression resampling is discussed by Freedman \cite{freedman1981} and Mammen \cite{mammen1993}; growing-dimension limitations are further illustrated by El Karoui and Purdom \cite{elkaroui2018}. The calculations below specify the boundaries for our joint-moment construction rather than importing validity from those references.
For fixed designs and $f(x)=x'a_0$, the auxiliary moment at truth is $(A_U-A_L)(\beta-a_0)$. It is deterministic conditionally on designs, whereas row-weight recentering removes it in the bootstrap world. Matching this conditional bias and covariance requires additional assumptions; the iid theorem is not a fixed-design theorem.

For a finite frame of size $M=n+N$, sample $n$ labels uniformly without replacement and take the entire complement as unlabeled. Let $\beta_M$ be the full-frame OLS projection, $s_i=x_i(x_i'\beta_M-y_i)$ and $h_i=x_i(x_i'\beta_M-f_i)$. Let $V_F,C_F,H_F$ be frame covariances with divisor $M-1$ and put $p_F=n/M$. Exact sampling identities give
\[\operatorname{Var}\!\left(\sqrt n\begin{bmatrix}\bar s_L\\\bar h_U-\bar h_L\end{bmatrix}\right)
 =\Omega_F=\begin{pmatrix}(1-p_F)V_F&-C_F\\-C_F'&(1+\rho_n)H_F\end{pmatrix}.\]
Indeed $\bar s_F=0$, $\bar h_U-\bar h_L=-(1+\rho_n)(\bar h_L-\bar h_F)$, and $\operatorname{Var}(\bar z_L)=(1-p_F)S_z/n$. The independent-group covariance differs by $p_F\operatorname{diag}(V_F,0)$, not a scalar factor on the whole matrix. Under a finite-population CLT and covariance LLNs, positive definite limiting independent-group covariance, negligible maximum leverage and nondegenerate coordinate variances, the limiting CC map $K_I=[K_1,K_2]$ has $K_1=A^{-1}$ and the bootstrap-minus-actual variance gap is $p_FA^{-1}V_FA^{-1}$. Ideal base intervals are then asymptotically conservative. This conditional extension requires those finite-frame regularity conditions; the exact score identity alone does not certify them for a particular dataset, nor justify contracting observed intervals.
\subsection{Coverage under complementary sampling and local anchor error}
\label{sec:theory-closeout}
This section studies complementary finite-frame sampling, local anchor error and finite bootstrap draws for the construction in Section~\ref{sec:hc3-parb}. Each result uses its stated sampling and regularity conditions.

\subsection{Complementary finite-frame sampling}
\begin{theorem}[Finite-frame marginal coverage]\label{thm:finite-frame}
Consider a sequence of deterministic frames of size $M=n+N\to\infty$, with fixed $d$ and fixed prediction values. Select the labeled set uniformly among subsets of size $n$, using its entire complement as unlabeled. Suppose $p_M=n/M\to p\in(0,1)$, $A_M=M^{-1}\sum_i x_ix_i'\to A\succ0$, $\beta_M=A_M^{-1}M^{-1}\sum_i x_iy_i$ is bounded, and
\[\sup_M M^{-1}\sum_i(\|x_i\|+|y_i|+|f_i|)^{8+\delta}<\infty,
 \qquad v_M=M^{-1}\sum_i(y_i-x_i'\beta_M)^2\to v>0.\]
Let $V_M,C_M,H_M$ be the frame score covariances defined in Appendix~\ref{app:hc3-proof}, and suppose they converge to $V,C,H$. Assume the finite-population score CLT
\begin{equation}\sqrt n\begin{bmatrix}\bar s_L\\\bar h_U-\bar h_L\end{bmatrix}
 \Rightarrow N(0,\Omega_F),\qquad
 \Omega_F=\begin{pmatrix}(1-p)V&-C\\-C'&H/(1-p)\end{pmatrix}.\label{eq:frame-clt}\end{equation}
The CLT is an explicit sampling assumption; frame moment bounds alone are not used as a substitute for it. Define
\[\begin{gathered}
 \Omega_I=\begin{pmatrix}V&-C\\-C'&H/(1-p)\end{pmatrix}\succ0,
 \qquad G=\begin{bmatrix}A\\0\end{bmatrix},\\
 K_I=(G'\Omega_I^{-1}G)^{-1}G'\Omega_I^{-1},\\
 \Sigma_I=K_I\Omega_IK_I',\qquad \Sigma_F=K_I\Omega_FK_I'.
 \end{gathered}\]
For coordinate $j$, assume $(\Sigma_F)_{jj}>0$. Then the original fully refitted interval construction satisfies
\[\sqrt n(\widehat\beta-\beta_M)\Rightarrow N(0,\Sigma_F),\qquad
 \Sigma_I-\Sigma_F=pA^{-1}VA^{-1}\succeq0.\]
With ideal conditional quantiles and $z=z_{1-\alpha/2}$, its base interval has limiting coverage
\begin{equation}P\{\beta_{M,j}\in J_j\}\longrightarrow
 2\Phi\!\left(z\sqrt{(\Sigma_I)_{jj}/(\Sigma_F)_{jj}}\right)-1
 \ge1-\alpha.\label{eq:frame-coverage}\end{equation}
The expanded interval inherits the lower bound. These are marginal design-based statements for growing frames, not guarantees for a single fixed dataset or a logistic target.
\end{theorem}

\begin{proof}
For any deterministic frame array $z_i$, simple random sampling gives
\[E\bar z_L=\bar z_M,\qquad
 \operatorname{Var}(\bar z_L)=(1-p_M)S_{z,M}/n.\]
The same identity for each required data product, together with the moment bound, gives convergence in probability of labeled and complementary empirical moments up to degree four to their frame values. The bound also gives $\max_i\|x_i\|^2/n=o(1)$; hence the labeled maximum leverage is $o_p(1)$. The supervised pilot differs from $\beta_M$ by $A_L^{-1}\bar s_L$ up to sign, so is consistent by~\eqref{eq:frame-clt}. The radial bound gives $m-\beta_M=O_p(n^{-1/2})$. HC3 response adjustments are negligible in empirical score second moments as in~\eqref{eq:hc3-small}. All fitted block covariances, including those in bootstrap draws, therefore obey the LLNs used in the iid proof.

In particular, labeled and complementary prediction-score covariances have the same limit $H$, so their difference and the protected update vanish. Thus $K_n\to K_I$, and the affine identity and~\eqref{eq:frame-clt} give the asserted sampling limit. Conditional on the selected rows, the algorithm still uses independent row multipliers in the two groups. The multiplier argument in Appendix~\ref{app:hc3-proof} therefore has covariance $\Omega_I$, rather than $\Omega_F$; it remains valid for these deterministic triangular arrays because the moment bounds imply the conditional Lindeberg condition. The raw OLS normalization has limit $c_j^2=v(A^{-1})_{jj}>0$. Ideal quantile endpoints are consequently $\widehat\beta_j\pm z\sqrt{(\Sigma_I)_{jj}/n}+o_p(n^{-1/2})$.

Finally $K_IG=I$ implies that the first block of $K_I$ equals $A^{-1}$. Since $\Omega_I-\Omega_F=p\operatorname{diag}(V,0)$, the variance-gap identity follows. The normal sampling limit proves~\eqref{eq:frame-coverage}; samplewise containment $J_j\subseteq I_j$ proves the final assertion. No finite-population correction is applied to the deployed intervals.
\end{proof}

\subsection{Local mismatch: exact limiting gain and width}
\begin{theorem}[Local anchor coverage and length]\label{thm:local-anchor}
Consider one coordinate in a sequence of experiments. Suppose
\[\bigl(\sqrt n(\widehat\beta_j-\beta_{n,j}),
       \sqrt n(a_j-\beta_{n,j})\bigr)\Rightarrow(Z,D),
 \qquad Z\sim N(0,\tau^2),\quad \tau>0,\]
and the ideal base endpoints satisfy
\[\sqrt n(L_j-\beta_{n,j})=\sqrt n(\widehat\beta_j-\beta_{n,j})-h+o_p(1),
 \quad
 \sqrt n(U_j-\beta_{n,j})=\sqrt n(\widehat\beta_j-\beta_{n,j})+h+o_p(1),\]
where $h>0$ is deterministic. Take fixed $0<\gamma<1$, and set $r=(1-\gamma)/\gamma$. Assume the two limiting expanded endpoints have no atom at zero. Then
\begin{align}
 \lim P\{\beta_{n,j}\in I_j\}
 &=P\{|Z|\le h\}
   +P\{Z>h,\ D\le-r(Z-h)\}\notag\\
 &\quad+P\{Z<-h,\ D\ge-r(Z+h)\},\label{eq:local-gain}\\
 \sqrt n|I_j|&\Rightarrow 2h+\gamma(|D-Z|-h)_+.
 \label{eq:local-width}
\end{align}
No independence of $Z$ and $D$ is required. In the regular iid theorem $h=z\tau$, so the first term in~\eqref{eq:local-gain} is $1-\alpha$. A strict local coverage increase occurs exactly when one of the two added probabilities is positive; this need not improve closeness to nominal coverage. If the scaled widths are uniformly integrable, their expectations also converge to the expectation in~\eqref{eq:local-width}.
\end{theorem}

\begin{proof}
The continuous mapping theorem gives limiting expanded endpoints
\[l_\gamma=Z-h-\gamma(Z-h-D)_+,\qquad
 u_\gamma=Z+h+\gamma(D-Z-h)_+.\]
Containment covers every $|Z|\le h$. If $Z>h$, the upper constraint is automatic and $l_\gamma\le0$ is equivalent to $D\le-r(Z-h)$. If $Z<-h$, the lower constraint is automatic and $u_\gamma\ge0$ is equivalent to $D\ge-r(Z+h)$. The three events are disjoint; the boundary assumption allows passage of coverage probabilities to the limit. Subtracting endpoints gives
\[u_\gamma-l_\gamma=2h+\gamma\{(Z-h-D)_++(D-Z-h)_+\}
 =2h+\gamma(|D-Z|-h)_+.\]
Uniform integrability supplies the optional expectation conclusion.
\end{proof}

The theorem applies to local predictor mismatch when
$\sqrt n(a_{0,n,j}-\beta_{n,j})\to\kappa$ and the centered anchor error has a joint limit $W$: then $D=\kappa+W$. Such a triangular sequence requires its own joint CLT and base-quantile consistency; the fixed-distribution iid theorem alone does not imply uniform local validity. If $D=0$, the two gain events have probability zero, yet the width increment is $\gamma(|Z|-h)_+$. Thus even a perfectly centered limiting anchor does not by itself imply strict improvement: partial expansion can lengthen a missed interval without reaching the truth. For fixed nonzero mismatch, Appendix~\ref{app:hc3-proof} instead gives nonshrinking width and one-sided limiting coverage. The fixed $\gamma=3/4$ is a frozen construction choice, not a proved optimum.

\begin{theorem}[Joint law at an accurate projection anchor]\label{thm:anchor-joint-law}
Under Theorem~\ref{thm:hc3-iid}, suppose additionally $a_0=\beta$. Write $H=\operatorname{Var}\{h(X;\beta)\}$ for the population auxiliary-score covariance and $K=[K_1,K_2]$. Then
\begin{equation}\sqrt n\begin{bmatrix}\widehat\beta-\beta\\a-\beta\end{bmatrix}
 \Rightarrow N\!\left(0,
 \begin{pmatrix}
 \Sigma&\rho K_2HA^{-1}\\
 \rho A^{-1}HK_2'&\rho A^{-1}HA^{-1}
 \end{pmatrix}\right).\label{eq:joint-anchor-covariance}\end{equation}
Whenever the selected coordinate pair is nondegenerate, this joint law satisfies the boundary requirement of Theorem~\ref{thm:local-anchor}, with $h=z\sqrt{\Sigma_{jj}}$. Its coverage and width formulas therefore apply without an additional unproved joint-limit assumption in this fixed-distribution accurate-anchor case.
\end{theorem}
\begin{proof}
Here $E h=0$ and $a-\beta=-A_U^{-1}\bar h_U$. The iid affine expansion gives
\[\sqrt n(\widehat\beta-\beta)
 =-K_1\sqrt n\bar s_L+K_2\sqrt n\bar h_L
   -K_2\sqrt n\bar h_U+o_p(1).\]
The two independent group CLTs give the stated joint normal limit. Only the last summand shares unlabeled randomness with the anchor, and $\operatorname{Var}(\sqrt n\bar h_U)\to\rho H$, yielding the positive cross block in~\eqref{eq:joint-anchor-covariance}. A nondegenerate bivariate normal law has a density, so the two affine pieces of each endpoint have no atom at zero.
\end{proof}

\subsection{Finite-draw coverage accounting and computational failure}
\begin{theorem}[Quantile-error transfer]\label{thm:quantile-transfer}
Let $T=(\widehat\beta_j-\beta_j)/s_{n,j}$ be a measurably extended observed root, and let $F$ be a deterministic continuous strictly increasing CDF with
$\sup_x|P(T\le x)-F(x)|\le\varepsilon_s$.
Conditionally on the observed data, let the extended bootstrap roots be iid with CDF $F_*$, and suppose
\[P\{\|F_*-F\|_\infty>\varepsilon_b\}\le\delta_b.\]
Use $B\ge2$ roots and default linear-interpolation empirical quantiles at $p_l=\alpha/2$, $p_u=1-\alpha/2$. For any $t>0$ put
\[\eta=\varepsilon_b+t+2/B<\alpha/2,
 \qquad \Delta=\delta_b+4\exp(-2Bt^2).\]
Let $\pi_*$ be the conditional probability that a single draw fails. If $\delta_o$ bounds outer-fit failure, set
\[\delta_f=\min\{1,\delta_o+E[\min(1,B\pi_*)]\}.\]
When an undefined interval counts as noncoverage, the implemented base coverage $C_J$ obeys
\begin{equation*}1-\alpha-2\eta-2\varepsilon_s-\Delta-\delta_f
 \ \le C_J\le\ 1-\alpha+2\eta+2\varepsilon_s+\Delta.\end{equation*}
Bounds may be truncated to $[0,1]$. The expanded coverage is at least the lower bound. The constants $\varepsilon_s,\varepsilon_b,\delta_b$ are assumptions or externally established bounds, not quantities certified by this theorem from the observed data.
\end{theorem}

\begin{proof}
On $\|F_*-F\|_\infty\le\varepsilon_b$, examine the four fixed thresholds $F^{-1}(p_k\pm\eta)$, $k=l,u$. At each threshold the empirical CDF is an average of $B$ Bernoulli variables. The exponential-Markov bound for a bounded Bernoulli sum gives deviation probability at most $\exp(-2Bt^2)$ in the required direction. A union bound gives probability at most $4\exp(-2Bt^2)$ that any threshold fails. The $2/B$ margin accommodates both adjacent order statistics of the linear interpolant (its position is $1+(B-1)p_k$). Therefore, except on an event of probability at most $\Delta$,
\[F^{-1}(p_k-\eta)\le q_k\le F^{-1}(p_k+\eta),\qquad k=l,u.\]
The fixed inner interval between $F^{-1}(p_l+\eta)$ and $F^{-1}(p_u-\eta)$ is then contained in $[q_l,q_u]$, which is contained in the corresponding fixed outer interval. The sampling-CDF error costs at most $2\varepsilon_s$ for either fixed interval. Subtracting or adding $\Delta$ proves the bounds for extended roots. The conditional union bound for any failed inner draw is $\min(1,B\pi_*)$; outer failure adds at most $\delta_o$. Declaring a failed interval to have zero inclusion can lower coverage by at most $\delta_f$ and cannot increase it. Finally, containment transfers the lower bound to expansion. Independence between $T$ and the conditional quantiles was not used.
\end{proof}

This is a finite-sample transfer inequality, not a primitive CC-PPI coverage-rate theorem. At $B=399$ its generic exponential bound can be vacuous; the sharper asymptotic rank bracket in Appendix~\ref{app:hc3-proof} remains the applicable fixed-$B$ result. For an abort-on-any-failure batch with $B_n\to\infty$, sufficient conditions are $\delta_o\to0$ and $E\min(1,B_n\pi_*)\to0$; the latter follows from $B_n\pi_*\to_p0$. Combined with $\varepsilon_s,\varepsilon_b,\delta_b\to0$ and a choice $t_n\to0$, $B_nt_n^2\to\infty$, the lower bound tends to $1-\alpha$. Consistency alone supplies no such failure rate or uniform near-singularity approximation bound.

\subsection{Scope of the supplementary results}
The finite-frame theorem connects the linear real-data sampling protocol to an explicit conservative limiting coverage formula, subject to a score CLT. The local theorem identifies exactly which joint anchor errors repair a base miss and the accompanying width cost. The transfer inequality separates sampling approximation, bootstrap approximation, Monte Carlo error and computational failure. None proves uniform small-sample superiority, an optimal expansion strength, validity for logistic or multi-response targets, or a data-driven rule for switching between Wald and bootstrap. The designed Wald failure example and fixed-design simulation evidence retain their separate scopes.

\subsection{Implementation and resampling qualifications}
\label{sec:interval-sources}
The exponential weights follow the Dirichlet representation of the Bayesian bootstrap \cite{rubin1981} within the frequentist exchangeably weighted framework \cite{praestgaard1993}. Validity for CC-PPI requires the centering, pilot and covariance refitting arguments established above. The HC3 leverage factor modifies the response tuples \cite{mackinnonwhite1985}; no HC3 covariance matrix is applied to the observed CC estimate.

Random-design and fixed-design bootstrap laws differ \cite{freedman1981,mammen1993}. Score-based perturbations, higher-order design/error interactions and the choice of studentization are also relevant \cite{kline2010,kline2011,lee2014}; no higher-order refinement from those procedures is asserted for the present construction. Growing-dimensional bootstrap validity requires separate analysis \cite{elkaroui2018}.

All row weights are normalized separately within the labeled and unlabeled groups. The ridge is $n^{-2}I$, the protected covariance update has half strength, the expansion fraction is $3/4$, and the reported experiments use 399 inner draws. A near-unit leverage value, undefined inverse, nonpositive Schur covariance or failed inner fit is recorded as failure. Failed draws are not silently discarded or regenerated. The increasing-batch failure qualification is given in Appendix~\ref{app:hc3-proof}.

For uncalibrated GMM, setting the protected update strength to zero gives the same regular iid limiting map $K$ under the stated nondegeneracy conditions; the center, raw-response OLS scale and quantile argument are unchanged. Multi-response and logistic numerical extensions are labeled separately. The logistic extension uses Bernoulli responses, a supervised logistic center and robust supervised scales, rather than OLS HC3 pseudo-responses. Neither extension is covered by the scalar-response linear-projection theorem.

Original PPI logistic score regions are distinguished from numerical coordinate projections, which are local optimization outputs without a global endpoint certificate. Source-author RePPI OLS intervals and manuscript-convention intervals retain their respective centers and fold conventions. The supplementary machine-readable records preserve these distinctions, failures and all available confidence levels.

\subsection{Exact samplewise coverage gain}
For any defined random interval $J=[L,U]$, anchor $a$, target $b$ and $0<\gamma<1$, let $r=(1-\gamma)/\gamma$ and use the expansion in \eqref{eq:new-ci}. Then, without a distributional assumption,
\begin{align*}
 P\{b\in I\}-P\{b\in J\}
 &=P\{L>b,\ a-b\le-r(L-b)\}\\
 &\quad+P\{U<b,\ a-b\ge r(b-U)\}\ge0.
\end{align*}
The length increment is exactly $\gamma\operatorname{dist}(a,J)$. Indeed, if $L>b$, the lower expanded endpoint is at most $b$ exactly when $a-b\le-r(L-b)$; if $U<b$, the upper endpoint reaches $b$ exactly when $a-b\ge r(b-U)$. The events are disjoint and exhaust the misses repaired by expansion. Subtracting the base inclusion indicator proves the identity, and subtracting endpoint lengths proves the width formula. Consequently every valid lower bound for parent coverage also bounds expanded coverage, with an explicit nonnegative gain.

\section{Simulation details and diagnostics}
\label{app:literature-protocol}
\label{supp:simulation}
\label{app:completed-results}

The simulations compare parameter MSE for linear and logistic regression with a separate focused covariance-calibration comparison. A multivariate surrogate model compares the population efficiencies of unrestricted, diagonal and scalar corrections.

\subsection{Setup and methods}

In the regression simulations, $X\sim N(0,\Sigma)$ with $\Sigma_{ij}=0.6^{|i-j|}$ and $d\in\{2,5,10,20\}$. We draw one labeled and one unlabeled design per cell and seed, with $N=5n$, then repeat responses on those fixed designs. There are two design seeds. Linear regression uses $Y=X\beta+\epsilon$, $\epsilon\sim N(0,1)$, $\|\beta\|=1$, and $f(x)=x^\top\alpha$ with $\alpha=\beta+r\mathbf1_d/\sqrt{\mathbf1_d^\top\Sigma\mathbf1_d}$. The grid is $n\in\{50,100,500,1000,2000,10000\}$ and $r\in\{.2,.5,1\}$, with 1,000 responses per design. Logistic regression uses $\Pr(Y=1\mid X)=\operatorname{expit}(X^\top\theta)$, where $\theta=0.6b/\|b\|$, $b_1=1$, $b_2=0.5$ and $b_j=0.4(-0.6)^{j-1}$ for $j\ge3$. Its predictor is $f(x)=\operatorname{expit}(x^\top\alpha)$ with $\alpha=\theta+r\mathbf1_d/\sqrt{\mathbf1_d^\top\Sigma\mathbf1_d}$. It uses $n\in\{100,500,2000\}$, $r\in\{.2,.5\}$ and 200 responses per design.

Comparators are supervised estimation, original PPI, PPI++, PSPA, PDC and RePPI. PSPA uses the full-Jacobian one-step update, with a lower-capped version retained as a sensitivity check. The PPI++ covariance plug-in is evaluated at the supervised pilot. Linear moments have closed-form solutions, and logistic fits use damped score solving and trust-region GMM least squares. Strength grids contain 101 linear points and 11 logistic points. Full numerical protocols, additional configurations and convergence diagnostics are in Supplementary Section~\ref{app:empirical-diagnostics}.

\subsection{Single-outcome linear regression}
\label{subsec:sim-linear}

Table~\ref{tab:main-linear-grid} reports the small-sample comparison at $r=.2$. At $d=20,n=50$ and $d=10,n=50$, CC-PPI has 53.3\% and 57.2\% of supervised MSE. At $d=20,n=100$ it has 52.8\%, and at $d=10,n=500$ it has 74.7\%. The correctly specified conditional mean gives zero population supervised--auxiliary cross covariance here, so the observed gains arise from finite-sample fitting.

The representative results are reported in Table~\ref{tab:main-linear-grid}.

\subsection{Multi-outcome surrogate model}
\label{subsec:sim-multi}

We now extend the predictor-of-$X$ setting to an observed surrogate available in both inference samples. Let $p$ be the covariate dimension and $q$ the number of outcomes, so the parameter dimension is $d=pq$. Take $Y=X\Theta+\epsilon$ and the supplied surrogate $\widehat Y=X\Theta+\eta$, where $\eta=\Gamma\epsilon+\xi/2$, $\Gamma=I_q+.8(\mathbf1_q\mathbf1_q^\top-I_q)$, and independent $\epsilon,\xi$ have covariance $\Sigma_\epsilon=(1-\rho_e)I_q+\rho_e\mathbf1_q\mathbf1_q^\top$. The joint law of the covariates and surrogate is common to the two inference samples. The surrogate contains information beyond the displayed regression covariates and is not a deterministic function of them.

Using feature-major score vectorization and $r_s=n/N=1/5$, the population score blocks and Jacobian inverse are
\begin{equation*}V=\Sigma_X\otimes\Sigma_\epsilon,\quad
 C=\Sigma_X\otimes(\Sigma_\epsilon\Gamma^\top),\quad
 D=(1+r_s)\Sigma_X\otimes(\Gamma\Sigma_\epsilon\Gamma^\top+\Sigma_\epsilon/4),\quad
 J^{-1}=\Sigma_X^{-1}\otimes I_q.\end{equation*}
For $W=J^{-\top}J^{-1}$, all three classes minimize the same parameter trace risk. Their coefficients are
\begin{equation*}L_*=CD^{-1},\qquad
 \omega_*=(W\circ D)^{-1}\operatorname{diag}(C^\top W),\qquad
 \ell_*={\operatorname{tr}(WC)}/{\operatorname{tr}(WD)},\end{equation*}
where $\circ$ is the entrywise product, the diagonal class uses $\operatorname{Diag}(\omega_*)$, and $\operatorname{diag}$ here extracts a vector. The expressions follow by differentiating the quadratic trace risk. Covariances are computed analytically; no Monte Carlo approximation or truncation of $L_*$ is used. The nested-class ordering is full $\le$ diagonal $\le$ scalar.

\begin{table}[H]\centering\small
\caption{Exact population loss of the optimal diagonal and scalar classes relative to the full correction (percent). In this symmetric design both restricted classes coincide; values do not depend on $p$.}
\label{tab:sim-multi}
\begin{adjustbox}{max width=\linewidth}\begin{tabular}{rrrrr}\toprule
$\rho_e$&$q=1$&$q=2$&$q=3$&$q=5$\\\midrule
0.2&0.000&0.565&2.218&4.933\\
0.7&0.000&0.332&1.455&3.432\\
\bottomrule\end{tabular}\end{adjustbox}
\end{table}

Cross-outcome information improves the full correction in this model; symmetry makes the diagonal and scalar optima coincide. At $\rho_e=.7$, their relative excess risks are approximately .332, 1.455 and 3.432 percent for $q=2,3,5$. The unrestricted optimum is defined in Section~\ref{sec:asymptotic}.

\subsection{Logistic regression}
\label{subsec:sim-logistic}

Table~\ref{tab:new-point} shows that accuracy gains extend to nonlinear logistic moments. The reported sizes use convergence-checked fits and all 200 paired responses in each supervised denominator. Supplementary Section~\ref{app:empirical-diagnostics} records smaller-size solver diagnostics.

The representative results are reported in Table~\ref{tab:new-point}.
\subsection{Numerical validation designs}

The numerical validation in Section~\ref{sec:finite-sample} uses Gaussian covariates with $\Sigma_{ij}=0.6^{|i-j|}$, independent unit Gaussian response errors, and a unit-length slope proportional to $\operatorname{linspace}(1,0.5,d)$. The external predictor is $f(x)=x^\top\alpha$, where $\alpha=\beta+r\mathbf1/\sqrt{\mathbf1^\top\Sigma\mathbf1}$. The eight laws combine $d\in\{2,5\}$, $n\in\{40,160\}$ and $r\in\{0.2,0.5\}$, with $N=5n$. Covariates are shared across paired discrepancies. All candidate strengths use the same correction, spectral cap and covariance ridge. Repeated responses at a fixed design estimate conditional risk; independently sampled designs measure its variation across covariate samples.

The feasible-selection comparison uses eight independent designs per law and, for each design, 400 training responses and 400 independent evaluation responses. Its penalty coefficient is prespecified. The SURE selector uses observed data, whereas the fixed-strength diagnostic reference also uses the generating parameter. Training and evaluation responses are kept separate when comparing the resulting risks.

The fixed-path comparison in Table~\ref{tab:fs-path-space} uses 16 designs per law. Two independent streams of 4,000 responses train the fixed-strength references on a 101-point grid; a third stream of 8,000 responses evaluates risk. The 128 law--design combinations therefore use 2,048,000 response datasets. Independent reference training assesses selection variability; grid approximation and Monte Carlo uncertainty are treated separately, as described in Supplementary Section~\ref{app:path-proofs}.

The conditional projection calculations use four fixed designs with $d\in\{3,8\}$. For each design, eight OLS residual draws are combined with three independent streams of 1,536 Gaussian responses per residual. The generating parameter is used to evaluate the theoretical decomposition. These conditional calculations illustrate the projection identities and are separate from the eight-law risk comparison in Table~\ref{tab:fs-path-space}. Gaussian/chi-square quadrature at orders 40 and 80 is used for the numerical illustration following \eqref{eq:fs-shift-budget}; the analytic bounds establish a nonempty prediction-error tolerance without relying on a certified quadrature error.

\begin{figure}[htbp]\centering
\begin{tikzpicture}
\begin{axis}[width=.91\linewidth,height=8.2cm,
 xmin=0,xmax=3.0,ymin=.4,ymax=8.6,
 xtick={0,.5,1,1.5,2,2.5,3,3.5},ytick={8,7,6,5,4,3,2,1},
 yticklabels={{2, 40, 0.2},{2, 40, 0.5},{2, 160, 0.2},{2, 160, 0.5},{5, 40, 0.2},{5, 40, 0.5},{5, 160, 0.2},{5, 160, 0.5}},
 xlabel={Relative risk reduction (\%)},
 ylabel={$(d,n,r)$},tick label style={font=\small},label style={font=\small},
 xmajorgrids=true,grid style={gray!20},axis line style={gray!65},
 legend cell align=left,legend style={font=\small,draw=none,at={(.5,-.20)},anchor=north,legend columns=1}]
\addplot[only marks,mark=triangle*,mark size=2.8pt,orange!80!black] coordinates {(2.566250,8) (2.350249,7) (2.432190,6) (1.299909,5) (1.700387,4) (2.365690,3) (1.348259,2) (1.787697,1)};
\addlegendentry{Further reduction of independent reference relative to half strength}
\end{axis}
\end{tikzpicture}
\caption{Remaining percentage reduction from half strength to the independently trained fixed-strength reference, computed from Table~\ref{tab:fs-path-space}. Risks are law-average OLS-normalized values.}\label{fig:main-calibration-gain}
\end{figure}

\par\medskip\noindent\textbf{Focused calibration comparison}\label{supp:calibration-mechanism}\par\nobreak\smallskip
The only methods in this comparison are the protected full fits at strengths zero and one half. We use $(d,n,N)=(2,60,300)$ and $(5,100,500)$, with no additional intercept, $a=1$, $\eta=n^{-2}$ and $\kappa=1/2$. For the linear model, each coordinate has six labeled and six unlabeled nonzero rows, alternating signs, with zeros elsewhere. Hence $p=6/n$, $u=6/N$, and both designs satisfy the conditions of Theorem~\ref{thm:linear-multiparameter}. Predictions are fixed at $f=0$, $Y=X\beta+\sigma Z$, and $Z$ is standard Gaussian. The noise scan fixes $\beta=0$ and uses $\sigma=q\sqrt\eta$. The mismatch scan fixes $q=1$ and $\beta=c(\sigma/\sqrt6)\mathbf1_d/\sqrt d$. The same 10,000 response vectors are used throughout the scans within a dimension. Every response recomputes the pilot, centered covariance and protected update. All fits are finite; batched fits agree with the production routine to $2.1\times10^{-17}$ in maximum absolute parameter difference.

Table~\ref{tab:calibration-mechanism-linear} gives both risks and pointwise paired delta-method normal Monte Carlo intervals for the percentage risk gain, accounting for the estimated GMM-risk denominator and its dependence on the paired risk difference, conditional on the fixed designs. These are not simultaneous or random-design uncertainty bounds. The scaled gain $(R(0)-R(1/2))/\sigma^4$ approaches the leading coefficient in the linear expansion as noise decreases; no theoretical noise threshold is estimated from the plot.

\begin{table}[H]\centering\scriptsize
\caption{Linear calibration comparison: risks and paired gain intervals. Gain and its interval are percentages of GMM risk.}\label{tab:calibration-mechanism-linear}
\begin{adjustbox}{max width=\linewidth}\begin{tabular}{rlrrrrrr}\toprule
$d$&Scan&$q$&$c$&$R(0)$&$R(1/2)$&Gain (\%)&95\% MC interval\\\midrule
2&noise&0.05&0&8.622e-08&8.622e-08&0.0015&[0.0014, 0.0015]\\
2&noise&0.1&0&3.448e-07&3.448e-07&0.0059&[0.0057, 0.0061]\\
2&noise&0.2&0&1.378e-06&1.378e-06&0.0236&[0.0229, 0.0243]\\
2&noise&0.5&0&8.561e-06&8.548e-06&0.1459&[0.1418, 0.1500]\\
2&noise&1&0&3.356e-05&3.337e-05&0.5617&[0.5465, 0.5769]\\
2&noise&2&0&0.000126&0.0001236&1.8914&[1.8566, 1.9263]\\
2&noise&5&0&0.000662&0.0006356&3.9954&[3.9616, 4.0293]\\
2&mismatch&1&0&3.356e-05&3.337e-05&0.5617&[0.5465, 0.5769]\\
2&mismatch&1&0.25&3.407e-05&3.388e-05&0.5583&[0.5424, 0.5743]\\
2&mismatch&1&0.5&3.557e-05&3.538e-05&0.5492&[0.5313, 0.5672]\\
2&mismatch&1&1&4.15e-05&4.129e-05&0.5068&[0.4849, 0.5286]\\
2&mismatch&1&2&6.451e-05&6.434e-05&0.2740&[0.2514, 0.2966]\\
2&mismatch&1&4&0.000147&0.0001483&-0.9008&[-0.9225, -0.8791]\\
2&mismatch&1&8&0.0003652&0.0003848&-5.3710&[-5.4269, -5.3151]\\
5&noise&0.05&0&7.731e-08&7.731e-08&0.0009&[0.0008, 0.0009]\\
5&noise&0.1&0&3.092e-07&3.092e-07&0.0035&[0.0034, 0.0035]\\
5&noise&0.2&0&1.236e-06&1.236e-06&0.0138&[0.0135, 0.0141]\\
5&noise&0.5&0&7.696e-06&7.69e-06&0.0857&[0.0840, 0.0873]\\
5&noise&1&0&3.038e-05&3.028e-05&0.3346&[0.3283, 0.3409]\\
5&noise&2&0&0.0001162&0.0001148&1.2198&[1.2006, 1.2389]\\
5&noise&5&0&0.0006175&0.000596&3.4915&[3.4759, 3.5070]\\
5&mismatch&1&0&3.038e-05&3.028e-05&0.3346&[0.3283, 0.3409]\\
5&mismatch&1&0.25&3.055e-05&3.044e-05&0.3345&[0.3281, 0.3408]\\
5&mismatch&1&0.5&3.105e-05&3.095e-05&0.3333&[0.3266, 0.3400]\\
5&mismatch&1&1&3.308e-05&3.297e-05&0.3273&[0.3197, 0.3350]\\
5&mismatch&1&2&4.118e-05&4.106e-05&0.2911&[0.2819, 0.3003]\\
5&mismatch&1&4&7.287e-05&7.282e-05&0.0704&[0.0620, 0.0787]\\
5&mismatch&1&8&0.000188&0.0001899&-1.0057&[-1.0165, -0.9949]\\
\bottomrule\end{tabular}\end{adjustbox}\end{table}

For logistic regression, each coordinate instead has three labeled rows $e_j$ and six unlabeled rows $e_j$, with zero rows elsewhere. Predictions are $1/2$, and $\beta=s\mathbf1_d/\sqrt d$. We compute exact conditional risk on $\mathcal E=\{K_j\in\{1,2\}\ \forall j\}$ using the closed-form path, and independently verify it by enumerating all $2^d$ conditional sign patterns with full covariance matrices. Maximum fit and risk discrepancies are $1.4\times10^{-15}$ and $8.0\times10^{-15}$. At $\beta=0$, $P(\mathcal E)=(3/4)^d$. Table~\ref{tab:calibration-mechanism-logistic} reports the event probability throughout the signal scan. These exact conditional calculations do not replace the failure accounting in the general logistic experiments.
\begin{table}[H]\centering\small
\caption{Exact logistic conditional risks; no Monte Carlo approximation.}\label{tab:calibration-mechanism-logistic}
\begin{adjustbox}{max width=\linewidth}\begin{tabular}{rrrrrr}\toprule
$d$&$\|\beta\|$&$P(\mathcal E)$&$R^c(0)$&$R^c(1/2)$&Gain (\%)\\\midrule
2&0&0.5625&0.07265&0.06167&15.117\\
2&0.2&0.5569&0.10504&0.09465&9.885\\
2&0.4&0.5405&0.20236&0.19376&4.249\\
2&0.6&0.5144&0.36505&0.35938&1.551\\
2&1&0.4403&0.88962&0.89304&-0.384\\
2&1.5&0.3285&1.92996&1.94988&-1.032\\
2&2&0.2228&3.41621&3.45688&-1.190\\
2&3&0.0822&7.80166&7.89068&-1.141\\
2&4&0.0250&14.15704&14.29678&-0.987\\
5&0&0.2373&0.15495&0.13104&15.429\\
5&0.2&0.2349&0.18791&0.16457&12.421\\
5&0.4&0.2280&0.28686&0.26521&7.547\\
5&0.6&0.2169&0.45195&0.43311&4.169\\
5&1&0.1852&0.98179&0.97180&1.017\\
5&1.5&0.1366&2.02307&2.02986&-0.335\\
5&2&0.0901&3.49426&3.52345&-0.836\\
5&3&0.0291&7.77205&7.85930&-1.123\\
5&4&0.0068&13.90783&14.06454&-1.127\\
\bottomrule\end{tabular}\end{adjustbox}\end{table}

\subsection{Simulation and real-data implementation}
\label{app:implementation}

Linear simulations use 1,000 paired response datasets for each design in the 144 design--configuration combinations of Section~\ref{sec:simulation}. Logistic simulations use 200 response datasets for each of 48 design--configuration combinations. Failed fits, including underidentified RePPI folds, remain in the reported comparison. The noisy-surrogate example evaluates 32 exact population configurations and solves the optimal diagonal correction problem rather than taking the diagonal of the optimal full matrix.

Real-data evaluation uses two seeds and 50 labeled subsets per seed. The prediction rule is fitted on the separate training region and held fixed across all inference subsets. The fixed-strength reference is trained on 50 further subsets using the full-region proxy parameter. Table~\ref{tab:real-preprocessing} specifies the outcomes, covariates and train--evaluation splits. Each inference design includes an intercept. Target covariates, including coded indicators, are standardized using means and standard deviations from the predictor-training region; a zero standard deviation is replaced by one. Outcomes retain their original scales. The predictor uses the unstandardized feature inputs.

\begin{table}[H]\centering\small
\caption{Real-data outcomes and preprocessing. Listed target covariates exclude the intercept. Evaluation-pool sizes and dimensions are given in Table~\ref{tab:real-datasets}.}
\label{tab:real-preprocessing}
\begin{adjustbox}{max width=\linewidth}\begin{tabular}{p{0.19\textwidth}p{0.73\textwidth}}\toprule
Dataset & Outcome, covariates and split\\\midrule
California Housing & Outcome: median house value. The point-estimation target uses median income, house age, average rooms, average bedrooms, population, latitude and longitude, plus an intercept ($d=8$). The frozen predictor additionally uses average occupancy. The retained interval sensitivity uses the original $d=9$ target. Train at or above the latitude median; evaluate below it.\\[3pt]
Bike Sharing & Outcome: hourly rental count. Target: temperature, humidity, wind speed, holiday and working-day indicators. The predictor additionally uses season, year, month, hour, weekday, weather and apparent temperature. Train on the first 70\% of rows in temporal order; use the first 5,000 remaining rows for evaluation.\\[3pt]
ACS Income & Outcome: annual income above \$50,000. Target: age, sex and weekly working hours. Predictor inputs: \texttt{AGEP, COW, SCHL, MAR, OCCP, POBP, RELP, WKHP, SEX, RAC1P}, with categorical codes retained as numeric inputs. Train on 2018; evaluate on the first 50,000 observations in the 2019 data.\\[3pt]
Wine Quality & Outcome: quality score. Target and predictor use all 11 physicochemical measurements. Train on red wine and evaluate on all white-wine observations.\\\bottomrule
\end{tabular}\end{adjustbox}
\end{table}

The fixed XGBoost predictors contain 100 trees. The training configuration uses maximum depth 6, learning rate 0.1 and random seed 42, with squared-error regression for the three continuous outcomes and binary logistic loss for ACS Income. The prediction rule is not retrained when the labeled inference subset changes.

Linear pilots and supervised cross-fitting folds must have full column rank. Nonlinear score solving uses residual tolerance $10^{-8}$ and at most 150 iterations. GMM least squares initially permits 300 evaluations and accepts a gradient residual at most $10^{-7}$. A normally terminated fit with residual in $(10^{-7},10^{-5}]$ is refined on the same frozen objective, with objective and parameter tolerances $10^{-14}$, gradient tolerance $10^{-12}$ and at most 300 additional evaluations; the acceptance threshold remains $10^{-7}$. The final objective must not exceed its value at the pilot by more than $10^{-10}$. These criteria check the computed solution without establishing global nonlinear minimization. RePPI covariance pseudoinversion uses relative cutoff $10^{-6}$. The remaining numerical ridges are specified below.

\par\medskip\noindent{\bfseries Comparator definitions and numerical conventions}\par\nobreak\smallskip

PPI solves the unweighted prediction-corrected score equation. PPI++ uses the scalar plug-in in equation (8) of \citet{angelopoulos2024ppiplus}, evaluated at the supervised pilot and clipped to $[0,1]$. PSPA follows Algorithm 1 of \citet{miao2025pspa}, with its printed upper cap and the full corrected-score Jacobian; the additional lower cap is a separately named sensitivity variant. PDC follows the one-step Algorithm 1 of \citet{gan2024prediction}, with the sample-size factor $N/(n+N)$ and the same supplied prediction scores.

RePPI follows Algorithm 1 of \citet{ji2025predictions}: the imputed-score factor uses the total labeled ratio $n/N$, and the three rotated estimates are aggregated by evaluation-fold size. The real-data recalibration learner fits the outcome on the target covariates and supplied prediction; the simulation's linear prediction is already in the span of its target covariates, so the linear learner uses those covariates and an intercept. Linear recalibration uses a numerical ridge of $10^{-10}$; real-data score covariance inversion uses the stated pseudoinverse cutoff. These learner and stabilization choices are part of the implementation rather than claims that every implementation of RePPI is identical.

Linear PPI, PPI++, PSPA and PDC use their displayed equations without additional covariance ridges. In logistic regression, the pilot-Hessian inverse used for scalar and coordinate tuning and the PDC auxiliary covariance inverse use numerical ridges of $10^{-6}$; score Newton steps and one-step parameter solves use $10^{-8}$. CC-PPI uses only its specified covariance ridge $\eta=n^{-2}$ in the frozen GMM objective, followed by the spectral cap $\kappa=1/2$ and blend $\lambda=1/2$. Both logistic score blocks are sign-reversed relative to the loss-gradient convention, which leaves the covariance and quadratic criterion unchanged.

\subsection{Extended empirical diagnostics}
\label{app:empirical-diagnostics}

\par\medskip\noindent{\bfseries Logistic simulation at smaller labeled sizes}\par\nobreak\smallskip

\begin{table}[H]\centering\small
\caption{Failed logistic fits at $n=100$, $r=.2$, out of 400 responses across two designs. These counts diagnose the implemented solvers, not a universal statistical threshold.}
\label{tab:sim-logistic-failures}
\begin{adjustbox}{max width=\linewidth}\begin{tabular}{rrrr}\toprule
$d$&Sup.&PDC&CC\\\midrule
2&0&0&0\\
5&0&0&0\\
10&0&0&0\\
20&0&0&2\\
\bottomrule\end{tabular}\end{adjustbox}
\end{table}

\subsection{Complete pooled coverage and length comparisons}
\label{supp:coverage-tolerance}
The 92--98\% range used in Section~\ref{subsec:new-interval-experiment} is an author-selected post hoc descriptive tolerance around nominal 95\%. The lower cutoff defines configurations requiring repair; it does not change the interval procedure or its nominal level. The following table includes all 18 configurations, each pooling two fixed designs with 1,000 responses per design. Table~\ref{tab:main-repair} includes all 18 configurations, including the seven below the Wald cutoff. Proposed lengths exceed OLS $t$ at $(n,d)=(500,2),(2000,2),(10000,10)$, so a length advantage is not asserted there.
\begin{table}[H]\centering
\caption{All pooled conditional-design configurations: coordinate-average coverage and relative interval length.}\label{tab:all-pooled-intervals}
\begin{center}\small
\begin{tabular}{rrllll}
\hline
$n$ & $d$ & CC Wald (\%) & Expanded (\%) & OLS $t$ (\%) & Length / OLS $t$ (\%) \\
\hline
50 & 2 & 73.78 & 97.15 & 94.33 & 68.9 \\
50 & 5 & 90.31 & 93.92 & 95.00 & 77.2 \\
50 & 10 & 87.62 & 93.34 & 95.07 & 68.6 \\
50 & 20 & 78.23 & 93.46 & 94.70 & 72.0 \\
100 & 2 & 78.45 & 93.85 & 95.45 & 80.6 \\
100 & 5 & 92.40 & 93.96 & 95.29 & 73.8 \\
100 & 10 & 91.49 & 93.12 & 95.03 & 69.6 \\
100 & 20 & 88.73 & 92.84 & 94.78 & 66.0 \\
500 & 2 & 92.73 & 94.43 & 94.43 & 102.5 \\
500 & 5 & 93.58 & 94.00 & 95.23 & 83.9 \\
500 & 10 & 93.51 & 92.47 & 94.85 & 75.8 \\
500 & 20 & 94.02 & 92.34 & 94.73 & 70.8 \\
1000 & 10 & 95.52 & 95.06 & 95.04 & 85.1 \\
2000 & 2 & 94.60 & 97.00 & 94.98 & 138.0 \\
2000 & 5 & 94.31 & 95.42 & 95.20 & 99.5 \\
2000 & 10 & 95.12 & 95.64 & 95.29 & 92.0 \\
2000 & 20 & 94.37 & 93.85 & 95.12 & 85.0 \\
10000 & 10 & 94.96 & 96.56 & 95.17 & 110.6 \\
\hline
\end{tabular}
\end{center}

\end{table}
Coverage is averaged over coordinates within a response trial. Monte Carlo errors therefore use those trial-level averages, retaining within-trial coordinate dependence. For fixed designs, response-based Monte Carlo uncertainty is conditional on the retained designs; pooling two design seeds does not certify performance over arbitrary designs. The complete seed table below retains the one proposed interval coverage below 92\%. Neither the pooled classification nor the tolerance gives simultaneous or coordinatewise coverage guarantees.

\subsection{Simulation results by design seed}
Each row below is one setting with 1,000 fresh response draws at the indicated fixed designs, $N=5n$ and discrepancy 0.2. Entries give coordinate-average coverage in percent / mean interval length at nominal 95\%. This table preserves all 36 seed-specific settings behind the pooled main table; every configuration has two design seeds. The CC Wald and new CC entries use the same observed point. All their intervals in these completed simulations are defined.
\begingroup\small
\begin{longtable}{rrrlll}
\hline
$n$ & $d$ & Seed & CC Wald & CC new & OLS $t$ \\
\hline\endfirsthead
\hline
$n$ & $d$ & Seed & CC Wald & CC new & OLS $t$ \\
\hline\endhead
50 & 2 & 2026101301 & 67.25 / 0.3962 & 96.45 / 0.5882 & 93.75 / 0.8335 \\
50 & 2 & 2026101302 & 80.30 / 0.4266 & 97.85 / 0.5316 & 94.90 / 0.7917 \\
50 & 5 & 2026101301 & 90.72 / 0.5222 & 94.10 / 0.5732 & 95.34 / 0.7327 \\
50 & 5 & 2026101302 & 89.90 / 0.5479 & 93.74 / 0.6024 & 94.66 / 0.7909 \\
50 & 10 & 2026101301 & 87.61 / 0.5540 & 92.92 / 0.6478 & 94.97 / 0.9159 \\
50 & 10 & 2026101302 & 87.62 / 0.5076 & 93.76 / 0.6188 & 95.16 / 0.9305 \\
50 & 20 & 2026101301 & 78.47 / 0.4695 & 93.27 / 0.7310 & 94.75 / 1.0355 \\
50 & 20 & 2026101302 & 77.99 / 0.4369 & 93.63 / 0.7063 & 94.66 / 0.9601 \\
100 & 2 & 2026101301 & 89.70 / 0.4427 & 93.35 / 0.4500 & 95.20 / 0.5262 \\
100 & 2 & 2026101302 & 67.20 / 0.3020 & 94.35 / 0.4247 & 95.70 / 0.5594 \\
100 & 5 & 2026101301 & 92.88 / 0.3905 & 94.22 / 0.4030 & 95.56 / 0.5526 \\
100 & 5 & 2026101302 & 91.92 / 0.3724 & 93.70 / 0.3876 & 95.02 / 0.5180 \\
100 & 10 & 2026101301 & 91.45 / 0.4012 & 92.98 / 0.4092 & 95.20 / 0.5618 \\
100 & 10 & 2026101302 & 91.53 / 0.3800 & 93.25 / 0.3990 & 94.86 / 0.5987 \\
100 & 20 & 2026101301 & 89.31 / 0.3875 & 92.97 / 0.4321 & 94.88 / 0.6366 \\
100 & 20 & 2026101302 & 88.15 / 0.3669 & 92.70 / 0.4198 & 94.67 / 0.6551 \\
500 & 2 & 2026101301 & 94.30 / 0.2194 & 96.35 / 0.2344 & 94.45 / 0.2239 \\
500 & 2 & 2026101302 & 91.15 / 0.2106 & 92.50 / 0.2222 & 94.40 / 0.2214 \\
500 & 5 & 2026101301 & 91.68 / 0.1973 & 92.44 / 0.1978 & 94.84 / 0.2438 \\
500 & 5 & 2026101302 & 95.48 / 0.2186 & 95.56 / 0.2090 & 95.62 / 0.2411 \\
500 & 10 & 2026101301 & 94.67 / 0.2081 & 93.51 / 0.1933 & 94.97 / 0.2535 \\
500 & 10 & 2026101302 & 92.34 / 0.2010 & 91.42 / 0.1903 & 94.73 / 0.2530 \\
500 & 20 & 2026101301 & 93.59 / 0.1965 & 92.19 / 0.1828 & 94.58 / 0.2594 \\
500 & 20 & 2026101302 & 94.44 / 0.2033 & 92.49 / 0.1847 & 94.88 / 0.2598 \\
1000 & 10 & 2026101301 & 95.66 / 0.1585 & 95.33 / 0.1493 & 94.78 / 0.1780 \\
1000 & 10 & 2026101302 & 95.37 / 0.1611 & 94.79 / 0.1520 & 95.29 / 0.1762 \\
2000 & 2 & 2026101301 & 94.60 / 0.1094 & 96.95 / 0.1518 & 94.75 / 0.1104 \\
2000 & 2 & 2026101302 & 94.60 / 0.1080 & 97.05 / 0.1507 & 95.20 / 0.1089 \\
2000 & 5 & 2026101301 & 94.14 / 0.1138 & 95.20 / 0.1191 & 95.52 / 0.1199 \\
2000 & 5 & 2026101302 & 94.48 / 0.1146 & 95.64 / 0.1202 & 94.88 / 0.1206 \\
2000 & 10 & 2026101301 & 95.18 / 0.1189 & 95.88 / 0.1157 & 95.27 / 0.1257 \\
2000 & 10 & 2026101302 & 95.06 / 0.1184 & 95.40 / 0.1159 & 95.30 / 0.1261 \\
2000 & 20 & 2026101301 & 94.47 / 0.1123 & 93.71 / 0.1068 & 95.19 / 0.1267 \\
2000 & 20 & 2026101302 & 94.27 / 0.1158 & 93.99 / 0.1089 & 95.06 / 0.1271 \\
10000 & 10 & 2026101301 & 95.00 / 0.0546 & 96.71 / 0.0615 & 95.21 / 0.0554 \\
10000 & 10 & 2026101302 & 94.92 / 0.0549 & 96.41 / 0.0613 & 95.12 / 0.0556 \\
\hline\end{longtable}
\endgroup

\subsection{Validation of partial anchor expansion}
The fixed $3/4$ rule was selected from the disclosed discovery grid $\{0,.25,.5,.75,1\}$, then tested on 9,600 fresh Gaussian-response datasets using an earlier Gaussian-design radial bootstrap parent. Of these, 6,400 were main experiments, 1,600 direction controls, 800 poor-predictor controls and 800 scalar local controls. Multivariate datasets used 999 inner draws; scalar controls used 399. These experiments explain the endpoint rule and do not use the current empirical joint-moment resampling law. They are not added to the current 36,000-trial denominator and are not covered by the current theorem by substitution.

The eight pooled main configurations, each containing four fixed designs and 800 responses, are retained below. Entries are coverage in percent / mean length as a percentage of OLS $t$ length. The reference OLS coverages ranged from 94.81\% to 95.93\%.
\begin{center}\small
\begin{adjustbox}{max width=\linewidth}\begin{tabular}{rrlll}\hline
$n$ & $d$ & Parent & Partial $3/4$ & Full anchor\\\hline
50&2&96.19 / 86.60&98.00 / 87.23&98.00 / 87.44\\
100&2&94.38 / 87.75&96.19 / 88.76&96.19 / 89.10\\
50&5&93.88 / 80.59&96.68 / 81.15&96.95 / 81.34\\
100&5&93.75 / 81.29&96.25 / 81.89&96.45 / 82.09\\
50&10&93.84 / 77.46&95.79 / 77.92&96.83 / 78.07\\
100&10&93.11 / 77.08&95.69 / 77.62&96.08 / 77.80\\
50&20&93.86 / 74.61&95.39 / 75.01&96.77 / 75.14\\
100&20&93.57 / 74.86&95.51 / 75.30&96.49 / 75.45\\\hline
\end{tabular}\end{adjustbox}
\end{center}
The direction controls gave 95.58\% and 95.70\% coverage at 81.98\% and 79.38\% OLS length. The poor-predictor $d=5,n=50$ control reached 95.55\% at 100.78\% OLS length, losing its length advantage; the $d=20,n=50$ control gave 95.97\% at 87.78\%. The scalar local $n=400$ control remained undercovered at 90.25\%, while $n=1600$ gave 95.00\%. One individual old $d=10,n=50$ design gave only 93.85\%. Thus the positive pooled results did not establish all-design or all-coordinate repair. All control and design-specific outputs remain in the original report.

\subsection{Data provenance and regeneration}
The numerical archive provides \path{simulation95.json} and \path{realdata95.json}, including all implemented method variants, distinct original-source and manuscript PPI/RePPI conventions, and failures. The earlier partial-anchor validation has a separate archive and summary, \path{synthesis-en.md}. Method-wise summaries, interval endpoints and source hashes accompany the results.
\subsection{Diagnostics for the supporting intervals}
\label{app:ci-focused-diagnostics}
\begin{table}[H]\centering
\caption{New interval validation: coordinate-average coverage (\%) / mean length at nominal 95\%. iid trials redraw the complete samples; fixed trials condition on two design pairs.}\label{tab:new-interval}
\begingroup\fontsize{9}{11}\selectfont
\setlength{\tabcolsep}{3pt}
\begin{adjustbox}{max width=\linewidth}\begin{tabular}{lrrllll}
\hline
Design & $n$ & $d$ & CC Wald & Parent & Expanded & OLS $t$ \\
\hline
iid & 100 & 2 & 85.90/0.375 & 91.50/0.403 & 94.30/0.411 & 95.40/0.502 \\
iid & 500 & 2 & 93.50/0.206 & 92.10/0.201 & 95.70/0.221 & 94.60/0.220 \\
iid & 100 & 5 & 93.32/0.392 & 91.28/0.395 & 94.64/0.400 & 95.20/0.559 \\
iid & 500 & 5 & 92.92/0.213 & 88.76/0.201 & 93.28/0.207 & 93.96/0.242 \\
fixed & 50 & 2 & 76.80/0.395 & 93.70/0.515 & 96.20/0.522 & 94.10/0.874 \\
fixed & 50 & 5 & 91.36/0.498 & 91.76/0.570 & 94.88/0.576 & 95.88/0.861 \\
fixed & 100 & 5 & 92.96/0.412 & 90.36/0.405 & 94.72/0.410 & 94.68/0.561 \\
fixed & 500 & 10 & 95.14/0.212 & 89.20/0.191 & 94.16/0.195 & 94.96/0.250 \\
\hline
\end{tabular}\end{adjustbox}
\endgroup

\end{table}
The focused validation uses four iid and four fixed-design configurations, each with 500 trials and 399 inner draws. All 4,000 observed intervals are defined. The batched linear computation was replayed against the unchanged production routine at $(n,d)=(50,2),(100,5),(500,10)$ using the identical exponential weights and 399 inner draws; the maximum discrepancy across inner points, scales and interval endpoints was $3.1\times10^{-14}$. The same fitting rule is used throughout; only the expansion fraction varies in the sensitivity table.

\begin{table}[H]\centering
\caption{Expansion sensitivity on the same new fixed-design trials: coverage (\%) / mean length. Columns give the expansion fraction $\gamma$.}\label{tab:ci-gamma}
\begingroup\small
\setlength{\tabcolsep}{3pt}
\begin{adjustbox}{max width=\linewidth}\begin{tabular}{rrlllll}
\hline
$n$ & $d$ & $0$ & $1/4$ & $1/2$ & $3/4$ & $1$ \\
\hline
50 & 2 & 93.70/0.515 & 94.30/0.518 & 95.60/0.520 & 96.20/0.522 & 96.20/0.524 \\
50 & 5 & 91.76/0.570 & 92.28/0.572 & 93.52/0.574 & 94.88/0.576 & 95.24/0.577 \\
100 & 5 & 90.36/0.405 & 91.64/0.407 & 92.96/0.409 & 94.72/0.410 & 94.92/0.412 \\
500 & 10 & 89.20/0.191 & 90.88/0.192 & 92.80/0.194 & 94.16/0.195 & 94.20/0.196 \\
\hline
\end{tabular}\end{adjustbox}
\endgroup

\end{table}
For the empirical joint-moment bootstrap at $n=100,d=5$, partial expansion raises coverage from 90.36\% to 94.72\% and mean length from 0.4047 to 0.4104 on the same fixed-design trials. At $n=500,d=10$, full expansion gives 94.20\% coverage, compared with 95.14\% for Wald. The sensitivity analysis evaluates fixed fractions and does not justify selecting a fraction from evaluation coverage.

\begin{table}[H]\centering
\caption{Paired coverage changes in percentage points, with outer-trial Monte Carlo standard error in parentheses.}\label{tab:ci-paired}
\begingroup\small
\setlength{\tabcolsep}{3pt}
\begin{adjustbox}{max width=\linewidth}\begin{tabular}{lrrlll}
\hline
Design & $n$ & $d$ & Parent--Wald & Expanded--parent & Expanded--Wald \\
\hline
iid & 100 & 2 & 5.60 (1.08) & 2.80 (0.51) & 8.40 (1.11) \\
iid & 500 & 2 & -1.40 (0.53) & 3.60 (0.58) & 2.20 (0.63) \\
iid & 100 & 5 & -2.04 (0.38) & 3.36 (0.36) & 1.32 (0.38) \\
iid & 500 & 5 & -4.16 (0.42) & 4.52 (0.41) & 0.36 (0.42) \\
fixed & 50 & 2 & 16.90 (1.25) & 2.50 (0.49) & 19.40 (1.26) \\
fixed & 50 & 5 & 0.40 (0.35) & 3.12 (0.36) & 3.52 (0.42) \\
fixed & 100 & 5 & -2.60 (0.42) & 4.36 (0.39) & 1.76 (0.42) \\
fixed & 500 & 10 & -5.94 (0.35) & 4.96 (0.29) & -0.98 (0.29) \\
\hline
\end{tabular}\end{adjustbox}
\endgroup

\end{table}
\begin{table}[H]\centering
\caption{Outer-trial standard errors in percentage points, minimum empirical expanded coverage over coordinates (\%), and mean seconds per trial for the batched validation.}\label{tab:ci-uncertainty}
\begingroup\small
\setlength{\tabcolsep}{3pt}
\begin{adjustbox}{max width=\linewidth}\begin{tabular}{lrrrrrrr}
\hline
Design & $n$ & $d$ & Wald SE & Parent SE & Expanded SE & Min. coord. & Seconds \\
\hline
iid & 100 & 2 & 1.22 & 0.93 & 0.71 & 93.20 & 0.020 \\
iid & 500 & 2 & 0.85 & 0.91 & 0.63 & 95.60 & 0.077 \\
iid & 100 & 5 & 0.52 & 0.63 & 0.45 & 93.00 & 0.032 \\
iid & 500 & 5 & 0.56 & 0.69 & 0.47 & 92.80 & 0.132 \\
fixed & 50 & 2 & 1.36 & 0.82 & 0.59 & 94.40 & 0.016 \\
fixed & 50 & 5 & 0.61 & 0.62 & 0.45 & 92.20 & 0.023 \\
fixed & 100 & 5 & 0.55 & 0.63 & 0.43 & 91.60 & 0.031 \\
fixed & 500 & 10 & 0.34 & 0.51 & 0.34 & 91.40 & 0.439 \\
\hline
\end{tabular}\end{adjustbox}
\endgroup

\end{table}
Standard errors use the coordinate-average inclusion for each outer trial; coordinates within that trial are not treated as independent repetitions. For fixed designs, the pooled standard error also reflects differences between the two fixed design pairs. The minimum-coordinate column is a descriptive diagnostic, not a simultaneous lower confidence bound. Timing covers the batched CC fitting and interval calculation, including the observed fit and 399 inner draws, and is machine-specific; it is not the runtime of the original sequential routine. Full paired summaries and coordinate-specific inclusion are retained with the numerical release.

The 36-setting conditional-design experiment and focused validation have separate denominators. Comparisons retain method-specific conventions and failure counts. Source RePPI intervals use their source point estimates. Real-data inclusion concerns the finite-frame coefficient and does not estimate population coverage. Complete comparisons and ablation endpoints are retained in the numerical release.

\section{Data applications and diagnostics}
\label{supp:data}

\subsection{Setup and methods}

We study four datasets with externally trained XGBoost predictors \citep{chen2016xgboost}. California Housing uses a northern training region and a southern evaluation region; Bike Sharing uses an earlier training period and a later evaluation period; ACS Income trains on 2018 and evaluates on 2019; Wine Quality trains on red wine and evaluates on white wine. The shift is between predictor training and inference, whereas labeled and unlabeled inference points are random disjoint subsets of the same evaluation region. The target uses a specified low-dimensional covariate subset, and the predictor may use additional observed features.

For each labeled size, we evaluate 50 random subsets per seed, with two seeds (100 evaluation subsamples total). The target is the supervised parameter on the complete evaluation region, using the empirical-target convention of \citep{ji2025predictions}. All estimators use the same labeled subset and unlabeled complement. The externally fitted prediction rules are held fixed across inference subsets. Outcomes, feature selection, preprocessing and predictor settings are given in Supplementary Section~\ref{app:implementation}.

The main tables report sizes for which the supervised pilot and solver checks support an accuracy comparison. The extended sample-size grid, method-specific failure counts and stability profiles are reported in Supplementary Section~\ref{app:empirical-diagnostics}; a dagger in a main table denotes a failed fit or an unavailable ratio.

The fixed-strength reference CC-PPI$^{\rm ref}$ selects a strength on 50 separate training subsamples using the full-region target, then evaluates it on the common 100 fresh subsamples. It uses 101 linear grid points and 11 logistic points. This independent reference measures the remaining accuracy available along the calibration path.

\subsection{Estimation accuracy}

Figure~\ref{fig:real-data} compares the methods across labeled sample sizes. Tables~\ref{tab:real-california}--\ref{tab:real-wine} report selected sizes, starting at $n=25$ for California Housing and Wine Quality and $n=75$ for ACS Income. Supplementary Section~\ref{app:empirical-diagnostics} retains the complete grid and convergence diagnostics.

\begin{table}[H]\centering\small
\caption{California Housing: MSE relative to supervised (percent).}
\label{tab:real-california}
\begin{adjustbox}{max width=\linewidth}\begin{tabular}{rrrrrrrrr}\toprule
$n$&Sup.&PPI&PPI++&PSPA&PDC&RePPI&CC&CC$^{\rm ref}$\\\midrule
25&100.0&51.3&99.2&$\ge10^3$&$\ge10^3$&$\ge10^3$&32.7&31.6\\
50&100.0&49.9&98.7&197.1&$\ge10^3$&$\ge10^3$&38.6&37.1\\
100&100.0&58.9&97.8&190.3&581.6&$\ge10^3$&46.1&44.2\\
500&100.0&70.1&91.7&124.8&261.6&827.3&58.3&58.6\\
2000&100.0&106.3&85.1&142.3&95.1&112.8&71.8&71.7\\
\bottomrule\end{tabular}\end{adjustbox}
\end{table}

\begin{table}[H]\centering\small
\caption{Bike Sharing: MSE relative to supervised (percent).}
\label{tab:real-bike}
\begin{adjustbox}{max width=\linewidth}\begin{tabular}{rrrrrrrrr}\toprule
$n$&Sup.&PPI&PPI++&PSPA&PDC&RePPI&CC&CC$^{\rm ref}$\\\midrule
300&100.0&16.2&17.5&16.0&19.1&$\dagger$&16.6&16.5\\
500&100.0&17.1&17.8&17.8&22.0&$\dagger$&17.4&17.3\\
1000&100.0&13.7&13.1&12.9&13.2&16.0&12.4&12.3\\
2000&100.0&38.0&15.0&15.2&14.8&16.4&14.7&14.7\\
\bottomrule\end{tabular}\end{adjustbox}
\end{table}

\begin{table}[H]\centering\small
\caption{ACS Income: MSE relative to supervised (percent).}
\label{tab:real-acs}
\begin{adjustbox}{max width=\linewidth}\begin{tabular}{rrrrrrrrr}\toprule
$n$&Sup.&PPI&PPI++&PSPA&PDC&RePPI&CC&CC$^{\rm ref}$\\\midrule
75&100.0&163.5&90.4&80.4&83.2&$\dagger$&77.3&76.1\\
100&100.0&167.1&94.9&84.4&87.2&$\dagger$&74.9&71.4\\
200&100.0&95.6&87.6&80.2&80.8&134.5&77.0&76.0\\
500&100.0&110.3&94.3&89.0&88.8&91.6&86.0&85.0\\
2000&100.0&99.6&86.7&85.8&85.9&79.7&87.2&87.8\\
\bottomrule\end{tabular}\end{adjustbox}
\end{table}

\begin{table}[H]\centering\small
\caption{Wine Quality: MSE relative to supervised (percent).}
\label{tab:real-wine}
\begin{adjustbox}{max width=\linewidth}\begin{tabular}{rrrrrrrrr}\toprule
$n$&Sup.&PPI&PPI++&PSPA&PDC&RePPI&CC&CC$^{\rm ref}$\\\midrule
25&100.0&33.5&99.6&$\ge10^3$&$\ge10^3$&$\dagger$&25.7&24.5\\
50&100.0&51.7&99.3&$\ge10^3$&286.4&$\ge10^3$&45.1&43.4\\
100&100.0&73.1&99.6&169.8&122.4&$\ge10^3$&53.4&51.8\\
500&100.0&74.9&98.5&119.3&96.5&167.8&67.9&65.9\\
2000&100.0&76.0&103.7&150.7&53.8&94.0&65.8&62.8\\
\bottomrule\end{tabular}\end{adjustbox}
\end{table}

At $n=500$, CC-PPI MSE is 58.3\% of supervised MSE for California Housing, 17.4\% for Bike Sharing, 86.0\% for ACS Income and 67.9\% for Wine Quality. For Wine, CC-PPI improves on PDC at $n=50$ (45.1\% versus 286.4\%), whereas PDC is better at $n=2000$ (53.8\% versus 65.8\%). The reduced California point target has $d=8$; the original $d=9$ sensitivity is retained below. Relative performance therefore varies with the labeled sample size and the use of prediction information.

The independent reference gives a measure of the accuracy left beyond half strength. At $n=2000$, the default/reference ratios are 71.8/71.7 on California Housing, 14.7/14.7 on Bike Sharing and 65.8/62.8 on Wine Quality. The default is close to the independently selected strength in these comparisons.
\par\medskip\noindent{\bfseries Real-data size-grid diagnostics}\par\nobreak\smallskip

All 100 evaluation subsamples are retained at each size. A failed fit receives infinite extended computational loss, with successful-fit MSE separately saved. The supervised baseline must have full column rank in linear regression; nonlinear fits undergo the residual, gradient and objective checks described in Supplementary Section~\ref{app:literature-protocol}. If the baseline fails, its normalized ratio is unavailable. Main-table regimes begin at $n=25$ for California Housing and Wine Quality, $n=300$ for Bike Sharing and $n=75$ for ACS Income; the full grid below records the smaller sizes and all comparator failures.

\begin{table}[H]\centering\small
\caption{Failed fits out of 100 evaluation subsamples. Counts are method specific; no failed trial is replaced by zero loss.}
\label{tab:real-failures}
\begin{adjustbox}{max width=\linewidth}\begin{tabular}{rrrrrrr}\toprule
Task&$n$&Sup.&PPI++&PDC&CC&RePPI\\\midrule
ACS&50&0&0&0&1&94\\
ACS&100&0&0&0&0&21\\
ACS&200&0&0&0&0&0\\
Bike&50&25&25&25&25&95\\
Bike&100&7&7&7&7&75\\
Bike&200&1&1&1&1&34\\
\bottomrule\end{tabular}\end{adjustbox}
\end{table}

\begin{table}[H]\centering\small
\caption{Smallest tested labeled size from which all larger tested sizes have zero observed failures and MSE ratio at most 200 percent. This is a finite-grid diagnostic, not a population sample-size guarantee.}
\label{tab:real-stability}
\begin{adjustbox}{max width=\linewidth}\begin{tabular}{rrrrr}\toprule
Estimator&California&Bike&ACS&Wine\\\midrule
Supervised&25&300&50&25\\
PPI&25&300&75&25\\
PPI++&25&300&50&25\\
PSPA&100&300&50&100\\
PDC&1000&300&50&75\\
CC-PPI&25&300&75&25\\
RePPI&2000&1000&200&500\\
\bottomrule\end{tabular}\end{adjustbox}
\end{table}

\begin{figure}[H]\centering
\includegraphics[width=\textwidth]{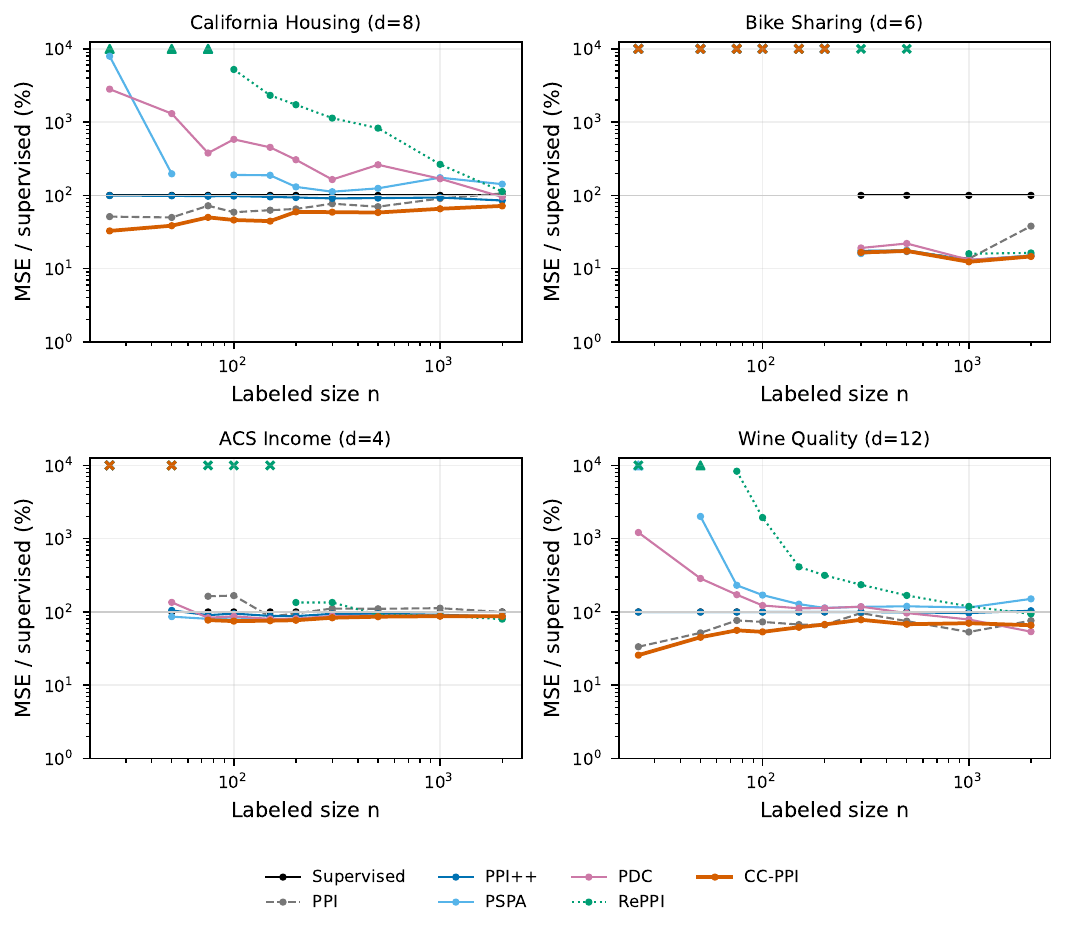}
\caption{Extended comparison: MSE ratio to the full-trial supervised estimate versus labeled size, on logarithmic axes. California uses $d=8$. The horizontal reference is 100. Upward triangles mark finite ratios at least $10^4$; crosses mark failed fits or unavailable denominators. Method-specific failure counts are reported in Table~\ref{tab:real-failures}.}
\label{fig:real-data-extended}
\end{figure}

\begin{table}[H]\centering\small
\caption{California Housing: MSE relative to supervised (percent).}
\label{tab:real-california-extended}
\begin{adjustbox}{max width=\linewidth}\begin{tabular}{rrrrrrrrr}\toprule
$n$&Sup.&PPI&PPI++&PSPA&PDC&RePPI&CC&CC$^{\rm ref}$\\\midrule
25&100.0&51.3&99.2&$\ge10^3$&$\ge10^3$&$\ge10^3$&32.7&31.6\\
50&100.0&49.9&98.7&197.1&$\ge10^3$&$\ge10^3$&38.6&37.1\\
75&100.0&72.1&97.3&$\ge10^3$&378.1&$\ge10^3$&50.1&49.0\\
100&100.0&58.9&97.8&190.3&581.6&$\ge10^3$&46.1&44.2\\
150&100.0&62.8&95.3&187.9&452.5&$\ge10^3$&44.5&43.3\\
200&100.0&65.2&93.5&130.6&306.5&$\ge10^3$&59.5&58.6\\
300&100.0&76.9&91.0&112.2&164.4&$\ge10^3$&59.0&58.2\\
500&100.0&70.1&91.7&124.8&261.6&827.3&58.3&58.6\\
1000&100.0&90.6&93.4&174.5&168.6&265.2&65.6&64.6\\
2000&100.0&106.3&85.1&142.3&95.1&112.8&71.8&71.7\\
\bottomrule\end{tabular}\end{adjustbox}
\end{table}

\begin{table}[H]\centering\small
\caption{Bike Sharing: MSE relative to supervised (percent).}
\label{tab:real-bike-extended}
\begin{adjustbox}{max width=\linewidth}\begin{tabular}{rrrrrrrrr}\toprule
$n$&Sup.&PPI&PPI++&PSPA&PDC&RePPI&CC&CC$^{\rm ref}$\\\midrule
25&$\dagger$&$\dagger$&$\dagger$&$\dagger$&$\dagger$&$\dagger$&$\dagger$&$\dagger$\\
50&$\dagger$&$\dagger$&$\dagger$&$\dagger$&$\dagger$&$\dagger$&$\dagger$&$\dagger$\\
75&$\dagger$&$\dagger$&$\dagger$&$\dagger$&$\dagger$&$\dagger$&$\dagger$&$\dagger$\\
100&$\dagger$&$\dagger$&$\dagger$&$\dagger$&$\dagger$&$\dagger$&$\dagger$&$\dagger$\\
150&$\dagger$&$\dagger$&$\dagger$&$\dagger$&$\dagger$&$\dagger$&$\dagger$&$\dagger$\\
200&$\dagger$&$\dagger$&$\dagger$&$\dagger$&$\dagger$&$\dagger$&$\dagger$&$\dagger$\\
300&100.0&16.2&17.5&16.0&19.1&$\dagger$&16.6&16.5\\
500&100.0&17.1&17.8&17.8&22.0&$\dagger$&17.4&17.3\\
1000&100.0&13.7&13.1&12.9&13.2&16.0&12.4&12.3\\
2000&100.0&38.0&15.0&15.2&14.8&16.4&14.7&14.7\\
\bottomrule\end{tabular}\end{adjustbox}
\end{table}

\begin{table}[H]\centering\small
\caption{ACS Income: MSE relative to supervised (percent).}
\label{tab:real-acs-extended}
\begin{adjustbox}{max width=\linewidth}\begin{tabular}{rrrrrrrrr}\toprule
$n$&Sup.&PPI&PPI++&PSPA&PDC&RePPI&CC&CC$^{\rm ref}$\\\midrule
25&$\dagger$&$\dagger$&$\dagger$&$\dagger$&$\dagger$&$\dagger$&$\dagger$&$\dagger$\\
50&100.0&$\dagger$&104.9&86.2&135.1&$\dagger$&$\dagger$&$\ge10^3$\\
75&100.0&163.5&90.4&80.4&83.2&$\dagger$&77.3&76.1\\
100&100.0&167.1&94.9&84.4&87.2&$\dagger$&74.9&71.4\\
150&100.0&86.0&88.1&79.3&80.3&$\dagger$&75.9&74.1\\
200&100.0&95.6&87.6&80.2&80.8&134.5&77.0&76.0\\
300&100.0&111.2&94.1&86.9&88.1&134.6&83.3&81.8\\
500&100.0&110.3&94.3&89.0&88.8&91.6&86.0&85.0\\
1000&100.0&112.5&91.4&88.8&88.9&89.0&87.5&87.4\\
2000&100.0&99.6&86.7&85.8&85.9&79.7&87.2&87.8\\
\bottomrule\end{tabular}\end{adjustbox}
\end{table}

\begin{table}[H]\centering\small
\caption{Wine Quality: MSE relative to supervised (percent).}
\label{tab:real-wine-extended}
\begin{adjustbox}{max width=\linewidth}\begin{tabular}{rrrrrrrrr}\toprule
$n$&Sup.&PPI&PPI++&PSPA&PDC&RePPI&CC&CC$^{\rm ref}$\\\midrule
25&100.0&33.5&99.6&$\ge10^3$&$\ge10^3$&$\dagger$&25.7&24.5\\
50&100.0&51.7&99.3&$\ge10^3$&286.4&$\ge10^3$&45.1&43.4\\
75&100.0&76.5&99.5&230.9&171.8&$\ge10^3$&56.1&53.7\\
100&100.0&73.1&99.6&169.8&122.4&$\ge10^3$&53.4&51.8\\
150&100.0&67.2&99.8&127.7&112.1&412.1&61.7&59.5\\
200&100.0&66.6&99.9&113.3&112.8&315.8&67.4&64.6\\
300&100.0&96.5&99.6&117.3&118.0&235.1&78.1&75.4\\
500&100.0&74.9&98.5&119.3&96.5&167.8&67.9&65.9\\
1000&100.0&53.2&96.5&115.0&78.8&119.0&70.2&69.3\\
2000&100.0&76.0&103.7&150.7&53.8&94.0&65.8&62.8\\
\bottomrule\end{tabular}\end{adjustbox}
\end{table}
\par\medskip\noindent\textbf{Original-target sensitivity.} The eight-parameter specification was adopted after inspecting the occupancy-direction sensitivity. The original target is retained below; its supervised denominator and coefficient vector differ from the reduced model.\par
\begin{table}[H]\centering\small
\caption{Original nine-parameter California target: retained sensitivity comparison.}
\label{tab:california-original9}
\begin{adjustbox}{max width=\linewidth}\begin{tabular}{rrrrrrrrr}\toprule
$n$&Sup.&PPI&PPI++&PSPA&PDC&RePPI&CC&CC$^{\rm ref}$\\\midrule
25&100.0&24.9&100.0&264.7&$\ge10^3$&$\dagger$&20.8&19.4\\
50&100.0&16.8&99.9&188.7&513.7&$\ge10^3$&31.4&28.8\\
75&100.0&16.4&99.9&114.1&139.2&$\ge10^3$&33.6&31.3\\
100&100.0&13.6&99.8&100.6&130.5&$\ge10^3$&44.8&41.5\\
150&100.0&10.8&99.7&95.1&78.5&473.2&55.0&51.9\\
200&100.0&6.7&99.6&97.5&58.6&182.8&64.3&61.5\\
300&100.0&6.2&98.6&92.4&39.7&97.1&64.6&61.9\\
500&100.0&5.0&96.7&87.4&41.8&115.1&68.9&67.4\\
1000&100.0&5.2&89.5&81.8&25.8&59.4&68.6&66.1\\
2000&100.0&7.8&89.4&83.7&56.4&36.9&69.4&65.3\\
\bottomrule\end{tabular}\end{adjustbox}
\end{table}

\subsection{Interval results}
\noindent\textbf{Interval target convention.} All California interval rows in this supplement retain the original nine-parameter target. They are separate sensitivity diagnostics; no eight-parameter California coverage experiment has been run.\par
The interval experiments use the same evaluation frames and fixed predictors as the point-estimation comparisons above.

California, Wine and ACS each use five labeled sample sizes, with two seeds and 50 subsamples per seed; Bike uses four sizes under the same replication scheme. There are 19 size configurations and 1,900 subsamples in total. Table~\ref{tab:interval-realdata} pools 100 subsamples per row and reports coordinate-average reference inclusion, not population or simultaneous coverage. The supervised reference interval is OLS HC3 normal for the three linear datasets and sandwich Wald for logistic ACS. The ACS construction is an experimental logistic extension, outside the linear-projection validity theorem.

\begin{table}[H]
\caption{Real-data reference inclusion (\%) / mean interval length at nominal 95\%. A dagger indicates failed internal fits; the displayed means use defined intervals.}
\label{tab:interval-realdata}
\begin{center}\small
\begin{adjustbox}{max width=\linewidth}\begin{tabular}{lrlll}
\hline
Dataset & $n$ & CC Wald & CC new & Supervised \\
\hline
California (9D) & 25 & 75.00 / 1.6551 & 93.89 / 2.9353 & 96.33 / 5.9468 \\
California (9D) & 50 & 81.11 / 1.3259 & 94.11 / 2.3920 & 92.00 / 3.1106 \\
California (9D) & 100 & 81.67 / 1.1070 & 93.00 / 1.6754 & 88.89 / 2.0824 \\
California (9D) & 500 & 83.89 / 0.6785 & 95.56 / 1.2406 & 89.11 / 1.2501 \\
California (9D) & 2000 & 88.78 / 0.4733 & 98.00 / 0.8609 & 90.22 / 0.8811 \\
Wine & 25 & 72.00 / 0.7490 & 93.50 / 1.4976 & 97.50 / 3.6304 \\
Wine & 50 & 82.75 / 0.6158 & 92.00 / 1.0387 & 96.58 / 1.5882 \\
Wine & 100 & 87.08 / 0.5051 & 93.00 / 0.8477 & 95.17 / 0.9224 \\
Wine & 500 & 93.08 / 0.2972 & 96.67 / 0.3864 & 95.67 / 0.3779 \\
Wine & 2000 & 94.08 / 0.1752 & 99.00 / 0.2550 & 92.58 / 0.2346 \\
Bike & 300 & 97.83 / 20.1112 & 99.50 / 23.3745 & 97.33 / 47.1276 \\
Bike & 500 & 98.00 / 16.9732 & 98.83 / 19.7080 & 96.50 / 35.9363 \\
Bike & 1000 & 99.83 / 14.1379 & 99.67 / 16.8031 & 97.83 / 25.1675 \\
Bike & 2000 & 100.00 / 12.3722 & 100.00 / 15.0417 & 99.00 / 17.6698 \\
ACS & 75 & 90.00 / 1.0095 & 93.16 / 1.1418$^{\dagger}$ & 94.75 / 1.2515 \\
ACS & 100 & 91.50 / 0.8766 & 94.70 / 1.0185$^{\dagger}$ & 94.25 / 1.0676 \\
ACS & 200 & 93.50 / 0.6241 & 94.50 / 0.7930 & 94.50 / 0.7298 \\
ACS & 500 & 92.25 / 0.4057 & 96.00 / 0.6126 & 92.25 / 0.4509 \\
ACS & 2000 & 95.25 / 0.2042 & 96.50 / 0.4616 & 96.25 / 0.2230 \\
\hline
\end{tabular}\end{adjustbox}
\end{center}

\end{table}

At $n=25$, inclusion increases from 75.00\% to 93.89\% for California and from 72.00\% to 93.50\% for Wine. Expansion also increases width where Wald inclusion is already high. For ACS at $n=2000$, Wald gives 95.25\% inclusion and length 0.2042; the proposed interval gives 96.50\% and 0.4616. Bike exhibits the same width cost.

ACS has 95/100 defined proposed intervals at $n=75$ and 99/100 at $n=100$; the remaining trials fail internally. All other displayed proposed intervals are defined. Appendix~\ref{app:completed-results} reports both successful-interval inclusion and all-trial inclusion with undefined intervals counted as non-inclusion. Original PPI score-region projections and the distinct source/manuscript RePPI conventions remain separately identified in the supplementary numerical records.

The proposed interval is not uniformly preferable to other methods. For California at $n=2000$, original PPI Algorithm~4 gives 93.78\% inclusion and length 0.4875, compared with 98.00\% and 0.8609 for the proposed interval. For Wine at $n=100$, OLS $t$ gives 93.50\% and 0.8201, compared with 93.00\% and 0.8477. These are descriptive comparisons with Monte Carlo uncertainty. Overall, the results support more stable coverage in several severe-undercoverage settings while retaining adverse cases and explicit width costs.

\subsection{Real-data reliability denominators}
There are 100 outer subsamples in each size cell. ``Inclusion'' and ``Length'' below use defined new CC intervals. ``All trials'' assigns zero inclusion to an undefined interval and is a descriptive reliability measure. Rates are percentages, averaged over coordinates, against the complete evaluation-frame target. An average does not certify each coordinate. Six retained ACS trials fail internally; their omission from a successful-interval mean must not be mistaken for successful inference on every trial.
\begingroup\small
\begin{longtable}{lrrrrr}
\hline
Dataset & $n$ & Defined & Inclusion & All trials & Length \\
\hline\endfirsthead
\hline
Dataset & $n$ & Defined & Inclusion & All trials & Length \\
\hline\endhead
California (9D) & 25 & 100/100 & 93.89 & 93.89 & 2.9353 \\
California (9D) & 50 & 100/100 & 94.11 & 94.11 & 2.3920 \\
California (9D) & 100 & 100/100 & 93.00 & 93.00 & 1.6754 \\
California (9D) & 500 & 100/100 & 95.56 & 95.56 & 1.2406 \\
California (9D) & 2000 & 100/100 & 98.00 & 98.00 & 0.8609 \\
Wine & 25 & 100/100 & 93.50 & 93.50 & 1.4976 \\
Wine & 50 & 100/100 & 92.00 & 92.00 & 1.0387 \\
Wine & 100 & 100/100 & 93.00 & 93.00 & 0.8477 \\
Wine & 500 & 100/100 & 96.67 & 96.67 & 0.3864 \\
Wine & 2000 & 100/100 & 99.00 & 99.00 & 0.2550 \\
Bike & 300 & 100/100 & 99.50 & 99.50 & 23.3745 \\
Bike & 500 & 100/100 & 98.83 & 98.83 & 19.7080 \\
Bike & 1000 & 100/100 & 99.67 & 99.67 & 16.8031 \\
Bike & 2000 & 100/100 & 100.00 & 100.00 & 15.0417 \\
ACS & 75 & 95/100 & 93.16 & 88.50 & 1.1418 \\
ACS & 100 & 99/100 & 94.70 & 93.75 & 1.0185 \\
ACS & 200 & 100/100 & 94.50 & 94.50 & 0.7930 \\
ACS & 500 & 100/100 & 96.00 & 96.00 & 0.6126 \\
ACS & 2000 & 100/100 & 96.50 & 96.50 & 0.4616 \\
\hline\end{longtable}
\endgroup

\subsection{Selected interval comparisons}
\begin{table}[htbp]\centering
\caption{Selected reference inclusion (\%) / mean length at nominal 95\%. These are coordinate-average finite-frame inclusion rates.}\label{tab:new-real-interval}
\begingroup\small
\setlength{\tabcolsep}{3pt}
\begin{adjustbox}{max width=\linewidth}\begin{tabular}{lrlllll}
\hline
Data & $n$ & CC Wald & Expanded & Supervised & PPI++ & PDC \\
\hline
California (9D) & 50 & 81.11/1.326 & 94.11/2.392 & 92.00/3.111 & 67.11/1.419 & 47.67/1.717 \\
California (9D) & 500 & 83.89/0.678 & 95.56/1.241 & 89.11/1.250 & 84.22/0.767 & 73.89/0.621 \\
California (9D) & 2000 & 88.78/0.473 & 98.00/0.861 & 90.22/0.881 & 87.67/0.519 & 83.78/0.396 \\
Wine & 50 & 82.75/0.616 & 92.00/1.039 & 96.58/1.588 & 72.17/0.704 & 68.42/0.927 \\
Wine & 500 & 93.08/0.297 & 96.67/0.386 & 95.67/0.378 & 87.00/0.296 & 91.83/0.327 \\
Wine & 2000 & 94.08/0.175 & 99.00/0.255 & 92.58/0.235 & 85.67/0.191 & 96.25/0.168 \\
\hline
\end{tabular}\end{adjustbox}
\endgroup

\end{table}

\end{document}